\documentclass[sigconf,nonacm]{acmart}
\AtBeginDocument{%
  }

\copyrightyear{2026}
\acmYear{2026}
\setcopyright{cc}
\setcctype{by}
\acmConference[CIKM '26]{Proceedings of the 35th ACM International Conference on Information and Knowledge Management}{November 07--11, 2026}{Rome, Italy}
\acmBooktitle{Proceedings of the 35th ACM International Conference on Information and Knowledge Management (CIKM '26), November 07--11, 2026, Rome, Italy}
\acmDOI{10.1145/3799682.3841030}
\acmISBN{979-8-4007-2539-5/2026/11}

\usepackage{booktabs}
\usepackage{apxproof}
\usepackage{multirow}
\usepackage{xcolor-solarized}
\usepackage{cleveref}
\usepackage{paralist}
\usepackage{stmaryrd}
\usepackage{todonotes}
\usepackage[clock]{ifsym}
\usepackage[propagate-math-font=true]{siunitx}
\usepackage{multirow}
\usepackage{mathtools}
\usepackage[ruled,linesnumbered,noend]{algorithm2e}
\usepackage{tikz}
\usetikzlibrary{automata, positioning, arrows.meta}

\usepackage{pgfplots}
\usepackage{pgfplotstable}
\pgfplotsset{compat=1.18}
\usepackage{subcaption} % subfigures

\newtheoremrep{theorem}{Theorem}
\newtheoremrep{lemma}{Lemma}
\newtheoremrep{proposition}{Proposition}
\newtheoremrep{fact}{Fact}
\newtheoremrep{claim}{Claim}

\crefname{algocf}{alg.}{algs.}
\Crefname{algocf}{Algorithm}{Algorithms}

\crefname{definition}{definition}{definitions}
\Crefname{definition}{Definition}{Definitions}

\crefname{example}{example}{examples}
\Crefname{example}{Example}{Examples}

\newcommand{\mo}[1]{{#1}}
\newcommand{\bl}[1]{{#1}}
\newcommand{\co}[1]{{#1}}
\newcommand{\na}[1]{{#1}}

\newcommand{\CR}[1]{{#1}}
\newcommand{\algoName}{ARCO\xspace}
\newcommand{\algoNameLong}{Automata-based Rewriting of CN2RPQs under Ontologies\xspace}

\newcommand{\nop}[1]{}

\newcommand{\nesting}[1]{\langle #1\rangle}            % tuple

\newcommand{\A}{\mathcal{A}}

\newcommand{\I}{\mathcal{I}}
\newcommand{\J}{\mathcal{J}}

\newcommand{\G}{\mathcal{G}}

\newcommand{\T}{\mathcal{T}}
\newcommand{\node}[1]{[#1]}
\newcommand{\transclosure}[2]{#1 \sqsubseteq^*_\T #2}

\newcommand{\ISA}{\sqsubseteq}
\newcommand{\AND}{\sqcap}

\newcommand{\propnames}{\mathbf{K}}
\newcommand{\rolenames}{\mathbf{R}}
\newcommand{\conceptnames}{\mathbf{C}}
\newcommand{\indivs}{\mathbf{N}}
\newcommand{\allroles}{\overline{\rolenames}}
\newcommand{\datatests}{\mathbf{T^D}}
\newcommand{\nNFA}[1]{\alpha_{#1}}

\newcommand{\nnfa}{n-NFA}

\newcommand{\roleMath}[1]{{\color{solarized-violet}\mathit{#1}}}
\newcommand{\role}[1]{{\color{solarized-violet}\textit{#1}}}
\newcommand{\concept}[1]{{\color{solarized-green}\mathsf{#1}}}
\newcommand{\property}[1]{{\color{solarized-cyan}\textnormal{\textsc{#1}}}}

\newcommand{\ontoLang}{\textit{DL}\text{-}\textit{Lite}_\mathit{PG}}

\newcommand{\0}{\phantom{0}}
\newcommand{\skipping}[1]{\mathsf{skipCore}^\T(#1)}
\newcommand{\skippingAnon}[1]{\mathsf{skipAnon}^\T(#1)}
\newcommand{\algoMath}[1]{\mathsf{\MakeLowercase{\algoName}}^\T(#1)}
\newcommand{\nexist}{\langle\exists\rangle}

\begin{document}

%%
%% The "title" command has an optional parameter,
%% allowing the author to define a "short title" to be used in page headers.
\title{Rewriting Ontology-Mediated Property Graph Queries into GQL}

%%
%% The "author" command and its associated commands are used to define
%% the authors and their affiliations.
%% Of note is the shared affiliation of the first two authors, and the
%% "authornote" and "authornotemark" commands
%% used to denote shared contribution to the research.
\author{Bianca Löhnert}
% \authornote{Both authors contributed equally to this research.}
\email{bianca.loehnert@plus.ac.at}
\orcid{0009-0002-0297-5971}
% \author{G.K.M. Tobin}
% \authornotemark[1]
% \email{webmaster@marysville-ohio.com}
\affiliation{%
  \institution{University of Salzburg}
  %\city{Salzburg}
  \country{Austria}
}

\author{Nikolaus Augsten}
\email{nikolaus.augsten@plus.ac.at}
\orcid{0000-0002-3036-6201}
\affiliation{%
  \institution{University of Salzburg }
  %\city{Salzburg}
  \country{Austria}}

\author{Cem Okulmus}
\email{cem.okulmus@upb.de}
\orcid{0000-0002-7742-0439}
\affiliation{%
  \institution{Paderborn University }
  %\city{Paderborn}
  \country{Germany}
}

\author{Magdalena Ortiz}
\email{magdalena.ortiz@tuwien.ac.at}
\orcid{0000-0002-2344-9658}
\affiliation{%
 \institution{TU Wien }
 %\city{Vienna}
 \country{Austria}
}

%%
%% By default, the full list of authors will be used in the page
%% headers. Often, this list is too long, and will overlap
%% other information printed in the page headers. This command allows
%% the author to define a more concise list
%% of authors' names for this purpose.
% \renewcommand{\shortauthors}{Loehnert et al.}
\renewcommand{\shortauthors}{Bianca Löhnert, Nikolaus Augsten, Cem Okulmus, and Magdalena Ortiz}
%% No italics, no superscripts, not anonymous
%% Use footnote or author note to identify equal contribution, shared contribution, and/or contact author info

%%
%% The abstract is a short summary of the work to be presented in the
%% article.
\begin{abstract}
Ontology-based data access is intended for graph data, but practical support remains limited to SQL-like languages lacking the navigational and path-matching features fundamental to graph querying. Theoretical algorithms for navigational queries have long been available. Still, they have never been implemented, largely because practical graph query languages fell short of their theoretical counterparts and lacked the expressive power needed to support rewriting common ontology languages.

The recent standardisation efforts around GQL and SQL/PGQ finally allow us to overcome this barrier and present a practical query rewriting technique for ontology-mediated navigational graph queries. Ontologies are written in a DL-Lite variant tuned to property graphs, and the queries in a GQL fragment with  nested two-way regular path queries and which can be evaluated in Cypher. Preliminary experiments with our proof-of-concept prototype suggest that querying graph data with ontological knowledge may finally be within reach. 
\end{abstract}

%%
%% The code below is generated by the tool at http://dl.acm.org/ccs.cfm.
%% Please copy and paste the code instead of the example below.
%%


%%
%% Keywords. The author(s) should pick words that accurately describe
%% the work being presented. Separate the keywords with commas.

% \received{23 May 2026}
% \received[revised]{7 August 2026}
% \received[accepted]{20 August 2026}

%%
%% This command processes the author and affiliation and title
%% information and builds the first part of the formatted document.
\maketitle

\section{Introduction}

A core aim of the \emph{ontology-based data access} (OBDA) framework, in addition to the virtual integration of data sources, is to extend existing data with domain-specific knowledge without the need to materialise all implied facts; the framework still enables access to all consequences when answering user queries.
The key question at the core of OBDA is the following: how to capture all the semantic consequences of a query in the presence of an ontology, where the query may be specified over a possibly \emph{extended} signature with terms from both the data and the ontology. Such a pair of query and ontology is often called an \emph{ontology-mediated query (OMQ)} and the task of producing all consequences is called OMQ answering (OMQA).
% given domain-specific knowledge in the form of an ontology, and a query specified over a possibly \emph{extended} signature, with terms from both the data and the ontology---such a pair is often called an \emph{ontology-mediated query (OMQ)}---the goal is to evaluate the query capturing all semantic consequences that follow in the presence of the ontology.
The main technique for OMQA is called \emph{query rewriting}: the query is rewritten into a new query that uses only the \emph{reduced} signature present in the data, incorporating the knowledge in the ontology into the query itself.   The rewritten query can then be evaluated using a conventional database system that does not support ontologies, while still making effective use of the ontology to retrieve implicit answers.

% One of the key techniques to realise OBDA is \emph{query rewriting}: given domain-specific knowledge in the form of an ontology, and a query specified over a possibly \emph{extended} signature (with terms from both the data and the ontology), the query is rewritten into a new query that uses only the \emph{reduced} signature present in the data, while still capturing all semantic consequences that follow in the presence of the ontology.
% The rewritten query can be evaluated using a conventional database system that does not support ontologies, while still making effective use of the ontology.
% This key techniques underlying many OBDA systems is known as \emph{ontology-mediated query answering} (OMQA).

While commercial OBDA systems already exist for relational databases~\cite{DBLP:conf/ijcai/XiaoCKLPRZ18}, corresponding solutions for graph databases---the focus of this paper---are still missing. Instead of the relational model, graph databases focus on the labelled property graph (LPG) model, which allows for nodes and binary, directed edges between nodes. Both nodes and edges can be assigned multiple \emph{labels} to classify them, and they can be associated with data values via so-called \emph{properties}.
%The OBDA framework is already seeing a number of commercial applications, albeit only in the setting of relational databases. 
%While it is unlikely that relational databases are being replaced anytime soon, we none-the-less see new forms of storing and querying data. Graph databases are one such example, where we already have a large number of commercial systems. Instead of the relational model that underpins relational database systems, these instead focus on the labelled property graph (LPG) model. It allows for nodes and binary  edges over nodes, each of which can be mapped to a number of labels and associated with data values via so-called \emph{properties}.
%
The increasing popularity of graph databases has recently culminated in the ISO standardisation of two new graph query languages: GQL and SQL/PGQ.
The key distinguishing feature of these query languages is their \emph{navigational} capability, which allows them to match arbitrarily long paths in the data.
Both GQL and SQL/PGQ share the same navigational core \cite{Deutsch2022}, which includes the standard navigational query languages: (two-way) regular path queries ((2)RPQs), nested 2RPQs (N2RPQs), and their conjunctive variants (C2RPQs resp.\ CN2RPQs).

The extensive study of navigational queries in the context of OMQA over the past decade has led to tight complexity results for most description logics, ranging from lightweight to highly expressive, as well as for all standard navigational query languages (2RPQs, N2RPQs, C2RPQs, and CN2RPQs) \cite{Bienvenu2015a,Bienvenu2014,DBLP:conf/ijcai/CalvaneseEO09, Ortiz2011, DBLP:conf/kr/Ostropolski-Nalewaja24, DBLP:journals/jair/StefanoniMKR14, DBLP:journals/jair/DimartinoWCP25}. 
% \todo{add query answering in horn fragments with sebastian rudolph, ijcai11, iirc. So far only self  cites. 
% There is at leat one by sebastian on undecidability of navigation queries in the finite, and there should be papers from oxford for EL variants; giorgo stefanoni afair, and maybe hector perez urbina}
%The study of navigational queries in the context of OMQA has developed is far from novel. For more than a decade we have had tight complexity results for most description logics, ranging from lightweight to highly expressive, as well as for all standard navigational languages (2RPQs, N2RPQs, C2RPQs, and CN2RPQs), e.g.,   \cite{Bienvenu2015a,Bienvenu2014,DBLP:conf/ijcai/CalvaneseEO09}.
 However, these studies focused on the boundaries of decidability and computational complexity. The proposed algorithms are so unamenable % not a typo, don't "fix"
to implementation that, more than a decade later, not a single practical implementation has been proposed.
\mo{We differ from these early \co{theoretical} works, by looking to develop \co{practical} techniques for pure rewritings of navigational queries (i.e., the ontology reasoning is compiled entirely into the rewritten query, which can be evaluated with standard query engines without any data preprocessing) that can pave the way towards practical implementations.}

There were two significant roadblocks that discouraged the development of practical techniques for OMQA over graph databases.
First, practical languages imposed ad-hoc restrictions on navigational languages that rendered them unsuitable for this setting, for example, by providing only partial support for RPQs and by failing to support the basic homomorphism semantics that is natural in the OMQA context.
%First, practical languages imposed ad-hoc restrictions on navigational languages that made them inadequate. Cypher, the query language of the most widely used graph database system Neo4j, did not support RPQs in full, and did not enable the basic homomorphism semantics that is natural in the OMQA context. 
Second, it has been shown that even plain RPQs in the presence of the simplest DL-Lite ontologies cannot be rewritten into C2RPQs \cite{DBLP:conf/esws/LöhnertAOO25}. Instead, such rewritings require \emph{nested} RPQs, which existing technologies did not support.
The arrival of GQL and SQL/PGQ has removed both obstacles at once: the standards include CN2RPQs and basic homomorphism semantics \cite{Francis2022,Francis2023}, and practical graph query languages are now rapidly being updated to support these features \cite{neo4j-gql-conformance}. This has finally made practical OMQA over graph databases feasible.

%\todo[inline]{\Cref{fig:overview_rewriting_algo} provides an overview of our rewriting algorithm, which takes a join-on-free conjunctive nested two-way regular path query as input. The first step $remove\exists Var$ removes all the unbound variables and rewrites each atom $\alpha(x,y)$ in the query into a N2RPQ $\nesting{\alpha}(x)$. This is only possible because of the join-on-free restriction. After this step, we can rewrite each of the resulting N2RPQ atoms using $skipAnonymous$ and $skipCore$. Both functions require an N2RPQ and a $\ontoLang$ ontology as input. The final output is again a join-on-free CN2RPQ that can be translated into GQL syntax and evaluated by an arbitrary GQL database system.}

\begin{figure}[h]
    \centering
    \includegraphics[width=\linewidth]{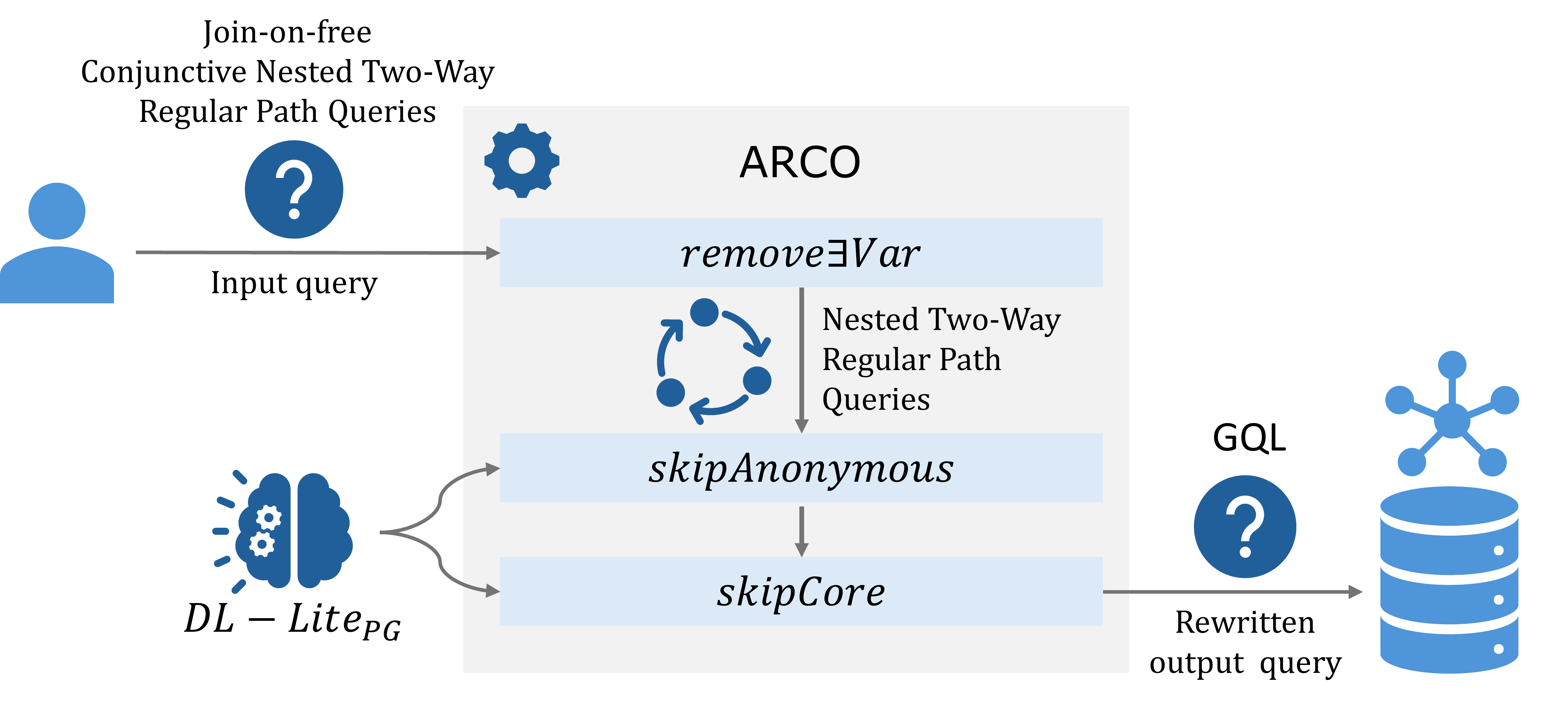}
    \caption{High-level overview of the rewriting algorithm.}
    \label{fig:overview_rewriting_algo}
    \Description{An abstract diagram that explains the structure of the rewriting procedure we present in this paper.}
\end{figure}

Seizing this opportunity, we propose the first practical query rewriting technique for a navigational query language that contains RPQs. 
%% the so-called \emph{join-on-free} CN2RPQs 
in the presence of DL-Lite ontologies. \Cref{fig:overview_rewriting_algo} provides an overview of the overall architecture, which comprises a DL-Lite ontology expressed, for example, in the Web Ontology Language (OWL)~\cite{owl2}; our rewriting algorithm \co{``\algoNameLong'' or simply \algoName for short}; and a property graph database with support for GQL, such as Neo4j,\footnote{\url{https://neo4j.com/}} Google Spanner,\footnote{\url{https://cloud.google.com/spanner}} or Oracle Database (via SQL/PGQ).\footnote{\url{https://www.oracle.com/database/}} The user poses a query over the extended signature, i.e., the query may use terms from both the ontology and the underlying database. Our rewriting algorithm compiles the user query into a semantically equivalent rewritten query that incorporates the consequences of ontology reasoning directly into the query expression. The rewritten query is subsequently evaluated by the graph database system. Although the database remains agnostic to the ontology, the answer to the rewritten query is  complete with respect to the ontology semantics.

% \todo[inline]{The three steps of the rewriting algorithm in \Cref{fig:overview_rewriting_algo} are not addressed so far.}
% \todo[inline]{Algorithm name in the figure needs to be updated!}

In terms of query language, our approach fully covers 2RPQs and N2RPQs in the presence of DL-Lite ontologies. We also consider conjunctive N2RPQs; to ensure the algorithm remains practical and useful, we impose a restriction called \emph{join-on-free}, whereby non-answer variables cannot be shared by multiple atoms. \CR{While techniques for rewriting full C2RPQs are known~\cite{Eiter2012}, this fragment serves as a useful starting point for our endeavour.}
We also support property data values, a central feature of the labelled property graph model that is often disregarded in the OMQA literature.
Property tests are also allowed along navigational paths.

On the ontology side, we define a variant of DL-Lite that can use property value tests on the left-hand-side of axioms, which enables the creation of concepts and roles based on these values.

In summary, we make the following contributions: 

    \begin{enumerate}
        \item \co{We present the first practical algorithm to rewrite navigational ontology-mediated queries with unrestricted regular path expressions, and with DL-Lite ontologies, into graph queries that can be evaluated with off-the-shelf  query engines for graph data adopting the property graph data model. 
        The queries can access the data values of properties,  also while navigating paths.}
        %%, over a DL-Lite ontology  into CN2RPQs join-on-free CN2RPQs over an.}
           % in the output \todo{what do you mean by ïn the output?} 
  \item \co{We leverage the power of graph queries with \emph{nested path navigation} to support the interaction between the regular path expressions in the given queries on the one hand, and the ontology axioms on the other.
  %%: the latter may use existential quantification to imply the existence of unnamed objects, which may participate in the navigation required by the path expressions.   
  To our knowledge, this is the first algorithm for 
a  navigational query language that fully supports property paths, which produces a   
  \emph{pure rewriting} that can be evaluated directly over the data, with the ontology consequences entirely rewritten into the query.}
    % \item This is the first work on OMQA over the LPG model that makes use of properties by supporting data tests over property values in both the ontology and the query language.
    \item We demonstrate the practical utility of our technique by implementing a prototype. The system takes as input an ontology in OWL format together with navigational queries, rewrites them into the graph query language Cypher, and evaluates them on a database in Neo4j.
    %We show-case the practical utility of the algorithm by providing a proof-of-concept implementation that takes as input an ontology in OWL format and a bespoke navigational query language 
    % real-world query language  
    %Cypher
    %and rewrites into Cypher.

    \end{enumerate}

%% COMMENTED OUT FOR SUBMISSION, SINCE IT BREAKS ANONYMITY 
% \co{This paper is a significantly extended version of a previous workshop paper~\cite{DBLP:conf/dlog/LohnertAOO25}, with extensive changes to the rewriting algorithm to address issues with respect to completeness.}

This paper is structured as follows. In \Cref{sec:preliminaries} we present the key terminology of our setting and also introduce our novel description logic to express tests over property values.
In \Cref{sec:rewrite_cof_N2RPQ} we define procedures to rewrite N2RPEs in the presence of $\ontoLang$ ontologies.
Then, in \Cref{sec:rewrite_cof_CN2RPQ} we use these procedures to present our rewriting algorithm on join-on-free CN2RPQs. In \Cref{sec:implementation} we present our proof-of-concept implementation and an experimental evaluation that aims to show its performance. In \Cref{sec:conclusion} we summarise our results and  highlight future work.

% \begin{itemize}
%     \item Motivation for OBDA in PG setting: great success story in the relational world; we see graph databases are gaining popularity, and there there is a lot of (theoretical work) to show how to rewrite graph queries -> the time has come to put it to use

%     \item need to differentiate us from last years papers: what is still missing and why is it important:

%     \item Contributions:

\paragraph{Related Work.} \footnote{\CR{This work is a corrected version of the preliminary algorithm presented in \cite{DBLP:conf/dlog/LohnertAOO25}.}}
%We note that a few recent works have begun the effort of closing this gap towards the practical OMQA algorithms in the graph database setting. Already 
Early work by Di Martino et al.~\cite{Dimartino2016,DBLP:journals/jair/DimartinoWCP25} leverages the recursion in regular path queries to rewrite a fragment of $\mathcal{EL}$, the lightweight description logic underlying the OWL~2 EL profile. They use navigational queries as target language for query rewriting, but their source languages are only instance queries and conjunctive queries; we support navigational source queries. 
The work is primarily theoretical and not aimed at implementation. It is therefore more closely related to the previously discussed theoretical studies of navigational queries in the context of OMQA~\cite{Bienvenu2015a,Bienvenu2014,DBLP:conf/ijcai/CalvaneseEO09} than to this paper.
%The work is also primarily theoretical in nature and
%not aimed at implementation, and thus more in line with the previously mentioned theoretical studies of %navigational queries in the context of OMQA~\cite{Bienvenu2015a,Bienvenu2014} than with this paper.
Aiming at practical implementation, Dragovic  et al.~\cite{DBLP:conf/dlog/DragovicO023} considered \textit{DL-Lite} ontologies and a restricted fragment of C2RPQs that can be rewritten into unions of C2RPQs. 
%Aiming for practical implementation, Dragovic  et al.~\cite{DBLP:conf/dlog/DragovicO023} considered a restricted fragment of C2RPQs that can be rewritten into {U}C2RPQs, and \textit{DL-Lite} ontologies. 
However, the approach does not cover full RPQs.
%%and it provides only limited support for data tests. 
Like Di Martino et al.~\cite{Dimartino2016}, Löhnert et al.~\cite{DBLP:conf/esws/LöhnertAOO25} support a fragment of $\mathcal{EL}$ as an ontology; their regular path expressions are restricted to ensure rewritability into C2RPQs, i.e., to avoid the need for nested queries.
\co{Query rewriting in the presence of data tests has been explored previously,  such as the work by Savkovi{\'c} and Calvanese~\cite{DBLP:conf/ecai/SavkovicC12}, which considered the problem of OMQ answering over an extension of DL-lite with expressions over concrete domains called datatype expressions, and the ability to express hierarchies over them. Baader et al.~\cite{DBLP:conf/ijcai/BaaderBL17} extended this work by an even more expressive language, where the datatype expressions over the concrete domain can use $n$-ary predicates instead of just unary ones. Our restricted use of data values does not require reasoning over the data domains, and we can simply rely on the query engine to evaluate the data atoms.}
\CR{All proofs can be found in the extended version of this paper~\cite{arxivVersion}.}

%A next step came in the work from~\cite{DBLP:conf/esws/LöhnertAOO25}, where the ontology language is also a fragment of  $\mathcal{EL}$ very similar to the one of Di Martino et al.~\cite{Dimartino2016}. There are also restrictions on the regular path expressions to ensure rewritability into C2RPQs. 
% \end{itemize}

\tikzset{%
    pics/myN/.style n args={4}{code={%  
        % Changed font to \scriptsize and rule width to 1.8cm
        \node (#1) at (0,0) [draw,text=black,font=\scriptsize, align=center,thick, rounded corners,]{ \scriptsize \textnormal{\{}  $\concept{#2}$ \textnormal{\}} \\  \rule[8pt]{1.8cm}{0.5pt} \vspace{-3mm} \\  #3 \\ \rule[8pt]{1.8cm}{0.5pt} \vspace{-3mm} \\  #4 };
    }},
   pics/myE/.style n args={6}{code={% 
       % Changed font to \scriptsize and rule width to 1.8cm
       \draw[->] (#1) to[#6]     node[midway, black, draw=black,fill=white, rounded corners, align=center,thick, font=\scriptsize, label={  \scriptsize } ] { \textnormal{\{} $\role{#3}$  \textnormal{\}} \\  \rule[8pt]{1.8cm}{0.5pt} \vspace{-3mm} \\ #4 \\ \rule[8pt]{1.8cm}{0.5pt} \vspace{-3mm} \\  \scriptsize      \textnormal{ #5 }    }  (#2); 
    }},
}

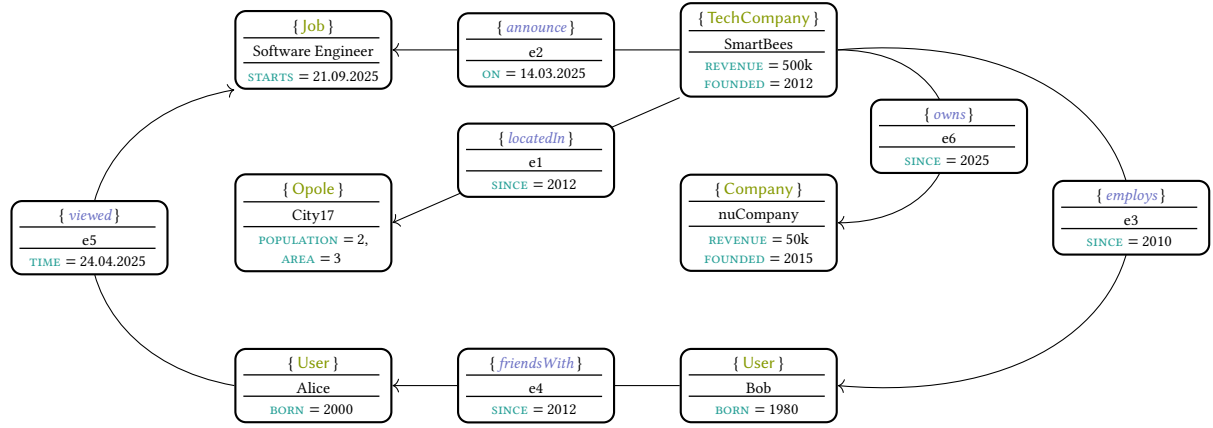
\begin{figure*}
    \centering % Centers the larger figure on the page
    \begin{tikzpicture}
    % Increased row sep from 0.5cm to 1.2cm, and column sep from 3cm to 3.8cm
    \matrix[row sep=1cm, column sep=3.8cm] {  
       \pic {myN={softEng}{Job}{Software Engineer}{$\property{starts} = $ 21.09.2025}}; &
       \pic {myN={smartB}{TechCompany}{SmartBees}{$\property{revenue} = $ 500k \\ $\property{founded}=$ 2012}}; \\
       \pic {myN={opole}{Opole}{City17}{
       $\property{population} = $ 2, \\ 
       $\property{area} = $ 3 }       
       }; &  
       \pic {myN={newComp}{Company}{nuCompany}{$\property{revenue} = $ 50k \\ $\property{founded} =$  2015}}; \\        
       \pic {myN={alice}{User}{Alice}{$\property{born} = $ 2000 }}; &
       \pic {myN={bob}{User}{Bob}{$\property{born} = $ 1980 }};\\        
    };    
    \pic {myE={smartB.south west}{opole.east}{locatedIn}{e1}{$\property{since} = $ 2012 }{}}; 
    \pic {myE={smartB.west}{softEng.east}{announce}{e2}{$\property{on} = $ 14.03.2025 }{}}; 
    \pic {myE={smartB.east}{bob.east}{employs}{e3}{$\property{since} = $ 2010 }{bend left=95, looseness=3}}; 
    \pic {myE={bob.west}{alice.east}{friendsWith}{e4}{$\property{since} = $ 2012 }{}}; 
    \pic {myE={alice.west}{softEng.south west}{viewed}{e5}{$\property{time} = $ 24.04.2025 }{bend left=80, looseness=1.7}}; 
    \pic {myE={smartB.east}{newComp.east}{owns}{e6}{$\property{since} = $ 2025 }{bend  left=90,looseness=2.2}};
    \end{tikzpicture}
    \caption{Property Graph of \Cref{ex:propertyGraph}}
    \label{fig:propertyGraphExample}
        \Description[A visualisation of a property graph.]{The plot shows various angular nodes, connected by edges that in turn have angular nodes in their middle. Each of these angular nodes is a small table, representing either a node or an edge in the property graph, with the table indicating first the labels of the node or the label of the edge, then in the next line the identifier of the object, and lastly a number of equations, were on the left-hand-side we have a property key and on the right-hand-side its respective value.}
\end{figure*}

\section{Ontology Mediated Querying of Property Graphs}
 \label{sec:preliminaries}

\textbf{Data model.}
We are interested in querying data
    %%In this paper the data is 
    given as 
    %%that is, the ABox, takes the form of a 
    finite \emph{property graphs}~\cite{DBLP:conf/amw/Angles18, DBLP:series/synthesis/2018Bonifati, DBLP:conf/sigmod/FrancisGGLLMPRS18}. 
We assume disjoint, countably infinite sets $\conceptnames$, $\rolenames$,  $\propnames$ and $\indivs$  of $\concept{\text{\emph{concept names}}}$,  $\role{\text{\emph{role names}}}$,  $\property{\text{\emph{property} keys}}$ and \emph{individuals}, respectively.  \co{In examples, we colour code the first three, as shown.}
We also assume a \emph{concrete domain} $(\mathbf{D}, \mathbf{P^D})$, with $\mathbf{D}$ a set of \emph{values} and $\mathbf{P^D}$ a set of binary predicates over $\mathbf{D}$. For example, $(\mathbf{D}, \mathbf{P^D})$ could be the integers with the usual $=$, $\leq$, $\geq$ predicates.

\begin{table*}[h]
    \centering
\setlength{\tabcolsep}{1pt}

\caption{Semantics of $\ontoLang$ concepts, roles and data tests given $\mathcal{A} = (N,E,\mathsf{label}, \mathsf{prop})$. \vspace{-4mm}}
    \label{tab:dlpg}
\begin{tabular}{ccc}
    \begin{tabular}[t]{ll}
        \toprule
        $C$ 		& $C^{\mathcal{A}} \subseteq N $  \\
        \midrule
        $\top$		& $N$ \\
        $A$ 	&  $\{ n \in N \mid A \in \mathsf{label}(n)  \}$ \\ 
        $\exists r.\top$ & $\{n \mid r(n,m) \in E \}$ \\
         $\lnot C$ 	& $N \setminus {C}^\mathcal{A} $        
        % property test       & $\delta(p)$     & $ \{ a \mid (a,v) \in p^\mathcal{I} \} \cup  \{ (a,a') \mid ((a,a'),v) \in p^\mathcal{I} \} $ \\
        %conjunction         & $C\sqcap D$   & $C^\I\cap D^\I$ \\
        % concept inclusion 	& $C \ISA D$	& $C^\mathcal{I} \subseteq D^\mathcal{I}$ & 
        %
        
        %
         % $p \odot v?$  & 
        % $\{ a \mid (a,v') \in p^\mathcal{I} \mbox{~for some $v'$ with~}v' \odot v  \} \ \cup $ & set of data tests   & $\{T_1,\dots,T_n\}$ & {$T_1^\I\cap \dots\cap T_n^\I$} 
          \\
        % & & $\{ (a,a') \mid ((a,a'),v') \in p^\mathcal{I} \mbox{~for some $v'$ with~} v' \odot v \}$ &  \\ 
        \bottomrule
    \end{tabular} & 
           \begin{tabular}[t]{ll}
        \toprule
         $F$ & $F^{\mathcal{A}} \subseteq N \times N $ \\
        \midrule
        $r$ &  $\{(n,m) \mid  r(n,m) \in E \}$	  \\
        $r^-$			&  $\{(m,n) \mid  r(n,m) \in E \}$ \\ 
        \\

        % property test       & $\delta(p)$     & $ \{ a \mid (a,v) \in p^\mathcal{I} \} \cup  \{ (a,a') \mid ((a,a'),v) \in p^\mathcal{I} \} $ \\
        %conjunction         & $C\sqcap D$   & $C^\I\cap D^\I$ \\
        % concept inclusion 	& $C \ISA D$	& $C^\mathcal{I} \subseteq D^\mathcal{I}$ & 
        %
        
        %
         % $p \odot v?$  & 
        % $\{ a \mid (a,v') \in p^\mathcal{I} \mbox{~for some $v'$ with~}v' \odot v  \} \ \cup $ & set of data tests   & $\{T_1,\dots,T_n\}$ & {$T_1^\I\cap \dots\cap T_n^\I$} 
          \\
        % & & $\{ (a,a') \mid ((a,a'),v') \in p^\mathcal{I} \mbox{~for some $v'$ with~} v' \odot v \}$ &  \\ 
        \bottomrule
    \end{tabular} &     
           \begin{tabular}[t]{ll}
        \toprule
         $T$ & $T^{\mathcal{A}} \subseteq (N \times N) \cup N $ \\
        \midrule
         $p \odot v?$  & $ \{ n \in N \mid v' \odot v \text{ where } \mathsf{prop}(n,p) = v' \} \ \cup $  \\ 
         & $\{ (n,m) \mid   v' \odot v \text{ where } \exists r \in \rolenames . \big (  \mathsf{prop}(r(n,m),p) = v' \text{ and } r(n,m) \in E \big ) \}$ \\ 
         $\{T_1,\dots,T_n\}$ & $\{T_1^\mathcal{A}\cap \dots\cap T_n^\mathcal{A} \}$ \\
        % $\{ a \mid (a,v') \in p^\mathcal{I} \mbox{~for some $v'$ with~}v' \odot v  \} \ \cup $ & set of data tests   & $\{T_1,\dots,T_n\}$ & {$T_1^\I\cap \dots\cap T_n^\I$} 
          \\
        % & & $\{ (a,a') \mid ((a,a'),v') \in p^\mathcal{I} \mbox{~for some $v'$ with~} v' \odot v \}$ &  \\ 
        \bottomrule
    \end{tabular}
    
\end{tabular}
\end{table*}

%%We introduce the query and ontology languages we focus on, as well as the \emph{property graphs}.  
% $\ontoLang$, that can be rewritten into the query language introduced later in this section. Further, we define knowledge bases that are a pair of a TBox and an ABox. In this work we consider $\ontoLang$ TBoxes and ABoxes in form of property graphs. We do not include qualified inverses, since Calvanese et al. show in \cite{Calvanese2013} that conjunction on the left-hand side of a concept inclusion can be encoded if the ontology language allows inverse roles and existential restriction on both sides of concept inclusions (see Example \ref{ex:hiddenConjunction}). Therefore, by adding qualified inverses to the language of $\ontoLang$ the data complexity increases from \NL to \PTime and can not be fully rewritten into existing graph query languages anymore. 

	\begin{definition}[Property Graphs]
		A \emph{property graph} (PG) $\mathcal{A}$ has the form $(N,E,\mathsf{label},\mathsf{prop})$, where: 
		\begin{compactitem}    
			\item $N$ is a non-empty set of \emph{nodes}; %a subset of $N$ is defined as the \emph{individuals};
			\item $E$ is the set of \emph{edges}, where each edge is a triple $(r,n,n')$ with $r\in \rolenames$ and $n,n' \in N$; we may write such an edge in the form $r(n,n')$ and call it an \emph{$r$-edge}; 
            %%We with $n, n' \in N$ and ;
			\item $\mathsf{label}$ is a total function $N \rightarrow 2^{\conceptnames}$; 
            %% with $2^{\conceptnames}$ the power set over the set of concepts $\conceptnames$; and
			\item $\mathsf{prop}$ is a partial function $(N\cup E) \times \propnames \rightarrow \mathbf{D}$ mapping pairs $(u,p)$ with $u \in (N\cup E)$ and $p\in \propnames$ to a value in
            $\mathbf{D}$.
		\end{compactitem}
    If %%$\A$ is finite and 
    $N \subseteq \indivs$ and it is finite, %%all nodes are individuals, 
    we may call $\A$ an \emph{ABox}. 
    %% A pair of a TBox $\T$ and an ABox $\A$ is called a \emph{knowledge base}. 
    We say that $\A' = (N', E',\mathsf{label}',\mathsf{prop}')$
    is a \emph{subgraph} of $\A$
        % if  $N' \subseteq N$, $E' \subseteq E$ and $\mathsf{label}$ (resp. $\mathsf{prop}$) agrees with all assignments of $\mathsf{label}'$ (resp. $\mathsf{prop}'$).}
    if $N' \subseteq N$, $E' \subseteq E$,  $ \mathsf{label}'(n) = \mathsf{label}(n) $ for all $n \in N'$, and $\mathsf{prop}'(u,p) = \mathsf{prop}(u,p) $ for all $u \in N' \cup E'$, $p \in \propnames$.

	\end{definition}
    % \todo[author=\bf R3,color=green!20]{K and D are not defined before being used in the definition of prop.}
    % \todo[author=\textbf{R4},color=green!20]{ $E$ is stated to be a multiset, however it is used as a function mapping roles to binary relations. Set $K$ is undefined. The datatype domain $\mathbb{D}$ is defined locally, but also used later-on. } 

Note that
% that we allow 
% property-value pairs only in the ABox, and that
%%in our definition 
we use the same name of property keys %%is used 
for both nodes and edges, as usually done in PGs~\cite{DBLP:series/synthesis/2018Bonifati, DBLP:conf/sigmod/FrancisGGLLMPRS18}. We differ slightly from the usual definition of PGs, by allowing only a single edge for each role between pairs of nodes.
\CR{This puts our data model in between the usual DL setting and what real-world PG database systems like Neo4j~\cite{DBLP:conf/sigmod/FrancisGGLLMPRS18} support, where multiple edges with the same role between the same pair of nodes are allowed.}
% This restriction matches the data  model of Neo4j, the most popular PG database system \cite{DBLP:conf/sigmod/FrancisGGLLMPRS18}.}

% \todo[inline]{R1.7. I am not sure about the data model but I am pretty sure Neo4J allows to have more that one edge of the same type between the same nodes. So I do not understand "by allowing only a single edge for each role between pairs of nodes. This restriction matches the data model of Neo4j, the most popular PG database system" \\
% R3.3. In Section 2, the data model's "single edge for each role" restriction is explicitly justified by Neo4j compatibility, but the paper does not discuss what is lost in the multigraph case (e.g., two distinct "employs" relationships between the same company and person with different time periods), which is a realistic property-graph scenario.}

% of PGs allows only a single edge for each role between each pair of nodes. Also,  the same name of property keys is used for both nodes and edges, as usually done in PGs~\cite{DBLP:series/synthesis/2018Bonifati, DBLP:conf/sigmod/FrancisGGLLMPRS18}.  

\begin{example}\label[example]{ex:propertyGraph}
\Cref{fig:propertyGraphExample} shows profile information from a social network (such as LinkedIn), represented as a PG. 

\end{example}

\textbf{Ontology language.}
In ontology mediated query answering, an ontology  captures domain knowledge and it allows us to infer implicit facts from the given data at query time. 
Here, the ontology language is called $\ontoLang$. It extends the well-known DL-Lite~\cite{Calvanese2007} with tests over property values, which can be used in the left-hand-side both concept and role inclusions. 

% We recall the definition of $\ontoLang$, a well-known light-weight description logic.
% subsuming the popular lightweight languages DL-Lite$_\R$ and $\mathcal{ELH}$;
% \footnote{For space reasons we omit disjointness axioms, but they can be easily incorporated.} 
% in the next sections we restrict our attention to a fragment of it. For space reasons we omit disjointness axioms, but they can be easily incorporated.
% We also recall the usual \emph{normal form} for $\mathcal{ELHI}$ TBoxes; the proof that every TBox can be normalized (in linear time) while preserving the semantics is standard.
%%reasoning services considered in this paper. }

% extend it with some of the expressiveness of $\E\L$, while keeping data complexity %% of query answering 
% in \NL.

%, or the strings with operators like $mathit{prefix}$ and $\mathit{substring}$.}
\begin{definition}[$\ontoLang$] \label[definition]{def:ontoLang} 

 The set of \emph{roles} is defined as $\allroles=\rolenames \cup \{r^-\mid r\in \rolenames \}$, \co{ where $r^-$ is called the \emph{inverse of $r$}} and the set of \emph{data tests}
%over a key w.r.t. to a  data value and predicate we denote as
as $\datatests = \{ p \odot v \ ? \mid p \in \propnames, \odot \in \mathbf{P^D}, v \in \mathbf{D} \}$.  

    %%To define concept and role inclusions, 
    We use the following syntax,    where $A\in \conceptnames$, $r\in \allroles$ and $T \subseteq \datatests$. 
    \begin{align*}
         B &:= \top \mid A \mid \exists r.\top \mid T & \text{ (basic concepts, data-test concepts)} \\
         E &:= r \mid r^-  & \text{ (basic roles)} \\ 
         C &:= B\mid \neg B & \text{ (general concept)} \\
         F &:= E \mid T & \text{ (roles, data-test roles)} \\
         R &:= E \mid \neg E & \text{ (general role)} 
    \end{align*}
   \co{The symbol $\top$, called `top concept', refers to the set of all nodes of the domain in question.}
    A \emph{concept inclusion} (CI) has the form $B \ISA C$ and  a \emph{role inclusion} (RI) the form $F\ISA R$.  
    % , and a \emph{property inclusion} (PI) of the form $P \ISA V $. %%\in \rolenames$, 
       A TBox is a finite set of CIs and RIs. 
%           %or $r\ISA \lnot s$ 
    % %A \emph{transitivity axiom} is an expression $\textsf{trans}(r)$, where $r$ is a role.
%       % \todo[inline]{Does this normal form still make sense?}\vspace{-8mm}
%       \begin{align*}	
%            \ A &\ISA  B 
%              &\ \exists r.\top &\ISA B
%              &\ p \odot v &\ISA B 
%              &\ p \odot v &\ISA r 
%             &\  r &\ISA s   \\  
%            \ A &\ISA \lnot B 
%              &\ \exists r.\top &\ISA \neg B  
%              &\ p \odot v &\ISA \neg B        
%              &\ p \odot v &\ISA \neg r       
%             &\  r &\ISA \neg s   \\  
%            A &\ISA \exists r.X 
%              &\hphantom{\ \exists r.\top} &\hphantom{\ISA \exists r.X }
%              &\ p \odot v &\ISA \exists r.X  
    % \end{align*}    
%       where $X \in \{\top\} \cup \{B\}$, $A,B \in \conceptnames$, $r,s \in \allroles$ and $p,p' \in \propnames$.
We use $\transclosure{}{}$ to denote the reflexive transitive closure of $\{(r,t)\mid r\ISA t\in \T\}$ and call $r$ a \emph{subrole of $t$ (in $\T$)} if $\transclosure{r}{t}$.
\end{definition}

% in the input graph can be done using the standard DL-Lite algorithms~\cite{Calvanese2007}

%% All the algorithms given in this paper assume that we work over consistent ontologies.

%% We say $\I$ is a \emph{model} of $\mathcal{T}$ if $\I$ satisfies every axiom in $\T$.
%%\todo{define notion of entailment and complex concept}

% \todo[inline]{consistency checks (equivalent to \cite{Calvanese2007}?)}

\co{The semantics of concepts, roles and data sets is defined with respect to a given PG $\mathcal{A}$. In the description logic (DL) literature, this is usually defined by a so-called interpretation of the form $\I=(\varDelta^\I,\cdot^\I)$, with $\varDelta^\I$ a non-empty set called the \emph{abstract domain} and $\cdot^\I$ as the interpretation function. In this paper, we consider a given PG to be an interpretation in the DL sense. Hence, given a PG $\mathcal{A}$, by slight abuse of notation, we may also use it as an interpretation function that maps concepts, roles and data sets to, resp., sets of nodes, pairs of nodes or a union of the two. This function is given in \Cref{tab:dlpg}. 
In the sequel, when the PG $\mathcal{A} = (N,E,\mathsf{label}, \mathsf{prop})$ is clear from context, we shall refer to it as ``interpretation'' and, for example, refer to its set of nodes $N$ as its domain $\varDelta^\mathcal{A}$.
% We do this to better connect this work to the wider literature where the term ``interpretation' is more common. 
% It should be clear that the underlying object is just a property graph.
} 
% The semantics is usually defined via \emph{interpretations} of the form $\I=(\varDelta^\I,\cdot^\I)$, with $\varDelta^\I$ a non-empty set called the \emph{abstract domain}. $\cdot ^\I$ is the \emph{interpretation function}, which assigns to every $A\in \conceptnames$ a set $A^\I\subseteq \varDelta^\I$, to every $r\in \rolenames$ a relation $r^\I\subseteq \varDelta^\I\times\varDelta^\I$ and to every $p \in \propnames$ a relation $p^\I \subseteq (\varDelta^\I \cup ( \varDelta^\I \times \varDelta^\I  )  ) \times \mathbf{D}$.

% As usual, the semantics is given in terms of \emph{interpretations} $\I=(\varDelta^\I,\cdot^\I)$ with $\varDelta^\I$ a non-empty set called the \emph{domain} and $\cdot ^\I$ is the \emph{interpretation function}, which assigns to every $A\in \conceptnames$ a set $A^\I\subseteq \varDelta^\I$, and to every $r\in \rolenames$ a relation $r^\I\subseteq \varDelta\times\varDelta$. 
% It is extended to concepts, roles and data tests as in \Cref{tab:dlpg}. 
\co{As for the rest of the semantics, the }\mo{satisfaction of CIs and RIs is as expected: 
$\I \models C \sqsubseteq D$ iff $C^\mathcal{I} \subseteq D^\mathcal{I}$, and $\I \models E \sqsubseteq F$ iff $E^\mathcal{I} \subseteq E^\mathcal{I}$.} 
\co{We write $\I \models \T$ iff $\I \models C \sqsubseteq D$ for every $ C \sqsubseteq D \in \T$.  We write $\T\! \models  C \sqsubseteq D $ iff for every PG $\I$ with $\I \models \T\!$ we also have $\I \models  C \sqsubseteq D$.}

\co{
   For readabilty, in the examples we use CIs of the form
   $\concept{A} \ISA\exists \role{r}.\concept{B}$,
   %%$\concept{TechCompany}\ISA \exists\role{emp}.\concept{Engineer}$,
   which are not strictly in $\ontoLang$. These are just syntactic sugar: 
   we can write $\concept{A} \ISA\exists \role{r}.\concept{B} $ in $\ontoLang$
   %%%cannot directly express CIs like $\concept{A} \ISA\exists \role{r}.\concept{B} $, but it can be expressed by 
    using a fresh role name $\roleMath{r_B}$ and three axioms: 
    $\concept{A} \ISA\exists \roleMath{r_B}.\top $, 
    $\exists \roleMath{r_B}^-.\top \ISA B$, 
    $\roleMath{r_B} \ISA \roleMath{r}$.}

\begin{example} \label[example]{ex:soc_net}
% \todo{MO: I moved the example after the semantics}  
The following $\ontoLang$ TBox provides knowledge that can be applied to the social network graph of \Cref{ex:propertyGraph}: 
% \todo{explain axioms. Format in a more readable way? }
{ \small
\begin{align*}
    \T = \   
    \left \{ 
    \begin{aligned}
        & \{\property{born}\ge \text{1997} \ ?, \property{born}\le \text{2012} \ ? \} \ISA\concept{GenZ}, \\
        &  \exists\role{emp}^-.\top\ISA \concept{Employed},         \{\property{time}\ge(\text{2025-01-01}) ?\}\ISA \concept{Recent}, \\ 
        & \concept{Employed} \ISA\lnot\concept{Unemployed}, \concept{Opole}\ISA \concept{Poland}, \exists\role{ann}.\top\ISA\concept{Hiring} \\ 
        & \concept{TechCompany}\ISA \exists\role{emp}.\concept{Engineer},         
        \role{friend}\ISA\role{friend.}^-
    \end{aligned}    
    \right \}
\end{align*}
}
\co{ To briefly explain a few of the axioms, we see that the first axiom in $\T$ defines the new concept $\concept{GenZ}$, which applies to anyone born between the years 1997 and 2012. The second axiom says that someone who has an employer is employed, and the fourth that no one can be both employed and unemployed.}
%%$\T$ also informs us that indeed the concept $\concept{Opole}$ is part of the concept $\concept{Poland}$.}
%%    To improve readability, we thus use the former CI as a shorthand for the latter three.}
% { \small
%         $\T= \big \{ \ \{\property{born}\ge \text{1997} \ ?$,
%         $\property{born}\le \text{2012} \ ? \} \ISA\concept{GenZ}$, 
%         $\exists\role{employs}^-.\top\ISA \concept{Employed}$, $\{\property{time}\ge(\text{2025-01-01}) ?\}\ISA \concept{Recent}$,$ \concept{Employed} \ISA\lnot\concept{Unemployed}$, 
%         $\concept{Opole}\ISA \concept{Poland}$,  $\concept{TechCompany}\ISA \exists\role{employs}.\concept{Engineer}$, 
%         $\exists\role{announce}.\top\ISA\concept{Hiring}$,
%         $\role{friendsWith}\ISA\role{friendsWith}^-\big \}$
%     }
\end{example}

{For the problem of checking whether a given PG is consistent with an $\ontoLang$ TBox, we refer to 
% the work of 
Artale et al.~\cite{DBLP:conf/ecai/ArtaleRK12}; their algorithms can be adapted to account for our restricted use of property values.}

\medskip\noindent
    \textbf{Query Language.}
    We study \emph{conjunctive nested two-way regular path queries} (CN2RPQs), the extension of C2RPQs, the navigational query language for graphs that has received most attention in OMQA. We enhance CN2RPQs by \emph{data tests} similar to data tests for navigational conjunctive queries in \cite{DBLP:conf/dlog/DragovicO023} to query for property values, assuming that the predicates in $\mathbf{P^D}$ can be realized in GQL and Cypher. To represent CN2RPQs we rely on \emph{nested nondeterministic finite automaton} (\nnfa), following generally the notions from~\cite{Bienvenu2014}.
    % , with some simplifications. 

    We first define nested two-way regular path expressions using NFAs, and then define conjunctive nested two-way regular path queries based on them. 
We assume a given alphabet $\Sigma$, and to define nested automata, we extend it by allowing n-NFAs as symbols. 
    % Where convenient we use the N2RPE representation for n-NFAs or the notational short form $\mathbf{A}_{s_1, F_1}$ for an $n$-NFA $(\mathbf{A}, s_1, F_1)$, following \cite{Bienvenu2014}.

    {
\begin{definition} \label[definition]{def:cn2rpqNFA:Bart}
        Let $\mathbf{A}^0$ be the set of all NFAs over a given alphabet ${\Sigma}$. For $k > 0$, the set $\mathbf{A}^k$ of \emph{nested NFA} (\nnfa) over $\Sigma$ of nesting depth $k$ contains all automata over ${\Sigma}^k =  {\Sigma} \cup \{ \nesting{\alpha} \mid \alpha \in \mathbf{A}^i, 1 \leq i < k\}$. 
        %% We write $\mathbf{A}$ for the set $\bigcup_{k \geq 0} \mathbf{A}^k$ of all n-NFAs.  
\end{definition}

We omit the nesting depth of an \nnfa~when it is irrelevant. 
We are interested in \nnfa{s} that use the alphabet of symbols and tests that may occur in our PGs. 
    
    \begin{definition}
    A \emph{nested two-way regular path expression} (N2RPE) is an \nnfa~over the specific alphabet  
    
    \[ {\mathbf{\Sigma}_{PG}} =  2^{\allroles \cup \datatests} \cup \datatests \cup \{ \mathrm{C}? \mid \mathrm{C} \in \conceptnames \} .
        \]
    %%This alphabet 
    $\mathbf{\Sigma}_{PG}$ consists of three kinds of symbols: 1) sets of (possibly inverted) role labels and data tests, which allow the expression to perform a series of data tests while traversing an edge that has  all the role labels in the test; 2) individual data tests, which are checked against the properties of a node and lastly, 3) concept tests, which check whether the current node is part of the required concept. 
\end{definition}

\mo{Sets of binary (data and role) tests allow us to do several tests while traversing one edge. In contrast, sets of unary (data and concept) tests do not need to be performed simultaneously; they can be tested equivalently by consecutive transitions.}

\mo{We can now use N2RPEs as a query language for PGs. We define the semantics of a 2NRPE $\alpha$ via paths %%in the interpretation 
that match $\alpha$, possibly with outgoing paths to match the nested automata.}

\begin{definition}[Run of an N2RPEs]\label[definition]{def:runN2RPE}
    Let $\alpha_0 = \langle Q_0, \mathbf{\Sigma}_{PG},S_0,\delta_0, F_0 \rangle $ be an N2PRE of nesting depth $k$ and let $\alpha_1, \ldots, \alpha_m$ be the   N2PREs of nesting depth $< k$ mentioned in the transitions of $\alpha_0$. 
    We let $\alpha_i = \langle Q_i, \mathbf{\Sigma}_{PG}, S_i, \delta_i, F_i \rangle$ for each $\alpha_i$.
    Then, a \emph{run} for $\alpha$ on $\I$ is a finite node-labelled tree $(T,l)$ such that every node is labelled with an element from $\Delta^\I\times \bigcup_i Q_i$.
    Furthermore, for the root $r_T$ of $T$ we have that $l(r_T) = (n_r,s_r)$, where $s_r \in S_0$ and for every leaf node $e \in N(T)$, it holds that $l(e) = (n,s_n)$, with $s_n \in F_i$ for some $i$ and there exists exactly one leaf $n'$ such that $l(e') = (n',s')$, with $s' \in F_0$, called the \emph{accepting node}.
    For each non-leaf node $v \in N(T)$ having label $l(v)=(n_v,s_v)$ with $s_v \in Q_i$, one of the following holds:
    \begin{compactitem}[$\bullet$ \hspace{-0.3em}]
        \item $v$ has a unique child $v_c$ with $l(v_c)=(n_c,s_c)$, and there exists $(s_v,\sigma,s_c)\in\delta_i$, such that one of the following holds:
        % such that $\sigma\in \rolenames $ and $(n_v,n_c)\in\sigma^\I$

        \begin{itemize}
            \item If $\sigma \in 2^{\allroles \cup  \datatests} $ then:\\ \emph{(i)}  $ (n_v,n_c) \in  (p \odot v?)^\I$ for each data test $p \odot v? \in \sigma^\I$, and \\ \emph{(ii)} $(n_v, n_c) \in r^\I$ for each $r \in \sigma^\I \cap \allroles$;
            \item if $\sigma = \mathrm{C}?$ , then $n_c = n_v$ and $n_v \in \mathrm{C}^\I$, and 
            \item if $ \sigma = p \odot v?$, then $n_v = n_c$ and $n_v \in (p \odot v?)^\I$.
        \end{itemize}
        
        \item $v$ has exactly two children $v_{c_1}$ and $v_{c_2}$ with $l(v_{c_1})=(n_v,s_{c_1})$ and $l(v_{c_2})=(n_v,s_{c_2})$, with $s_{c_2}$ the initial state of $\alpha_j$, and there exists a transition $(s_v,\nesting{\alpha_j},s_{c_1})\in\delta_i$.
    \end{compactitem}
    If the root $r_T$ is labelled $(n_r,s_0)$ and the accepting node $u$ is labelled $(n_u,s_u)$ with $s_0 \in S_0$ and $s_u \in F_0$, then $(T,l)$ is called an $(n_r,n_u)$-run, 
    \mo{and if such a run exists, we write $(n_r,n_u)\in \alpha_0^{\I}$.} 
    %%Note that $\alpha_0^{\I}$ defines a binary relation over the nodes of $\I$.}
    % \todo[inline]{we are still missing the notation of a pair being in the query answer, right?}
\end{definition}

N2RPEs are a very rich query language in their own right, but we get even more flexible querying if we join N2RPEs using variables.

\begin{definition}    
 A \emph{conjunctive N2RPQ} (CN2RPQ) $q(\vec{x})$ is a conjunction of \emph{atoms} $\nNFA{1}(x_1,y_1) \land \cdots\nNFA{i}(x_i,y_i)\cdots\land \nNFA{n}(x_n,y_n)$ with each $\nNFA{i}$ an N2RPE and $\vec{x} \subseteq 
 \bigcup_i^n \{x_i,y_i\}$ is a tuple of \emph{answer variables}.  
 We call a CN2RPQ \emph{Boolean} if $\vec{x}$ is empty. 

The set of \emph{join variables} $\mathit{Join}(q(\vec{x}))$ of  
a CN2RPQ $q(\vec{x})$ is the set of those variables that occur in two different atoms of $q(\vec{x})$. 
If $\mathit{Join}(q(\vec{x})) \subseteq \vec{x}$ for a CN2RPQ that is not Boolean, then we call $q(\vec{x})$ a \emph{join-on-free} CN2RPQ. 
\end{definition}

% As in~\cite{Bienvenu2014}, we have that for every N2RPE $\pi$, one can construct in polynomial time an $n$-NFA $A_{s_0, F_0}$ such that $\pi^\I =  \mathbf{A}^\I_{s_0,F_0}$ for every interpretation $\I$. 

We may shorten $\nesting{\alpha}(x,x)$ atoms as $\nesting{\alpha}(x)$, and 
$A?(x)$ as $A(x)$. 
%%It will be convenient  to 
We assume that CN2RPQs are connected, that is, the graph whose vertices are the variables and that has an edge between two variables if they occur in the same atom is a connected graph. Disconnected queries can be answered as separate queries and then their answers intersected. 
Note that, under this assumption, every atom in a join-on-free 
CN2RPQ has at least one answer variable. 
% \todo[inline]{R1.6. The restriction to join-on-free queries is mentioned early in the paper, but its impact on real-world use cases is only briefly defended in the conclusion. A stronger justification for this choice should be better discussed.}

\begin{definition}[Semantics of CN2RPQs]
    Consider a CN2RPQs $q(\vec{x})$ and 
    a PG $\A=(N,E,\mathsf{label, prop})$.
 A \emph{match} $\mu$ for $q(\vec{x})$ in $\A$ is a mapping from the variables in $q$ to nodes in $N$ such that $(\mu(x),\mu(y)) \in\alpha^{\I_{\A}}$ for every atom $\alpha(x,y)$ occurring as a conjunct in $q(\vec{x})$. 
 The tuple $\mu(\vec{x})$ is then called an \emph{answer} for $q(\vec{x})$ in $\A$. 
\end{definition}

\begin{example} \label[example]{ex:query}
    To illustrate our query language, we use the schema introduced in \Cref{ex:soc_net} and present four example queries, with informal explanations of their semantics following below.     
 { \small
    \begin{align*}
        q_1(x,y) \coloneqq \ 
        % \begin{aligned}
        & \role{friendsWith}^*\cdot\role{employs}^-(x,y), \\  &\nesting{\role{locatedIn}\cdot\concept{Poland}?}\cdot\concept{Hiring}? \cdot \concept{TechCompany}?(y,z) 
        % \end{aligned}
        \\        
        % q_1'(y)&:=\role{friendsWith}^*\cdot(\role{employs}^-\cdot\nesting{\role{locatedIn}\cdot\concept{Poland}}\cdot \role{foundedBy})^*\cdot\concept{Hiring}?\cdot \concept{TechCompany}?(x,y) \\
        q_2(x,y) \coloneqq \ 
        % \begin{aligned}
        & \{\role{employs}^-,\property{since}\le{\scriptsize 2023}?\}\!\cdot\!\role{friendsWith}^*\!\cdot\! \property{born}\!\ge\! 
        {\scriptsize 2000}?(x,y), \\ 
        &\nesting{\{\role{viewed},\property{time}\ge 01.01.2024?\}}\cdot\role{friendsWith}(z,y) 
        % \end{aligned}
        \\ 
        q_3(x,y) \coloneqq \  
        % \begin{aligned}
                  & \property{founded} \leq 2003 ? \cdot  \property{revenue} \geq 10^5 ? \cdot  \\ 
                  &(\{ \role{owns}, \property{since} \geq 2020 ?\} \cdot \nesting{\role{locatedIn} \cdot \concept{Poland}?})^*(x,y)   
        % \end{aligned}
        \\
        q_4(x,y) \coloneqq \  
        % \begin{aligned}
        &\concept{Company}? \cdot \role{employs} \cdot \role{friendsWith} \cdot \concept{GenZ}?(x,y),\\   
        &\role{viewed} \cdot \concept{Job}? \cdot \role{announce}^- (y,x) 
        % \end{aligned}
    \end{align*}
    }
    The query $q_1(x,y)$ retrieves all technical companies located in Poland that are currently hiring and that have an employee who is a friend of friends of the user.
    % , e.g., as a recommendation for someone looking for a position. 
    In query $q_2(x,y)$ a recruiter might be interested in friends (and friends of friends) born after 2000 of an employee, who has been working for the company for over three years. In query $q_3(x,y)$, we look for companies that were founded before 2003, with a revenue of $10^5$ or more, and want to find all Polish companies that they own since 2020, and also look for all companies that these companies own, and so on. 
    % Nesting is used to restrict the search to companies located in Poland. 
    In query $q_4(x,y)$, we look for all companies that employ workers that are friends with some user of GenZ, and return the pairs of company and users, when the user also viewed a job offer by the company that the friend is currently employed.
    % Nesting allows us to check for the friend, while retaining the employee as the target node.  
    %\todo{Explanation of the queries needs to be updated, e.g., "founded before 2003" is not in query $q_3$. Cem: restored the query. This also matches how q3 is presented in the experiments.}
   
\end{example}
% The query $q(x)$  retrieves all
% datasets from an MRI with a certain specification and data on ambidextrous participants:
% \begin{align*}
% q(x):=&\mathsf{\langle Dataset\rangle(x) }\land \mathsf{\{Manufacturer="SIEMENS"}\land \mathsf{MagnetFieldStrength} \ge 3\}(x)\\ 
% &\land \mathsf{has}^*(x,y) \land \mathsf{\langle Participant\rangle}(y)\land \mathsf{\{Handedness="ambidextrous"\}}(y)
% \end{align*}
% We use here a concrete domain that contains integers and 
% data tests on the properties $\property{born}$, $\property{since}$, $\property{founded}$ and $\property{revenue}$. 
% The Kleene star is the distinctive navigational feature of graph query languages, absent from any FO-rewritable query language; it lets us explore paths of unbounded length across the property graph.

Following the OMQA 
literature, we use the \emph{homomorphism} or \emph{walk} semantics for path queries~\cite{Angles2017}, and in the presence of a TBox, adopt the \emph{certain answer} semantics.  
\begin{definition}\label[definition]{def:certainAnswer}
Consider an PG $\A$ and a TBox $\T$. A 
tuple $\vec{a}$ of individuals in $\A$ is called a \emph{certain answer} to $q(\vec{x})$ over $(\T,\A)$ if it is an answer to $q(\vec{x})$ in $\mathit{PG}(\I)$ for every model $\I$ of $\A$ and $\T$.  
\end{definition}

The worst-case complexity of evaluating navigational queries over ontologies is well understood. Unfortunately, when we take the ontology's size and the query into account, it is highly intractable. 

\begin{theorem}[Theorem 5.1 \bl{\& Corollary 6.8} ~\cite{Bienvenu2014}]
    In DL-Lite, N2RPQ answering is \bl{\textsc{NL}-complete in data complexity and} \textsc{Exp}-complete in combined complexity. 
    In all description logics containing $\mathcal{EL}$ and contained in Horn-$\mathcal{SHIQ}$, N2RPQ answering is \bl{\textsc{P}-complete in data complexity and} \textsc{Exp}-complete in combined complexity. 
    % \todo[inline]{and NL in data , right? same Thm? BL: for EL the data complexity is not NL.\\ MO: Of course! I meant DL-Lite, sorry. We could also mention P hardness in data for EL  (same claim or another?) BL: so Thm 5.1. is only about the Exp-c, Corollary 6.8. about the NL, I'll try to find the Thm for EL data complx. UPDATE: in Table 1 of the paper we would have the summary of all the complexities, refer to that?}
\end{theorem}

Our aim here is not worst-case optimal complexity, but to develop techniques that may lead to tractable evaluation of nested navigational queries over ontologies in practice. 
\section{Rewriting N2RPQs}
\label{sec:rewrite_cof_N2RPQ}
% \mo{We are given  infinite, countable and pairwise disjoint sets $\conceptnames$, $\rolenames$ and $\datatests$ of (resp.) concept names, role names and data tests.}\todo{why here? can it go?}

\mo{We present a rewriting algorithm for N2RPQs with data tests. In this section and the next, we assume a fixed $\ontoLang$ TBox $\T$, and provide an algorithm for rewriting a given input N2RPQ $q$ into an N2RPQ $q_{\T}$ that is sound and complete for all ABoxes that are consistent with $\T$.  
In the next section, we lift this algorithm to (restricted) conjunctions of N2RPQs.}

Our algorithm draws inspiration from and builds on previous techniques for navigational ontology-mediated querying  \cite{Calvanese1999,Eiter2012,Bienvenu2015,Ortiz2011}, but it focuses on potential implementability. 
% Both \co{functions} draw inspiration from \emph{loop computation} \cite{Bienvenu2013}, \emph{clipping}, and \emph{JumpTrans} \cite{Eiter2012,lohnert2025,Bienvenu2014}. \todo[inline]{we can be more accurate: clipping is mostly about the join variables, i.e. it should be mentioned in sect 4. the loop computation and the jumptrans are essentially two names for the same idea, but the former name has been used in a few papers on 2RPQs: our old JAIR paper on navigational queries, the one with sebastian that considers more expressive DLs (we used the same idea, but I do not remember if we called it looping), and the reasoning web tutorial. The latter name is the one used in the only nesting paper. We should also mention that the reasoning web paper has a more conceptual discussion of the looping technique.}

As usual for ontology-mediated querying, we  rely on the fact that for every 
PG $\A$ that is consistent with $\T$, there is a canonical model 
$\I_{\T,\A}$ that gives exactly the certain answers to all N2RPQs.
Below we give the construction for the canonical model.

\begin{definition}[Canonical Model]\label[definition]{def:canonicalmodel}
    Consider a $\ontoLang$ TBox $\T$ and a PG $\A$. 
    Let $\I_0 = \A$.
    Then, we construct the \emph{canonical model} $\I_{\T,\A}$ by exhaustively applying the following rules to $\I_0$.
    \begin{compactenum}[(Ch1)]
        \item\label{def:canonicalmodel:RI} If \ $r\ISA r'\in\T$, $(v,w)\in r^{\I_i}$ and $(v,w)\notin s^{\I_i}$, then $\I_{i+1}$ is $\I_i$ with ${r'}^{\I_{i+1}}={r'}^{\I_{i}}\cup \{(v,w)\}$ in case $r\in \rolenames$, otherwise ${r'}^{\I_{i+1}}={r'}^{\I_{i}}\cup \{(w,v)\}$.        
        % \todo{We should not use $s$ for roles! Cem: s -> r'} 
        \item\label{def:canonicalmodel:roleProps} If \ $T\ISA r'\in\T$, $(v,w)\in T^{\I_i}$ and $(v,w)\notin {r'}^{\I_i}$, then $\I_{i+1}$ is $\I_i$ with ${r'}^{\I_{i+1}}={r'}^{\I_{i}}\cup \{(v,w)\}$.
        \item\label{def:canonicalmodel:CI} If \ $A \ISA B\in\T$, $v\in A^{\I_i}$ and $v\notin B^{\I_i}$, then $\I_{i+1}$ is $\I_i$ with $B^{\I_{i+1}}=B^{\I_{i}}\cup\{v\}$. 
        \item\label{def:canonicalmodel:existsLHS} If \ $\exists r.\top \ISA B\in\T$, $v\in (\exists r.\top)^{\I_i}$ and $v\notin B^{\I_i}$, then $\I_{i+1}$ is $\I_i$ with $B^{\I_{i+1}}=B^{\I_{i}}\cup\{v\}$.        
        \item\label{def:canonicalmodel:conceptProps} If \ $T\ISA B\in\T$, $v\in T^{\I_i}$ and $v\notin B^{\I_i}$, then $\I_{i+1}$ is $\I_i$ with $B^{\I_{i+1}}=B^{\I_{i}}\cup\{v\}$.
        \item\label{def:canonicalmodel:existsRHS} If \ $A \ISA \exists r.\top\in\T$, $v\in A^{\I_i}$ and $v\notin (\exists r.\top)^{\I_i}$, then $\I_{i+1}$ is $\I_i$ with a fresh $w$ in $\varDelta^{\I_{i+1}}$.
        % and $C^{\I_{i+1}}=C^{\I_{i}}\cup\{w\}$. 
        In case $r\in \rolenames$ then $r^{\I_{i+1}}=r^{\I_{i}} \cup\{(v,w)\} $, otherwise $r^{\I_{i+1}}=r^{\I_{i}} \cup \{(w,v)\}$.
        \item\label{def:canonicalmodel:existsRHS2} If \ $\exists r.\top \ISA \exists {r'}.\top\in\T$, $v\in (\exists r.\top)^{\I_i}$ and $v\notin (\exists s.\top)^{\I_i}$, then $\I_{i+1}$ is $\I_i$ with a fresh $w$ in $\varDelta^{\I_{i+1}}$.
        % and $C^{\I_{i+1}}=C^{\I_{i}}\cup\{w\}$. 
        In case $s\in \rolenames$ then ${r'}^{\I_{i+1}}={r'}^{\I_{i}} \cup\{(v,w)\} $, otherwise ${r'}^{\I_{i+1}}={r'}^{\I_{i}} \cup \{(w,v)\}$.
    \end{compactenum}
\end{definition}

\mo{Note that the property values in $\ontoLang$ have no effect when it comes to the canonical model construction. Since data value tests only occur on the left-hand-side of axioms, in $\I_{\T,\A}$ the partial function $\mathsf{prop}$ stays as-is and does not assign any new data values, neither to nodes in $\A$ nor to freshly added ones. }

For convenience, we may think of $\I_{\T,\A}$ as a set of \emph{facts}: $A(o)$ and $r(o,o')$ assert node and edge labels, plus facts $o.(p=v)$,  $r(o,o').(p=v)$ that assert property value assignments.  
% $T(o)$ and $T(o,o')$, 
where $o,o'$ are nodes, $A \in \conceptnames$, $r \in \rolenames$, $p \in \propnames$ and $v \in \mathbf{D} $.
% and $T \in \datatests$.  
%\todo{do we want to use o for objects in the canonical model? or u,v,w?  or d, e, which is common and safe? We should unify!}
We call a fact \emph{explicit} if it is present in $\A$, and \emph{implicit} otherwise. 
 That is, implicit facts are those introduced by the canonical model construction that extends the explicit facts to satisfy the TBox $\T$. \mo{All property value assertions are explicit.}
Note that, as usual for DL-Lite variants, each rule application that adds an implicit fact $\xi_\ell$ tests for the presence of one fact $\xi_{\ell-1}$, which may be explicit or added by a previous rule application. In this way, every implicit fact is triggered by one explicit fact and a sequence of rule applications.

\mo{Note that all explicit facts mention only the individual names in $\A$, but the implicit ones may mention additional nodes introduced to satisfy existential restrictions. We call the latter, that is, the nodes in $\I_{\T,\A} \setminus \A$, the  \emph{anonymous nodes}, and we refer to the facts in $\I_{\T,\A}$ that mention at least one anonymous node as the \emph{anonymous part} of $\I_{\T,\A}$.
It is not hard to see that if there is a single fact $C(a)$ in $\A$, then all the anonymous nodes added during the model construction form a tree-shaped graph rooted at $a$. Moreover, since all rules are deterministic, this tree is unique up to renaming of nodes, and the canonical model of any ABox containing an individual that satisfies $C$ will contain a copy of this tree.}

%%and given in the {extended version of this paper}~\cite{appendix}.\todo{rephrase ref!}

% \todo[inline]{Let's bring back the def of canonical model, its not too long and worth having}

% and it remains intact for the individuals in $\A$.} That is, all data value assignments are explicit. 

\mo{Our rewriting of a single N2RPQ $q$ is based on two \emph{skipping functions} that take an N2RPE as input and output a modified N2RPE with additional transitions that allow it to `skip over' implicit facts, and instead look for the explicit facts that trigger their generation. The rewriting is done in two phases. First, to account for the anonymous part, the function $\mathsf{skipAnon}$ adds transitions that skip over the anonymous part of the canonical model. 
The answers of the resulting query over the facts of $\I_{\T,\A}$ involving only non-anonymous individuals coincide with the answers of the original $q$ over the whole $\I_{\T,\A}$. Then \co{a second function} $\mathsf{skipCore}$ adds more transitions to yield complete answers when evaluated over just $\A$ alone.
% , without any implicit facts. 
}

\subsection{Skipping over the Anonymous Part}
%\co{We present the skipping functions in two parts. The first part concerns with adding new transitions based paths that need to cross into the anonymous part, and the second part concerns adding new transitions based on paths that stay in the graph.}

\mo{Given $\T$ and $\A$, we use $\mathcal{G}_{\T,\A}$ to denote
the restriction of the canonical model $\I_{\T,\A}$ to the nodes in $\A$, and call it the \emph{core} of  $\I_{\T,\A}$.
The goal of the first phase of the rewriting is to obtain a query that skips over the anonymous part of $\I_{\T,\A}$ and finds all query answers when evaluated over the core $\mathcal{G}_{\T,\A}$ alone. 

To this aim, we rely on the following \textsf{SkipRel}$(\alpha,\T)$ relation, originally defined as two relations \textsf{JumpTrans} and \textsf{JumpFinal} in \cite{Bienvenu2014}. This relation characterises the parts of the input N2RPQ $\alpha$ that can be matched in the canonical models of each tree-shaped part that can possibly appear in a canonical model constructed from $\T$.}
% This definition is called is a   following \textsf{SkipRel} Definition and computation of JumpTrans by TBox Reasoning (based on Proposition 6.3.) from \cite{Bienvenu2014} with minor modifications}

% \todo[inline]{This is the place to define a partial run}
Recall the definition of run from \Cref{def:runN2RPE}. A \emph{partial run} for $\alpha$ is defined analogously, but does not need to satisfy the conditions on the leafs, i.e. it may simply end on non-final states. 

\begin{definition}\label[definition]{def:jumptrans}
    Let $\T$ be a $\ontoLang$ TBox and $\alpha_0$ an n-NFA of depth $k$, and let $\alpha_1, \dots, \alpha_n$ be the n-NFAs of nesting depth $\le k$ mentioned in the transitions of $\alpha_0$. The set \textsf{SkipRel}$(\alpha_0,\T)$ consists of tuples $(C,s_0^1,s_0^2,\Gamma)$ where $C\in\conceptnames\cup\{\top\}$, 
    $s_0^1 \in Q_0$ is a state from $\alpha_0$,
    $s_0^2 \in (Q_0 \cup \{\mathsf{f}\}) \setminus \{s_0^1\}$ is either a different state from $\alpha_0$ or the special marker $\mathsf{f}$, 
    and $\Gamma\subseteq \bigcup_{1 \leq j \leq n}Q_j$ is a set of states from the automata nested in $\alpha_0$. 
    Such a tuple $(C,s_0^1,s_0^2,\Gamma)$ belongs to \textsf{SkipRel}$(\alpha_0,\T)$ if there exists a partial run $(T,l)$ of $\alpha_0$ on the canonical model $\I_{\T,\{C(a)\}}$ that satisfies the following conditions: 
    \begin{compactitem}
        \item The root of $T$ is labelled $(a,s_0^1)$.
        %%\item If a node in $T$ is labelled $(a,s)$, then it is either the root or a leaf.
        \item If $s_0^2 \in Q_0$, there is a leaf node $v$ with $l(v)=(a,s_0^2)$.
        \item For every leaf node $v$ with $l(v)=(o,s)\neq(a,s_0^2)$, either $s \in F_j$ for some $0 \leq j \leq n$, 
        %%$\alpha_j\in\mathbf{A}^{k-1}$, 
        or $o=a$ and $s \in\Gamma$.
    \end{compactitem}
     If\CR{, additionally,} $T$ has at least one node $(o,s)$ with $o\neq a$, then we call the tuple  $(C,s_0^1,s_0^2,\Gamma)$ \emph{strictly anonymous}. 
\end{definition}

\mo{Intuitively, a tuple $(C,s_0^1,s_0^2,\Gamma)$ is in \textsf{SkipRel}$(\alpha_0,\T)$ precisely when in the canonical model of every ABox containing a node $a$ with $a \in C^{\I_{\T,\A}}$, we are guaranteed to find a match for the query $\alpha'$ obtained by making $s_0^1$ the initial state of $\alpha_0$, $s_0^2$ the final state, and making all the states in $\Gamma$ \bl{initial} states of the respective automata. The match uses only the tree of $\I_{\T,\A}$ rooted at $a$, and it is such that the automata will be at $a$ whenever they visit a state $s$ in $\Gamma$ (which intuitively acts as \bl{initial} state in the modified $\alpha'$ but not in the original $\alpha_0$). This allows us to 'glue' together the partial match in the tree rooted at $a$ with any partial matches for the rest of the query found in the rest of the canonical model. Hence, if we find a node that satisfies $C$, we can fully skip the subquery $\alpha'$.}
% \todo{clear? helpful? Cem: short of a picture, probably as clear as it can be.}

\co{The computation of  \textsf{SkipRel}$(\alpha_0,\T)$ 
%%for an n-NFA $\alpha_0$ and $\ontoLang$ TBox $\T$ 
can be easily reduced to standard description logics reasoning.}
\mo{In the claim below, we use a TBox in the well-known description logic $\mathcal{ELHI}$, where concept inclusions take the form $C \ISA D$ with $C,D$ concepts formed according to the grammar $C,D := \top \mid A \mid \exists r.C$, and
role inclusions take the form $r \ISA t$, with $A\in \conceptnames$ and $r,t \in \allroles$.}

\begin{lemma}[Proposition 6.3 in \cite{Bienvenu2014}]
%%Materialising  \textsf{SkipRel}]
\label{lem:matJT}
Consider an n-NFA $\alpha_0$ and a $\ontoLang$ TBox $\T$, and let $\alpha_1, \dots, \alpha_n$ be the n-NFAs of nesting depth $k$ mentioned in the transitions of $\alpha_0$. For every state $s\in Q_i$ of $\alpha_i$, with $1 \leq i \leq k$, we use a fresh concept name $A_s$, and we let $A_{\mathsf{f}} = \top$. 
%%If $\{s_1,s_2\}\subseteq S_i$, then 
Let  $\T'$\! be obtained by keeping only the 
$\mathcal{ELHI}$ axioms of \ $\T$\!, and adding to it the following concept inclusions, for every \mo{$\{s,s'\}\subseteq Q_i$, $0 \leq i \leq n$}: 
    \begin{compactitem}
        \item $\top\ISA A_s$, for every $s\in F_j$ with $\alpha_j$,
        \item $\exists r.A_{s'}\ISA A_s$ for every $(s,r,s')\in\delta_i$,   
        %%\bl{and there is some $C\ISA \exists t.\top\in\T$ with $\transclosure{t}{r}$ or $\transclosure{t}{r^-}$},
        %\todo{I dropped the color part which would need updated proofs. The proof in [6] does not have it, and in practice it is just a minor optimisation}
        \item $A_{s'}\AND B\ISA A_s$, whenever $(s,B?,s')\in\delta_i$, and
        \item $A_{s'}\AND A_{s''}\ISA A_s$, whenever $(s,\nesting{\alpha_j},s')\in\delta_i$ and \co{$s'' \in S_j$ is an initial state of $\alpha_j$.}
    \end{compactitem}
    Then the following holds: 
    \[
      (C,s_0^1,s_0^2,\Gamma)\in\textsf{SkipRel}(\alpha_0,\T) \text { iff } \T'\models (C\AND A_{s_0^2} \AND \bigsqcap_{s \in\Gamma} A_{s})\ISA A_{s_0^1}.
    \]
\end{lemma}
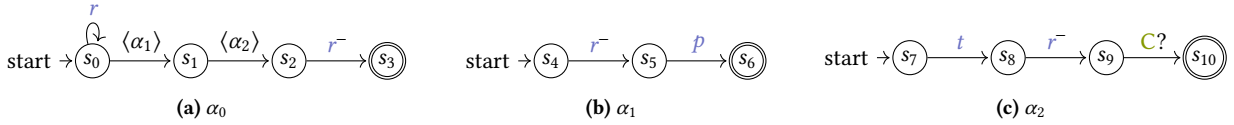
\begin{figure*}[h]
    \centering
    \begin{subfigure}{0.35\textwidth}
        \centering
            \begin{tikzpicture}[scale=0.5,node distance = 1.3cm,shorten >=1pt,state/.style={circle,draw,inner sep=1.5pt}]
            % 1. Wrap the automaton inside a local bounding box named 'automaton'
            \begin{scope}[local bounding box=automaton]
                \node[state, initial left] (s0) {${s_0}$};
                \node[state, right of=s0](s1){$s_1$};
                \node[state, right of=s1](s2){$s_2$};
                \node[state, accepting, right of=s2] (s3) {$s_3$};
                \path[->]   (s0)  edge[above] node{$\nesting{\alpha_1}$} (s1)
                            (s0)  edge[loop above] node{$\role{r}$} (s0)
                            %(s0)  edge[loop below, densely dotted] node{$\role{t}$} (s)
                            %(s0)  edge[below, bend right=50, densely dotted] node{$\concept{A}?$} (s1)
                            (s1) edge[above] node{$\nesting{\alpha_2}$} (s2)
                            (s2) edge[above] node{$\role{r}^-$} (s3);
            \end{scope}
            % 2. Place the name \alpha_1 to the left of the bounding box
            %\node[left=0.1cm of automaton.north west] {\LARGE $\alpha_1$:};
            %\node[right=0.1cm of automaton.north east] {\phantom{\LARGE $\alpha_1$:}};
        \end{tikzpicture}   
        \caption{$\alpha_0$}
    \end{subfigure}    
    \begin{subfigure}{0.25\textwidth}
        \centering
        \begin{tikzpicture}[scale=0.5,node distance = 1.3cm,shorten >=1pt,state/.style={circle,draw,inner sep=1.5pt}]
        \begin{scope}[local bounding box=automaton]
        \node[state, initial left] (s4) {$s_4$};
        \node[state, right of=s4] (s5) {$s_5$};
        \node[state, accepting, right of=s5] (s6) {$s_6$};
        \path[->]   (s4)  edge[above] node{$\role{r}^-$} (s5)
                    (s5)  edge[above] node{$\role{p}$} (s6);
        \end{scope}
    % 2. Place the name \alpha_1 to the left of the bounding box
    %\node[left=0.1cm of automaton.north west] {\LARGE $\alpha_2$:};
    %\node[right=0.1cm of automaton.north east] {\phantom{\LARGE $\alpha_1$:}};
    \end{tikzpicture}
    \caption{$\alpha_1$}
    \end{subfigure}
    \begin{subfigure}{0.35\textwidth}
        \centering
        \begin{tikzpicture}[scale=0.5,node distance = 1.3cm,shorten >=1pt,state/.style={circle,draw,inner sep=1.5pt}]
            \begin{scope}[local bounding box=automaton]
            \node[state, initial left] (s7) {$s_7$};
            \node[state, right of=s7] (s8) {$s_8$};
            \node[state, right of=s8] (s9) {$s_9$};
            \node[state, accepting, right of=s9] (s10) {$s_{10}$};
            \path[->]   (s7) edge[above] node{$\role{t}$} (s8)
                        (s8) edge[above] node{$\role{r}^-$} (s9)
                        (s9) edge[above] node{$\concept{C}?$} (s10);
            \end{scope}
            % 2. Place the name \alpha_1 to the left of the bounding box
            %\node[left=0.1cm of automaton.north west] {\LARGE $\alpha_2$:};
            %\node[right=0.1cm of automaton.north east] {\phantom{\LARGE $\alpha_1$:}};
        \end{tikzpicture}
        \caption{$\alpha_2$}
        \label{subfig:alpha2}
    \end{subfigure}
    \Description{A figure showing three automata.}
    \caption{Example N2RPQ with nesting level $k=1$.}
    \label{fig:exampleN2RPE}
\end{figure*}
\begin{example}\label[example]{ex:skipRel}
    To illustrate \textsf{SkipRel}, we consider an example n-NFA $\alpha_0$, given in \Cref{fig:exampleN2RPE}. We consider a 
        % $(\alpha_0,\T)$ let us consider an $\alpha_0$ from \Cref{fig:exampleN2RPE} and 
    TBox $\T\!=\{\concept{A}\ISA \exists \role{r}, \roleMath{r^-}\!\!~\ISA~\role{t},$ \\ $\concept{D}\ISA\exists \role{t'}.\top,\exists \role{t'}.\top\ISA \concept{C},\{\property{p} \ge 10?\}\ISA \concept{D}\}$. Then, the following holds: 
    % \textsf{SkipRel}$(\alpha_0,\T)$ is a set containing the following tuples 
    \begin{align*}
       \textsf{SkipRel}(\alpha_0,\T)\! =\!  \big \{ (\concept{A},s_0,s_3,\{s_5,s_8\}), (\concept{A},s_0,s_3,\{s_6,s_9\}), (\concept{A},s_5,s_6,\emptyset)\big \}
    \end{align*}
    To see why, for example $ (\concept{A},s_0,s_3,\{s_5,s_8\}) $ is in the set, we note that in the run of $\alpha_0$ starting from $s_0$, we must have a node $a$ in the graph that has an outgoing \role{r} edge and by $\roleMath{r^-}\ISA \role{t}$ also an incoming $\role{t}$. In state $s_5$ in $\alpha_1$ we previously followed a $t$ edge and in state $s_8$ in $\alpha_2$ an $\roleMath{r^-}$ edge. Thus, the run can just follow that same edge in reverse, landing at $a$ again.
\end{example}

% \todo[inline]{Sufficient Conditions for Polynomial Rewriting (see also Tree Witness Rewriting \cite{Kikot2012,Rodriguez2013})}

\mo{Deciding $\T \models C \ISA D$
is an instance of a well-studied problem: subsumption in $\mathcal{ELHI}$. It is \textsc{Exp}-complete \cite{KrotzschRH13},
but efficient off-the-shelf reasoners are available, \cite{DBLP:conf/kr/BateMGSH16, DBLP:journals/ws/SteigmillerLG14,DBLP:conf/ijcai/Kazakov09}. 
However, there are exponentially many possible candidates to populate the set \textsf{SkipRel}$(\alpha_0,\T)$, and testing them all (as assumed in \cite{Bienvenu2014}) would require strictly exponentially many calls to the external reasoner, which is clearly impractical. 
Fortunately, a simple pruning of the search space can discard many irrelevant candidates. First, 
we can restrict our attention to tuples that are strictly anonymous, as otherwise the partial runs do not use the anonymous part and are already present in $\mathcal{G}_{\T,\A}$.
For the strictly anonymous tuples, we make a simple relevance test: there must be some edge in the canonical model that allows the query to enter the anonymous part and exit it again.}
% we now provide formal conditions to prune the possible search space by discarding irrelevant candidates i.e. those that can never actually end up in the set based on the condition in~\Cref{lem:matJT}. }

\begin{definition}%[Candidates for \textsf{SkipRel}]
\label[definition]{def:candidates}
Let $\T$ be a $\ontoLang$ TBox and $\alpha_0$ an n-NFA of depth $k$, and let $\alpha_1, \dots, \alpha_n$ be the n-NFAs of nesting depth $k$ mentioned in the transitions of $\alpha_0$. 
    A tuple $(C,s_0^1,s_0^2,\Gamma)$ where $C\in\conceptnames\cup\{\top\}$, $s_0^1,s_0^2 \in Q_0$ %%states from $\alpha_i$ 
    and $\Gamma\subseteq \bigcup_{j}Q_j$, with $1 \leq i,j \leq n$, is called \emph{relevant} if $\{s_0^1\} \cup \Gamma$ does not contain any final states and the tuple satisfies all the following conditions (with $`\_`$ we denote an arbitrary state). 
    \begin{compactenum}[(1)]
%%        \item $s_0^1,s_0^2$ are states from $\alpha_i$ and $\{s_1,\dots,s_n\}\subseteq \bigcup_{\alpha_j}S_j$
        \item \label{item:defCand1} $B \ISA \exists r.\top \in \T$, $\transclosure{r}{r'}$, and $\transclosure{r}{t}$. 
        %%\todo{$r \ISA r' \in \T$ is not enough, we need $\transclosure{}{}$}
        %%such that $r,s,t\in\allroles$ 
        \item \label{item:defCand2} For some $0 \leq i \leq n$, we have $\{(\_,r',\_)\}\subseteq\delta_i$.
        \item \label{item:defCand3} If
        $s_0^2 \neq \mathsf{f}$, then
        for some $0 \leq i \leq n$, we have
        $\{(\_,t^-,\_)\}\subseteq\delta_i$. 
        %\todo{I weakened the conditions!} %%, and
        %%or (b) $s'\in F$.
        % \item for each $s \in\Gamma \cap Q_i$
        % %%${s_1,\dots,s_n\}$ and $s_j\in\alpha_j$ 
        % there is some $(\_,r_i^-,\_) \in \delta_i$ with $\transclosure{r}{r_i}$ 
    \end{compactenum}
\end{definition}

% \todo[inline]{the claim below is hard to make sense of. A ``candidate'' is something that is not already part of something. But here we use the notion to further restrict elements already in the set. This is bound to trip readers up.  Cem: does the renaming resolve your concern?}

\begin{lemma} \label{lemma:candidates}
    If $(C,s_0^1,s_0^2,\Gamma)\in \textsf{SkipRel}(\alpha_0,\T)$ and it is strictly anonymous, then it %%$(C,s_0^1,s_0^2,\Gamma)$ 
    is relevant. 
    %%a candidate for \textsf{SkipRel}$(\alpha_0,\T)$.
\end{lemma}
\begin{proof}
    Given an n-NFA $\alpha_0$, let $\alpha_1, \dots, \alpha_n$ be the n-NFAs of nesting depth $k$ mentioned in the transitions of $\alpha_0$.
    Consider $(C,s_0^1,s_0^2,\Gamma)\in \textsf{SkipRel}(\alpha_0,\T)$ with $\Gamma=\{s_1,\dots,s_n\}$, then we show that each of the conditions in \Cref{def:candidates} is satisfied. By \Cref{def:jumptrans} there is a partial run $(T,l)$ of $\alpha_0$ in the canonical model $\I$ of $\T$ and $\A=\{C(a)\}$ such that the following are satisfied. 
    \begin{compactitem}
        \item The root of $T$ is labelled $(a,s_0^1)$.
        \item If $s_0^2 \in Q_0$, there is a leaf node $v$ with $l(v)=(a,s_0^2)$.
        \item For every leaf node $v$ with $l(v)=(o,s)\neq(a,s_0^2)$, either $s \in F_j$ for some $0 \leq j \leq n$, 
        %%$\alpha_j\in\mathbf{A}^{k-1}$, 
        or $o=a$ and $s \in\Gamma$.
        \item $T$ has at least one node $(o,s)$ with $
        o\neq a$. 
    \end{compactitem}
    It is not hard to argue that the conditions of Definition~\ref{def:candidates} hold. 
    Since the construction of canonical model $\I$ is based on the ABox $\A=\{C(a)\}$ with $a$ being the only individual, if there is an individual $o\neq a$, it must have been added by rules  (Ch\ref{def:canonicalmodel:existsRHS}) or (Ch\ref{def:canonicalmodel:existsRHS2}) in \Cref{def:canonicalmodel}, which requires the existence of an  
    axiom $B \ISA \exists r.\top \in \T$, as desired, and it is triggered by $a \in B^{\I_{\T,\A}}$. 
    By construction, $(a,o) \in r^{\I_{\T,\A}}$. Moreover, for every $r' \in \allroles$, we have that 
    $(a,o) \in r'^{\I_{\T,\A}}$ implies $\transclosure{r}{r'}$, as
    (Ch\ref{def:canonicalmodel:RI}) is the only rule that can add role names for already existing pairs of nodes. 
    Moreover, the partial run can only visit $o$ after $a$ if it has a transition using such an $r'$. Hence, conditions (\ref{item:defCand1}) and (\ref{item:defCand2}) hold.  Finally,  if  $s_0^2 \neq \mathsf{f}$ then there is a leaf with label $(a,s_0^2)$, which means that the partial run returns to $a$ after $o$, which it can only do using some
$t^- \in \allroles$ that occurs in the transition function, and such that $(a,o) \in t^{\I_{\T,\A}}$, which as argued above,  implies $\transclosure{r}{t}$, ensuring that (\ref{item:defCand3}) holds.
    % \begin{compactenum}[(1)]
    %   %%  \item \emph{$s_0^1,s_0^2$ are states from $\alpha_i$ and $\{s_1,\dots,s_n\}\subseteq \bigcup_{j}Q_j$}: simply holds by \Cref{def:jumptrans}
    %     \item \emph{$\{A\ISA \exists r.\top,r\ISA r',r\ISA t\}\subseteq\T$ such that $r,r',t\in\allroles$}: 
    %     Note that w.l.o.g. it holds that $r\ISA r$ is always in $\T$.
    %     \item \emph{$\{(\_,r',\_),(\_,t^-,\_)\}\subseteq\delta_i$}: $s_0^1$ and $s_0^2$ are states of the same automaton $\alpha_i$; by definition of \textsf{SkipRel} the root of $T$ has label $(a,s_0^1)$ and there is a leaf node with label $(a,s_0^2)$; since the anonymous part of the canonical model $\I$ is a tree, and for the partial run to start and end in $a$, it has to pass an edge in both direction (unless depth of the run is $0$, i.e., labels of the run only consider individual $a$)
    %     \item \emph{for each $s_j\in\{s_1,\dots,s_n\}$ such that $s_j\in\alpha_j$ there is $(\_,r^-,\_)\in\delta_j$}: since $s_j\in\Gamma$, by definition of \textsf{SkipRel} there is a leave node with the label $(a,s_j)$; (similar argumentation as previous condition) partial run starts with $(a,s_0^1)$ and ends in $(a,s_j)$ and anonymous part is a tree --> have to pass edge in both directions
    %     \todo{finish (check case wit $\nesting{\alpha}\cdot r\cdot r^-$}
    % \end{compactenum}
\end{proof}

\mo{Each tuple $(C,s^1_0,s^2_0,\Gamma)$ in $\textsf{SkipRel}(\alpha_0,\T)$ tells us that if our data has a node $a$ that satisfies $C$, we can skip directly from $s^1_0$ to $s^2_0$, since a path that satisfies the transitions between them always exists in the canonical model, specifically in the anonymous part below $a$. Moreover, any nested automata that started but did not finish their runs in the anonymous part below $a$ will have to `exit' the anonymous part in some state $s$. 
These exit states $s$ are stored in $\Gamma$, and the corresponding nested automaton only needs to be executed from $s$ onwards to complete its run. 
We modify our input $\alpha_0$ to do precisely that: look for a $C$, jump from $s^1_0$ to $s^2_0$, and still look for partial runs that complete the nested automata, setting the states in $\Gamma$ as their initial states.

We introduce some notation.}
% \co{In our skipping procedure, before we can now finally show our skipping procedure that involves paths that traverse the anonymous part of the canonical model, we need some further notation to simplify the presentation.}
% % \todo{for the rewriting of a single JumpTrans tuple we first need a mechanism to "move" the initial state of an n-NFA}
% % \begin{definition}[Move initial state of n-NFA]
    For an n-NFA $\alpha={}$ $\langle Q, \mathbf{\Sigma}_{PG}, S, \delta,F \rangle$ of nesting depth $k$, we denote by $\alpha_{\mid s_0}$ the n-NFA $\langle Q, \mathbf{\Sigma}_{PG}, \{s_0\}, \delta,F \rangle$, that is, we make $s_0$ the single initial state. 
% \end{definition}

\begin{definition}[Skipping Anonymous]\label[definition]{def:skipAnonymous}
    Given a $\ontoLang$ TBox $\T$ and an n-NFA $\alpha_0=\langle Q_0, \mathbf{\Sigma}_{PG}, S_0, \delta_0,F_0 \rangle$,  let $\alpha_1, \dots, \alpha_n$ be the n-NFAs of nesting depth $k$ mentioned in the transitions of $\alpha_0$. 
Let $(C,s^1_0,s^2_0,\Gamma)\in\textsf{SkipRel}(\alpha_0,\T)$ with $\Gamma = \{s_1, \ldots s_n\}$. 
We associate a state to $s^2_0$ as follows: if $s^2_0 \in Q_0$, then $\mathit{state}(s^2_0) = s^2_0$, and 
if $s^2_0 = \mathsf{f}$, then $\mathit{state}(s^2_0)$ is a fresh state $s_\mathsf{f}$.  If $\Gamma = \emptyset$, we use $s'_1$ to denote $\mathit{state}(s^2_0)$, and for each  $s_j\in\Gamma$, we let $\alpha^j$ denote the $\alpha_i$ with $s_j \in Q_j$.    
Then $\skippingAnon{\alpha_0} = \langle Q'_0, \mathbf{\Sigma}_{PG}, S_0, \delta_0',F'_0 \rangle$ is the n-NFA  with: 
\begin{align*}
 Q'_0 = & Q_0 \cup \{\mathit{state}(s^2_0)\} \cup \{s'_j \mid s_j\in\Gamma \} \\    
 F'_0 = & F_0 \cup \{s_\mathsf{f}\} \text{~if~} s^2_0 = \mathsf{f} \text{~and~} F'_0 = F_0 \text{~otherwise, and}\\ 
  \delta'_0 =& \ \delta_0  \cup 
\{(s^1_0,C?,s'_1)\}\  \cup  \\
&\  \{(s'_1,\nesting{\alpha^1_{\mid s_1}},s'_2)\} \ \cup 
\cdots \cup \{(s'_{n-1},\nesting{\alpha^{n-1}_{\mid s_{n-1}}},s'_{n})\} \   \cup \\ 
&\  \{(s'_n,\nesting{\alpha^{n}_{\mid s_{n}}},s^2_0)\}
\end{align*}
\end{definition}
\noindent
In particular if $\Gamma = \emptyset$, the only new transition is 
$\{(s^1_0,C?,\mathit{state}(s^2_0))\}$.

% % Then we add states and transitions to $\alpha_0$ as follows: 
%     \begin{compactitem}
%         \item for each $s_j\in\Gamma$ add a fresh state $s'_j$ to $Q'_0$, \todo{I will finish fixing this def from home}
%         \item if $\Gamma=\emptyset$, then $\{(s^1_0,C,s^2_0)\}\subseteq\delta'_0$, and
%         \item if $\Gamma\neq\emptyset$, then $\{(s,C,s'_1)\}\cup\{(s'_i,\nesting{\alpha_{i\mid s_i}},s'_{i+1})\mid 1\le i< n\}\cup\{(s'_n,\nesting{\alpha_{n\mid s_n}},s')\}\subseteq\delta'_0$.
%     \end{compactitem}
% \end{definition}

%Recall that  $\mathcal{G}_{\T,\A}$ be $\I_{\T,\A}$ restricted to the nodes in $\A$. 
We can show that $\skippingAnon{\alpha}$ correctly skips over the anonymous part of $\I_{\T,\A}$ and finds all query answers when evaluated over the core $\mathcal{G}_{\T,\A}$ alone.

\begin{example} \bl{Consider the N2RPQ from \Cref{fig:exampleN2RPE} and the TBox $\T$ from \Cref{ex:skipRel}. Then, by $\skippingAnon{\alpha_0}$ and for the tuple $(\concept{A},s_0,s_3,\{s_5,s_8\})\in\textsf{SkipRel}(\alpha_0,\T)$ we add the following transitions and fresh states to $\alpha_0$.}
    \begin{center}
        \begin{tikzpicture}[scale=0.5,node distance = 1.5cm,shorten >=1pt,state/.style={circle,draw,inner sep=1.5pt}]
                % 1. Wrap the automaton inside a local bounding box named 'automaton'
            \begin{scope}[local bounding box=automaton]
                \node[state] (s0) {${s_0}$};
                \node[state, right of=s0](s5){$s'_5$};
                \node[state, right of=s5](s8){$s'_8$};
                \node[state, right of=s8] (s3) {$s_3$};
                \path[->]   (s0) edge[above] node{$\concept{A}$} (s5)
                            (s5) edge[above] node{$\nesting{\alpha_{1\mid s_5}}$} (s8)
                            %(s0) edge[loop above] node{$\role{r}$} (s0)
                            (s8) edge[above] node{$\nesting{\alpha_{2\mid s_8}}$} (s3);
            \end{scope}
                % 2. Place the name \alpha_1 to the left of the bounding box
                %\node[left=0.1cm of automaton.north west] {\LARGE $\alpha_1$:};
                %\node[right=0.1cm of automaton.north east] {\phantom{\LARGE $\alpha_1$:}};
        \end{tikzpicture}  
    \end{center} 
\end{example}

\begin{toappendix}
The idea and structure of the proof for \Cref{lemma:skipAnonCorrectness} is closely related to the proof of Proposition 6.6. in \cite{Bienvenu2014}. However, we use a different notation and make slightly different claims. For this reason, we provide the adapted proof below with references to the corresponding claims.
Before proving \Cref{lemma:skipAnonCorrectness}, we recall the following standard fact on canonical model homomorphism from the Description Logic literature. 
    
\begin{fact}\label{fact:model_isomorphism}
    Let $\T$ be a TBox and $\A$ an ABox. Then, suppose that two nodes $o_1$ and $o_2$ are such that $o_1\in A^{\I_{\T,\A}}$ iff $o_2\in A^{\J_{\T,\A}}$ for every concept name $A$. Then, the two submodels $\I_{\T,\A}$ rooted at $o_1$ and $\J_{\T,\A}$ rooted at $o_2$ are isomorphic.
\end{fact}
\end{toappendix}

\begin{lemmarep}\label{lemma:skipAnonCorrectness}
    Let $\T$ be a $\ontoLang$ TBox and $\alpha$ be an N2RPE. 
    Then, for every PG $\A$ and every pair of nodes $o_s,o_f$ from $\A$, it holds that $(o_s,o_f) \in \alpha^{\I_{\T,\A}}$ iff $(o_s,o_f) \in \skippingAnon{\alpha}^{\mathcal{G}_{\T,\A}}$.  
\end{lemmarep}

\begin{proofsketch}
The proof is contained in the proof of Proposition 6.6 in \cite{Bienvenu2014}; there is a one-to-one correspondence between the nested automata added by $\skippingAnon{\T}(\alpha_0)$ and the recursive calls 
to \textsf{EvalAtom} in cases (3.b.ii) and (3.b.iii). 
Only a minor adaptation is needed to argue that soundness and completeness are not compromised if the tests are restricted to strictly anonymous, relevant tuples in \textsf{SkipRel}.
The rest of the evaluation, which is done only over the individuals in $\A$, coincides exactly with the standard evaluation over $\mathcal{G}_{\T,\A}$.  
\end{proofsketch}

\begin{proof}
    Consider an $\ontoLang$ TBox $\T$, a PG $\A=(N,E,\textsf{label},\textsf{prop})$ and an N2RPE $\alpha=\langle Q,\Sigma_{PG},S,\delta,F\rangle$ of nesting depth $k$.\\
    % Completeness
    $(\Rightarrow)$ For the first direction, we show the claim by induction on the nesting depth $k$ of $\alpha$. In the following we prove that, whenever $(o_s,o_f)\in\alpha^{\I_{\T,\A}}$, then $(o_s,o_f) \in \skippingAnon{\alpha}^{\mathcal{G}_{\T,\A}}$ for every pair of nodes $o_s,o_f$ in $\A$. 
    From $(o_s,o_f)\in\alpha^{\I_{\T,\A}}$, it follows that there is a run $(T,l)$ for $\alpha$ on $\I_{\T,\A}$ such that $l(o_s,s_0)$ and $l(o_f,s_f)$ with $s_0\in S$ and $s_f\in F$.\\
    \textbf{Base Case.} Suppose that $k=0$, then $T$ consists of a single path with the sequence of labels $(o_0,s_0)(o_1,s_1)\dots(o_m,s_m)$ from root to leaf. Note that $o_0=o_s$, $o_m=o_f$ and $s_m=s_f$. We know from \Cref{def:runN2RPE} that for each $i\in \{1,\dots,m\}$ there is a transition $(s_{i-1},\sigma_i,s_i)\in\delta$ such that one of the following cases holds (here we make use of the fact that there are no transitions with nested NFAs of the form $\nesting{\alpha_i}$):
    \begin{compactitem}
        \item $\sigma_i\in 2^{\allroles\cup\datatests}$ (set of role labels and data tests) and either (1) or (2):
        \begin{compactenum}
            \item $(o_{i-1},o_i)\in(p\odot v?)^{\I_{\T,\A}}$ for each data test $p\odot v?\in\sigma_i^{\I_{\T,\A}}$
            \item $(o_{i-1},o_i)\in r^{\I_{\T,\A}}$ for each role $r\in\sigma_i^{\I_{\T,\A}}\cap\allroles$
        \end{compactenum}
        \item $\sigma_i=C?$ (concept test), $o_{i-1}=o_i$ and $o_i\in C^{\I_{\T,\A}}$
        \item $\sigma_i=p\odot v?$ (individual data test), $o_{i-1}=o_i$ and $o_i\in (p\odot v?)^{\I_{\T,\A}}$
    \end{compactitem}
    We can suppose without loss of generality that there is no $0\le i,i'\le m$ such that $i\neq i'$, $s_i=s_{i'}$, and $o_i=o_{i'}$ (this can always be ensured by deleting parts of the path). Now let $i_1 < \dots < i_p$ be all the indices $i$ such that $o_i\in N$ of $\A$ (note that $i_1=0$ and $i_p=m$). Intuitively, $i_1,\dots,i_p$ are the indices of the "projection" of the run obtained by restricting it to individuals occurring in the ABox. The following claim shows that either an individual of the run is in the core $\G_{\T,\A}$, or there must be a matching strictly anonymous tuple. This is a slightly stricter version of Claim 3 in \cite{Bienvenu2014}, whereby individuals must be in the core $\G_{\T,\A}$ and tuples must be strictly anonymous. 
    \begin{claim}[cf. Claim 3 in \cite{Bienvenu2014}]\label{claim:proof_completeness_basecase}
    For every $1<\ell\le p$, one of the following holds:
    \begin{compactenum}
        \item $(o_{i_{\ell-1}},o_{i_\ell})\in\sigma_{i_\ell}^{\G_{\T,\A}}$, or
        \item $o_{i_\ell}=o_{i_{\ell-1}}$ and there exists a strictly anonymous tuple $(C,s_{i_{\ell-1}},s_{i_\ell},\emptyset)\in \textsf{skipRel}(\alpha,\T)$ such that $o_{i_{\ell-1}}\in C^{\G_{\T,\A}}$.
    \end{compactenum}
    \end{claim}
    \textit{Proof of \Cref{claim:proof_completeness_basecase}.} For the case where $i_\ell=i_{\ell-1}+1$  we immediately obtain $(o_{i_{\ell-1}},o_{i_\ell})\in \sigma_{i_\ell}^{\G_{\T,\A}}$. Further, if $i_\ell\neq i_{\ell-1}+1$, then $o_{i_{\ell-1}+1}\not\in N$. From this we can infer that all the objects $o_{i_{\ell-1}+1},\dots,o_{i_{\ell-1}}$ must be descendants of $o_{i_\ell}=o_{i_{\ell-1}}$ in $\I_{\T,\A}$. Now, let $\J$ be the canonical model of $\T$ and $\{C(a)\}$. Then, by \Cref{fact:model_isomorphism} there is a homomorphism $h$ from $\I_{\T,\A}$ to $\J$ with $h(o_{i_{l-1}})=a$. Consider the labelled tree $(T',l')$ obtained from $(T,l)$ by (1) making the node labelled $(o_{i_{\ell-1}},s_{i_{\ell-1}})$ the root and the node labelled $(o_{i_\ell},s_{i_\ell})$ the unique leaf, and (2) replacing every node label $(o',s')$ by $(h(o',s')$. It holds that $(T',l')$ is a partial run of $\alpha$ on $\J$. Hence, by \Cref{def:jumptrans}, $(C,s_{i_{\ell-1}},s_{i_\ell},\emptyset)\in \textsf{skipRel}(\alpha,\T)$ and since $o_{i_{\ell-1}+1}\not\in N$ it follows that the tuple is strictly anonymous. \textit{(end proof of \Cref{claim:proof_completeness_basecase})} \\ 
    Now we can come back to the main claim and conclude the base case, we can simply construct a run $(T,l)$ for $\skippingAnon{\alpha}^{\mathcal{G}_{\T,\A}}$, by combining all $(o_{i_{\ell}},s_{i_{\ell}})$ from \Cref{claim:proof_completeness_basecase}. From the construction of $\skippingAnon{\alpha}$ in \Cref{def:skipAnonymous} it follows that $(o_s,o_f)\in\skippingAnon{\alpha}^{\mathcal{G}_{\T,\A}}$.\\
    \textbf{Induction Hypothesis.} Assume that if $(o_s,o_f)\in\alpha^{\I_{\T,\A}}$, then $(o_s,o_f)\in\skippingAnon{\alpha}^{\G_{\T,\A}}$ holds for $\alpha$ of nesting depth $\le k$. \\
    \textbf{Induction step.} 
    Let us now consider the case where $\alpha$ has nesting depth $k+1$. Let $v_0,\dots,v_m$ be the sequence of nodes in $T$ that begins with the root node labelled $(o_s,s_0)$ and ends with the unique leaf node labelled $(o_f,s_f)$. In the following we use $(o_\ell,s_\ell)$ for the label of node $n_\ell$. We can again assume w.l.o.g. that $(T,l)$ is minimal in the sense that there are no distinct positions $0\le i,i'\le m$ such that $n_i$ and $n_{i'}$ has the same label. Let $i_1<\dots < i_p$ be all of the indices $i$ such that $o_i\in N$ and $i_m=o_f\in N$.
    The intuition is the same as for the base case. However, we now also consider nested N2RPEs ($\Gamma$ is not empty), for which we can apply the induction hypothesis. As with Claim 6 in \cite{Bienvenu2014}, the difference here is that we restrict it to the core $\G_{\T,\A}$. Additionally, Claim 6 imposes a condition on the individuals in $\Gamma$ with regard to the evaluation function \textsf{EvalAtom}. However, this is unnecessary in our case, as the query engine handles the evaluation after the rewriting. 
    \begin{claim}[cf. Claim 6 in \cite{Bienvenu2014}]\label{claim:proof_completeness_inductionstep}
    For every $1<\ell\le p$, one of the following holds:
    \begin{compactenum}
        \item there is some $\alpha_j$ of nesting depth $\le k+1$ with $(s_{i_{\ell-1}},\sigma,s_{i_\ell})\in\delta_j$ such that $(o_{i_{\ell-1}},o_{i_\ell})\in\sigma_{i_\ell}^{\G_{\T,\A}}$, or 
        \item $o_{i_\ell}=o_{i_{\ell-1}}$ and there exists a strictly anonymous tuple $(C,s_{i_{\ell-1}},s_{i_\ell},\Gamma)\in\textsf{skipRel}(\alpha,\T)$ such that $o_{i_{\ell-1}}\in C^{\G_{\T,\A}}$
    \end{compactenum}
    \end{claim}
    \textit{Proof of \Cref{claim:proof_completeness_inductionstep}.} From \Cref{def:runN2RPE} there are three possible cases: 
    \begin{enumerate}[(a)]
        \item\label{claim6:caseUniqueChild} $n_{i_\ell}$ is the unique child of node $n_{i_{\ell-1}}$
        \item\label{claim6:caseOneOfTwoChildren} $n_{i_\ell}$ is one of two children of node $n_{i_{\ell-1}}$ (in case of a nested N2RPEs)
        \item\label{claim6:descendant} node $n_{i_\ell}$ is a descendant, but not a child of $n_{i_{\ell-1}}$ (when parts of the run traverse the anonymous part)
    \end{enumerate}
    First consider case (\ref{claim6:caseUniqueChild}), then by \Cref{def:runN2RPE} there must be a transition $(s_{i_{\ell-1}},\sigma,s_{i_\ell})\in\delta_j$ such that $(o_{i_{\ell-1}},o_{i_\ell})\in\sigma_{i_\ell}^{\G_{\T,\A}}$, which satisfies the first statement of the claim. 
    If case (\ref{claim6:caseOneOfTwoChildren}) holds, then by \Cref{def:runN2RPE} it follows that $(s_{i_{\ell-1}},\nesting{\alpha'},s_{i_\ell})\in\delta_j$, $o_{i_l}=o_{i_{\ell-1}}$, and the other child of $n_{i_{\ell-1}}$ has label $(o_{i_{\ell-1}},u)$, with $u$ the initial state of $\alpha'$. The labelled subtree of $(T,l)$ that is rooted at $(o_{i_{\ell-1}},u)$ is a run for $\alpha'$ on $\I_{\T,\A}$. Since the nesting depth of $\alpha'$ is $\le k$, we can apply the induction hypothesis. Thus, the first statement holds, i.e., $(s_{i_{\ell-1}},\sigma,s_{i_\ell})\in\delta_j$ such that $(o_{i_{\ell-1}},u)\in\sigma^{\G_{\T,\A}}$. 
    Finally, consider case (\ref{claim6:descendant}) where $n_{i_\ell}$ is not a child of $n_{i_{\ell-1}}$, which means that $o_{i_{\ell-1}+1}\notin N$ of $\A$. From this we can infer that all nodes $o_{i_{\ell-1}+1},\dots,o_{i_{\ell-1}}$ must be descendants of $o_{i_{\ell-1}}$ in $\I_{\T,\A}$, and $o_{i_l}=o_{i_{\ell-1}}$. Let $(T',l')$ be the partial run obtained from $(T,l)$ by 
    \begin{enumerate}[1.]
        \item making $n_{i_{\ell-1}}$ the new root node,
        \item making $v$ a leaf node by dropping all children, for every descendant $v$ of $n_{i_{\ell-1}}$ such that $l(v)=(o_{i_{\ell-1}},s)$ and such that there is no $v'$ (different from $v$ and $n_{i_{\ell-1}}$) with $l(v')=(o_{i_{\ell-1}},s')$ and occurring along the path from $n_{i_{\ell-1}}$ to $v$.
    \end{enumerate}
    Note that since $(T,l)$ is a run, there are three types of leaf nodes in $(T',l')$:
    \begin{itemize}
        \item node $n_{i_{\ell}}$ with label $(o_{i_{\ell-1}},s_{i_\ell})$ (on the main path)
        \item nodes with labels of the form $(o_{i_{\ell-1}},s)$
        \item nodes with labels of the form $(o,s_f)$ with $s_f\in F_j$ for some $\alpha_j$
    \end{itemize}
     (Observe that every label $(o,s)$ that appears in $(T',l')$ is an individual in the subtree in the anonymous part of $\I_{\T,\A}$ rooted at $o_{i_{\ell-1}}$.) Now let $\J$ be the canonical model of $\T$ and $C(d)$, then it follows from \Cref{fact:model_isomorphism} there exists a homomorphism $h$ between $\I_{\T,\A}$ to $\J$ with $h(o_{i_{l-1}})=d$. Consider the labelled tree $(T'',l'')$ obtained from $(T',l')$ by replacing every node label $(o,s)$ by $(h(o),s)$. Using the fact that $h$ is a homomorphism with $h(o_{i_{\ell-1}})=d$ and the above description of leaf nodes in $(T',l')$, one can show that $(T'',l'')$ is a partial run of $\alpha$ on $\J$ that satisfies:
    \begin{compactitem}
        \item the root of $T''$ is labelled $(d,s_{i_{\ell-1}})$
        \item there is a leaf node labelled $(d,s_{i_\ell})$
        \item for every leaf node $v$ with $l(v)=(o,s)\neq(d,s_{i_\ell})$, either $s\in F_j$ or $o=d$
    \end{compactitem}
    Since this partial run satisfies the conditions in \Cref{def:jumptrans} the tuple $(C,s_{i_{\ell-1}},s_{i_\ell},\Gamma)$ belongs to \textsf{skipRel}$(\alpha,\T)$, if we let $\Gamma$ contain all those states $u$ such that there is a leaf node $(d,u)$ with $u\neq s_{i_l}$. Since we have that $o_{i_{l-1}+1}\not\in N$, the tuple is strictly anonymous and the second statement of the claim holds. \textit{(end of proof \Cref{claim:proof_completeness_inductionstep})} \\
    Now, to conclude the proof for the first direction of the main claim, we construct a run $(T,l)$ for $\skippingAnon{\alpha}$ in $\G_{\T,\A}$ in a similar fashion as for the base case by combining all labels $(o_{i_\ell},s_{i_\ell})$ for $1<\ell<p$. For tuples $(C,s_{i_{\ell-1}},s_{i_\ell},\Gamma)\in\textsf{skipRel}(\alpha,\T)$ the run contains labels $(o_{i_\ell},u)$ for each $u\in\Gamma$. 
    % \todo{Cem: I fail to see how the IH can be applied here, how can we bound the nesting depth of the $\alpha^u$?}
    Since $u$ is a state from an NFA of depth $\le k$ we can apply the induction hypothesis, i.e.,  there is a run $(T_u,l_u)$ for $\alpha^u$. By \Cref{def:skipAnonymous} we can verify that $(T,l)$ is a valid run for $\skippingAnon{\alpha}$, thus $(o_s,o_f)\in\skippingAnon{\alpha}^{\G_{\T,\A}}$ holds. \\
    
    % Soundness
    ($\Leftarrow$) The second direction we prove by induction on the nesting depth of $\skippingAnon{\alpha}$. Suppose that $(o_s,o_f)\in\skippingAnon{\alpha}^{\G_{\T,\A}}$ for $\skippingAnon{\alpha}=\langle Q^\textsf{skip},\Sigma_{PG},S^\textsf{skip},\delta^\textsf{skip},F^\textsf{skip}\rangle$, then we prove that $(o_s,o_f)\in \alpha^{\I_{\T,\A}}$. 
    From $(o_s,o_f)\in\skippingAnon{\alpha}^{\G_{\T,\A}}$ it follows that there is a run $(T,l)$ for $\skippingAnon{\alpha}$ on $\G_{\T,\A}$ such that $l(o_s,s_0)$ and $l(o_f,s_f)$ with $s\in S^\textsf{skip}$ and $s_f\in F^\textsf{skip}$. 

    \textbf{Base Case.} Consider the nesting depth of $\skippingAnon{\alpha}$ to be $k=0$. Then, $T$ has only a single path with the sequence of labels $(o_0,s_0)(o_1,s_1)\dots(o_m,s_m)$ from root $o_0=o_s$ to leaf $o_m=o_f$. It suffices to show the following claim concerning a subsequence of the labels, which we derive from Claim 2 in \cite{Bienvenu2014}. Since our definition of a run requires the root to be labelled with an initial state, we show the claim for $\alpha^{\I_{\T,\A}}_{\mid s_i}$, which modifies the NFA $\alpha$ such that $s_i$ is the only initial state. 

    \begin{claim}[cf. Claim 2 in \cite{Bienvenu2014}]\label{claim:proof_soundness_basecase}
        For every $0\le i\le m$, it holds that $(o_i,o_m)\in\alpha^{\I_{\T,\A}}_{\mid s_i}$.
    \end{claim} 
    \textit{Proof of \Cref{claim:proof_soundness_basecase}.} The proof is by induction on $i$. The \textbf{base case} is when $i=m$, in which case $s_i=s_m$ and the claim trivially holds since $s_m$ is a final state. Now suppose that the claim holds for $i$ such that $g< i \le m$ (\textbf{induction hypothesis}). Consider the induction step with $i-1=g$, we know by \Cref{def:skipAnonymous} that $(o_{g},s_{g})$ must satisfy one of the following conditions: (i) there is a transition $(s_{g},\sigma,s_{i})\in\delta$, or (ii) there is a transition $(s_g,C,s_{i})$ such that there is a tuple $(C,s_g,s_{i},\emptyset)\in\textsf{skipRel}(\alpha,\T)$. If (i) holds, then $(o_g,o_{i})\in\sigma^{\I_{\T,\A}}$ by construction of the canonical model \Cref{def:canonicalmodel}. 
    %\todo[inline]{our definition of run requires to have the root labelled with an initial state (not so in BCOS14 paper); but then with our definition the statement below is not completely right}
    By induction hypothesis there is a $(T_i,l_i)$ run, which is a $(o_i,o_m)$-run for $\alpha_{\mid s_i}$ on $\I_{\T,\A}$. Now, let $(T',l')$ be the labelled tree obtained by creating a new node labelled $(o_{g},s_{g})$ and making it the leaf of $(T',l')$. It is easy to verify that $(T',l')$ is a run of $\alpha_{\mid s_g}$ on $\I_{\T,\A}$. Next consider the case (ii) where $(C,s_g,s_{i},\emptyset)\in\textsf{skipRel}(\alpha,\T)$, from  which follows that $o_g=o_{i}$ and $o_g\in C^{\I_{\T,\A}}$. From \Cref{def:jumptrans} we can infer that there is a partial run $(T_C,l_C)$ of $\alpha$ in the canonical model of $(\T,\{C(a)\})$ that satisfies the conditions below: 
    \begin{itemize}
        \item the root of $T_C$ is labelled $(a,s_g)$
        \item there is a unique leaf node $v$ with $l_C(v)=(a,s_{i})$
    \end{itemize}
    Since $o_g\in C^{\I_{\T,\A}}$ it follows from \Cref{fact:model_isomorphism} that there is a homomorphism $h$ from the canonical model of $(\T,\{C(a)\})$ to $\I_{\T,\A}$ with $h(a)=o_g$. We obtain $(T'_C,l'_C)$ by replacing every label $(o,s)$ in $T_C$ by $(h(o),s)$. Using the fact that $h$ is a homomorphism, one can show that $(T'_C,l'_C)$ defines a partial run of $\alpha$ in $\I_{\T,\A}$ that satisfies the following conditions:
    \begin{itemize}
        \item the root of $T'_C$ is labelled $(o_g,s_g)$
        \item there is a unique leaf node $v$ with $l'_C(v)=(o,s_{i})$
    \end{itemize}
    Let $(T,l)$ be the labelled tree obtained from $(T'_C,l'_C)$ by replacing the unique leaf node labelled $(o_g,s_{i})$ by the tree $(T_{i},l_{i})$. It follows from the components that $(T,l)$ is a $(o_g,o_m)$-run of $\alpha_{\mid s_g}$ on $\I_{\T,\A}$.
    \textit{(end proof of \Cref{claim:proof_soundness_basecase})}\\
    Since $s_0$ is an initial state of $\alpha$, the statement can be generalized and it follows that $(o_0,o_m)\in\alpha^{\I_{\T,\A}}$. \\
    \textbf{Induction Hypothesis.} Assume that if $(o_s,o_f)\in\skippingAnon{\alpha}^{\G_{\T,\A}}$, then $(o_s,o_f)\in \alpha^{\I_{\T,\A}}$ holds for $\skippingAnon{\alpha}$ of nesting depth $k$.\\
    \textbf{Induction Step.} Suppose that $(o_0,s_0)(o_1,s_1)\dots(o_m,s_m)$ is the sequence of the run $(o_s,o_f)\in\alpha^{\G_{\T,\A}}$ such that $o_s=o_0$, $o_f=o_m$ and $s_0,\dots,s_m\in Q$, i.e., states of the automaton with the highest nesting depth. 
    
    \begin{claim}[cf. Claim 4 in \cite{Bienvenu2014}]\label{claim:proof_soundness_inductionstep}
        For every $0\le i\le m$, it holds that $(o_i,o_m)\in\alpha^{\I_{\T,\A}}_{\mid s_i}$.
    \end{claim} 
    \textit{Proof of \Cref{claim:proof_soundness_inductionstep}.} The proof is by induction on $i$. The base case is when $i=m$, in which we consider a tree with a single node $i=m$ and the claim is trivially satisfies since $s_m$ is a final state. Now suppose that the claim holds for $i$ such that $g<i\le m$. Consider the case where $i-1=g$, then the pair $(o_{i},s_{i})$ satisfies one of the conditions, either: (i) there is a transition $(s_g,\sigma,s_{i})\in\delta$ or (ii) there is a tuple $(C,s_g,s_{i},\{s^\Gamma_1,\dots,s^\Gamma_n\})\in\textsf{skipRel}(\alpha,\T)$ and there are corresponding transitions $(s_g,C,s'_1),(s'_1,\nesting{\alpha^1_{\mid s^\Gamma_1}},s'_2),\dots, (s'_n,\nesting{\alpha^n_{\mid s^\Gamma_n}})$ in $\skippingAnon{\alpha}$. 
    If (i) holds we can immediately construct a $(o_g,o_m)$-run of $\alpha_{\mid s_g}$ on $\I_{\T,\A}$ in the same manner as in the base case. Next assume that case (ii) holds for $(o_{g},s_{g})$. Then, $o_g=o_{i}$ and there exists a tuple $(C,s_g,s_{i},\Gamma)$ such that $o_g\in C^{\I_{\T,\A}}$. By \Cref{def:jumptrans} the existence of $(C,s_g,s_{i},\Gamma)$ implies that there is a partial run $(T_C,l_C)$ of $\alpha$ in the canonical model of $(\T,\{C(a)\})$ that satisfies the following conditions: 
    \begin{itemize}
        \item the root of $T_C$ has label $(a,s_g)$
        \item there is a leaf node $v$ with $l(v)=(a,s_{i})$
        \item for every leaf node $v$ with $l_C(v)=(o,s)\neq (a,s_{i})$, either $s\in F_j$, or $o=a$ and $s\in\Gamma$
    \end{itemize}
    Since $o_g\in C^{\I_{\T,\A}}$, it follows from \Cref{fact:model_isomorphism} that there is a homomorphism $h$ from the canonical model of $(T,\{C(a)\})$ to $\I_{\T,\A}$ with $h(a)=o_g$. Let $(T',l')$ be obtained by replacing every label $(o,s)$ in $T_C$ by $(h(o),s)$. Since $h$ is a homomorphism, one can show that $(T'_C,l'_C)$ defines a partial run of $\alpha$ in $\I_{\T,\A}$ that satisfies the following conditions: 
    \begin{itemize}
        \item the root of $T'_C$ is labelled $(o_g,s_g)$
        \item there is a leaf node $v$ with $l'_C(v)=(o_g,s_{i})$
        \item for every leaf node $v$ with $l'_C(v)=(o,s)\neq(o_l,s_{i})$, either $s\in F_j$, or $o=o_g$ and $s\in\Gamma$.
    \end{itemize}
    Next we consider the states in $\Gamma$, by induction hypothesis we can find for each $s_j\in\Gamma$ a run $(T^\Gamma_{s_j},l^\Gamma_{s_j})$ for $\alpha^j_{\mid s_j}$ on $\I_{\T,\A}$. Finally, we can conclude the claim by applying the induction hypothesis and the fact that $c_g=c_i$. Thus, we can construct a run $(T_g,l_{g})$ such that $(o_g,o_m)\in\alpha^{\I_{\T,\A}}$ from $(T'_C,l'_C)$: 
    \begin{itemize}
        \item replace (unique) leaf node labelled $(o_g,s_{i})$ by the tree $(T_{i},l_{i})$
        \item replace each leaf node labelled $(o_g,s_j)$ for $s_j\in\Gamma$ by the tree $(T^\Gamma_{s_j},l^\Gamma_{s^\Gamma})$
    \end{itemize}
    \textit{(end proof of \Cref{claim:proof_soundness_inductionstep})} \\
    From \Cref{claim:proof_soundness_inductionstep} it follows that $(o_s,o_f)\in\alpha^{\I_{\T,\A}}_{\mid o_s}$, since $o_s$ is an initial state of $\alpha$, we can conclude for the main claim that $(o_s,o_f)\in\alpha^{\I_{\T,\A}}$.
\end{proof}

\subsection{Rewriting over the Core}

\mo{We know that
$\skippingAnon{\alpha}$ finds all query answers when evaluated over $\mathcal{G}_{\T,\A}$, but we want to execute the rewritten query over $\A$ alone, that is, over the explicit facts only.  
Again, we achieve this by adding transitions that allow us
to look for explicit facts only and skip over their consequences. 
As usual for DL-Lite variants,
in the canonical model obtained from a $\ontoLang$ TBox,
every implicit fact  $\xi_\ell$
is implied by some other fact $\xi_{\ell-1}$, which is either explicit or is also implied by another fact $\xi_{\ell-2}$ and so on, until we reach some explicit fact $\xi_{0}$ that triggered the existence of $\xi_\ell$.
We add transitions to $\alpha_0$ that allow it to use any of these $\xi_i$ when looking for some $\xi_\ell$.}

\begin{definition}[Skipping Core]\label[definition]{def:skipCore}
Let $\T$ be a $\ontoLang$ TBox, 
and let $\T^*$ be the result of closing $\T$ under the following rule:\\
$(\star)$ If $\{C\ISA \exists r.\top, \exists r'.\top \ISA D\} \subseteq \T^*$ and $\transclosure{r}{r'}$, then  $C \ISA D \in \T^*$.\\
Given a N2RPE $\alpha$ with alphabet $\mathbf{\Sigma}_{PG}$, we denote by $\skipping{\alpha}$ 
the \nnfa~obtained by exhaustively adding transitions as follows. 
% In each rule, $\delta$, $s_i$, $s_j$ and $s_k$ denote the transition function and states of one (arbitrary but fixed) automata  that occurs (possibly nested) in $\alpha$.  as a directly or indirectly nested automata in $\alpha$ -- 
% by exhaustive application of the rules below, where $r,s\in\allroles$, $C,D\in\conceptnames$ and $T=\{t_1,\dots,t_k\}\subseteq\datatests$. 
In each rule, $\delta$, $s_i$, $s_j$ and $s_k$ denote the transition function and states of one (arbitrary but fixed) automaton that occurs (possibly nested) in $\alpha$, and as usual, $r,p\in\allroles$, $C,D\in\conceptnames$ and $T=\{t_1,\dots,t_k\}\subseteq\datatests$. 
    \begin{compactenum}[(R1)]
        \item \label{def:skipping:ruleRI} If \ $r\ISA {p}\in \T$ and $(s_i,S,s_{j})\in\delta$ with $p\in S$ (or $p^-\in S$), then add $(s_i,R,s_{j})$ to $\delta$ with $R=(S\setminus\{p\})\cup\{r\}$ (or $R=(S\setminus\{p^-\})\cup\{r^-\}$).
        \item \label{def:skipping:ruleDataTestR} If \ $T\ISA r\in\T$ and $(s_i,R,s_j)\in\delta$ with $r \in R$, add $(s_i,T,s_j)$ to $\delta$. 
        \item \label{def:skipping:ruleCI} If \ $C\ISA D\in \T^*$ and $(s_i,D?,s_{j})\in\delta$, then add $(s_i,C?,s_{j})$ to $\delta$.
        \item\label{def:skipping:ruleExistLHS} 
        If \ $\exists r.\top \ISA B\in \T^*$ and $(s_i,B?,s_j)\in\delta$, then add $(s_i,\nesting{\alpha_{r}},s_j)$ to $\delta$ for a fresh two-state NFA 
        $\alpha_{r}$ with a single transition $\delta_r(s_0^1,\{r\},s_f)$ between its  initial state $s_0^1$ and its  final state $s_f$.
        \item \label{def:skipping:ruleDataTestC} If \ $T\ISA D\in \T^*$ and $(s_i,D?,s_j)\in\delta$, then add, for each $t_\ell \in T $, a transition  $(s'_{\ell},t_\ell?,s'_{\ell+1})$ to $\delta$, where 
        $s_i = s'_1$, $s'_{k+1} = s_j$, and $s'_2,\ldots, s'_k$ are fresh states.
        % OLD RULE (should be taken care of by JumpTrans) \item \label{def:skipping:ruleUnary} if $(s_i, \sigma, s_j), (s_j, \langle \alpha \rangle, s_k) \in \delta$ where $\sigma \in \mathbf{\Sigma}_{PG}$ and $\sigma^- \in L(\alpha)$, then add $(s_i,\sigma,s_k)$  to $\delta$. 
    \end{compactenum}
\end{definition}

Intuitively, each rule application adds a transition that allows 
$\alpha$ to `skip' an implicit fact $\xi_{j}$ in $\I_{\T,\A}$ and use instead a fact $\xi_{i}$, $i < j$, that participated in its creation. 
We now illustrate $\skipping{\alpha}$: 

\begin{example}
    \bl{Let us again consider the N2RPQ $\alpha_2$ from \Cref{fig:exampleN2RPE} and the TBox $\T$ from \Cref{ex:skipRel}. Then, by $\skipping{\alpha_2}$ we add the dotted transitions shown below.}
    \begin{center}
    \begin{tikzpicture}[scale=0.5,node distance = 1.3cm,shorten >=1pt,state/.style={circle,draw,inner sep=1.5pt}]
            \begin{scope}[local bounding box=automaton]
            \node[state, initial left] (s7) {$s_7$};
            \node[state, right of=s7] (s8) {$s_8$};
            \node[state, right of=s8] (s9) {$s_9$};
            \node[state, accepting, right of=s9] (s10) {$s_{10}$};
            \path[->]   (s7) edge[above] node{$\role{t}$} (s8)
                        (s8) edge[above] node{$\role{r}^-$} (s9)
                        (s8) edge[above,bend left=50, densely dotted] node{$\role{t}$} (s9)
                        (s9) edge[above] node{$\concept{C}?$} (s10)
                        (s9) edge[above,bend left=50, densely dotted] node{$\{\property{p} \ge 10?\}$} (s10)
                        (s9) edge[below,bend right=50, densely dotted] node{$\concept{D}?$} (s10);
            \end{scope}
            % 2. Place the name \alpha_1 to the left of the bounding box
            %\node[left=0.1cm of automaton.north west] {\LARGE $\alpha_2$:};
            %\node[right=0.1cm of automaton.north east] {\phantom{\LARGE $\alpha_1$:}};
        \end{tikzpicture}
    \end{center}
\end{example}

We can now show that $\skipping{\alpha}(x,y)$,  when evaluated over $\A$, correctly accounts for the implicit facts in $\G_{\T,\A}$.

\begin{lemmarep}\label{lemma:skipCorrectness}
    Let $\T$ be a $\ontoLang$ TBox and $\alpha$ be an N2RPE. 
    Then, for every PG $\A$ and every pair of nodes $o_s,o_f$ from $\A$, it holds that $(o_s,o_f) \in \alpha^{\G_{\T,\A}}$ iff $(o_s,o_f) \in \skipping{\alpha}^{\A}$.  
\end{lemmarep}

\begin{proofsketch}
\newcommand{\tboxsc}{\T^*_\mathit{core}}
Let $\tboxsc$ be the restriction of $\T^*$ to the axioms that do not have existential concepts on the right, and for readability, for a given $\A$, we denote by $\A^*$ the canonical model $\I_{\tboxsc,\A}$. 

 One can prove in the usual way (see, e.g., \cite{Eiter2012}) that $\G_{\T,\A} = \A^*$.
We can thus show that, for every $\alpha$: 
\[ \alpha^{\A^*} = (\skipping{\alpha})^\A. \]
Given the one-to-one correspondence between the rules (Ch\ref{def:canonicalmodel:RI}) to  (Ch\ref{def:canonicalmodel:conceptProps}), which add facts to $\A^*$ and the rules (R\ref{def:skipping:ruleRI}) to (R\ref{def:skipping:ruleDataTestC}), which add transitions to $\skipping{\alpha}$, 
the equality is established by a simple rule-by-rule analysis.
%  For the  \emph{$\supseteq$} direction, we show rule-by-rule  that the transitions added by \Cref{def:skipCore} mimic the rules for building the canonical model in the opposite direction, that is \ldots 
%  If an N2RPE $\alpha_{i+1}$ is obtained by applying a rule in \Cref{def:skipCore} to $\alpha_i$, then 
%  $(o,o')\in\alpha_{i+1}^{\A^*}$ implies $(o,o')\in{\alpha_{i}}^{\A^*}$, hence $(o,o')\in\alpha_{i+1}^{\A^*}$ implies 
%  \todo{the notation is too messy here: the indexes are over the skip rules or the canonical model construction? }
%%%%
% For the \emph{(only if)} direction, we rely on the fact that
% the canonical model construction naturally induces an ordering on the facts in $\I_{\T,\A}$, where an implicit fact always has a strictly higher degree than the facts that participate in its creation. We then show that if 
% $(o_s,o_f)\in\alpha_{i}^{\I_{\T,\A}}$ and this can only be witnessed using some implicit fact $f$, then it is possible to apply some rule in such a way that, with the additional transition, $(o_s,o_f)\in\alpha_{i+1}^{\I_{\T,\A}}$ can be witnessed using instead a fact $f'$ of strictly lower degree than $f$. By applying the rules exhaustively, we eventually have that $(o_s,o_f)\in \skipping{\alpha}^{\I_{\T,\A}}$ is witnessed using only facts of degree $0$, that is, facts in $\A$, and hence $(o_s,o_f)\in \skipping{\alpha}^{\A}$.
\end{proofsketch}

\begin{proof} 
\newcommand{\tboxsc}{\T^*_\mathit{core}}
Let $\tboxsc$ be the restriction of $\T^*$ to the axioms that do not have existential concepts on the right, and for readability, for a given $\A$, we denote by $\A^*$ the canonical model $\I_{\tboxsc,\A}$. 
One can prove in the usual way (see, e.g., \cite{Eiter2012}) that $\G_{\T,\A} = \A^*$.
We can thus show that, for every $\alpha$: 
\[ \alpha^{\A^*} = (\skipping{\alpha})^\A. \]
Given the one-to-one correspondence between the rules (Ch\ref{def:canonicalmodel:RI}) to  (Ch\ref{def:canonicalmodel:conceptProps}), which add facts to $\A^*$ and the rules (R\ref{def:skipping:ruleRI}) to (R\ref{def:skipping:ruleDataTestC}), which add transitions to $\skipping{\alpha}$, 
the equality is established by a simple rule-by-rule analysis in both directions.

For the  \emph{($\Leftarrow$)} direction, assume that an \nnfa~$\alpha_{i+1}$ was obtained by applying one of the rules (R\ref{def:skipping:ruleRI}) to (R\ref{def:skipping:ruleDataTestC}) in \Cref{def:skipCore} to $\alpha_i$. In the following we show for each of the rules that it holds if $(o_s,o_f)\in\alpha_{i+1}^{\G_{\T,\A}}(x,y)$, then $(o_s,o_f)\in\alpha_{i}^{\G_{\T,\A}}(x,y)$. 
\begin{compactitem}[rightmargin=0pt]
    \item[(R\ref{def:skipping:ruleRI})] Assume that $\alpha_{i+1}^{\I_{\T,\A}}(x,y)$ was obtained by applying the rule for $r\ISA p\in \T$ and $(s_j,P,s_{j+1})\in\delta_{i}$ with $p\in P$ (or $p^-\in P)$. Then, $\alpha_{i+1}^{\G_{\T,\A}}(x,y)$ contains the transition $(s_j,R,s_{j+1})$ with $R=(P\setminus\{p\})\cup\{r\}$ (or $R=(P\setminus\{p^-\})\cup\{r^-\}$). We assume that $(o_j,o_{j+1})\in r^{\G_{\T,\A}}$ (or $(o_{j+1},o_j)\in r^{\G_{\T,\A}}$) and from \Cref{def:canonicalmodel} it follows that $(o_j,o_{j+1})\in p^{\G_{\T,\A}}$ (or $(o_{j+1},o_j)\in p^{\G_{\T,\A}}$). 
    \item[(R\ref{def:skipping:ruleDataTestR}),(R\ref{def:skipping:ruleCI})] The proof for these rules works equivalent to (R\ref{def:skipping:ruleRI}).
    \item[(R\ref{def:skipping:ruleExistLHS})] Assume that $\alpha_{i+1}^{\I_{\T,\A}}(x,y)$ was obtained by the rule for $\exists r.\top \ISA B$ and $(s_j,B?,s_{j+1})\in\delta_i$. Then, the function adds a transition with a nested NFA $\nesting{\alpha_r}$ that contains a single transition with $r$. By \Cref{def:runN2RPE} of a run and $\alpha_{i+1}$ it holds that there has to be an individual $o_j=o_{j+1}$ and a pair $(o_j,o')$, such that $(o_j,o')\in r^{\G_{\T,\A}}$. From \Cref{def:canonicalmodel} it follows that $o_j\in B^{\G_{\T,\A}}$ and thus $(o_j,o_{j+1})\in \alpha_i^{\G_{\T,\A}}$. 
    \item[(R\ref{def:skipping:ruleDataTestC})] In this case $\nNFA{i+1}$ was obtained based on $T\ISA D\in\T^*$ and $(s_j,D?,s_{j+1})\in\delta_i$. For each $t_\ell\in T$ there is a transition $(s'_\ell,t_\ell?,s'_{\ell+1})\in\alpha_{i+1}$, where $s_j=s'_1,s'_{k+1}=s_j$ and $s'_2,\dots,s'_k$ are fresh states, i.e., $s'_2,\dots,s'_k\notin S_i$.
    By \Cref{def:runN2RPE} and $(s'_\ell,t_\ell?,s'_{\ell+1})\in\alpha_{i+1}$, it follows that $o_j=o_{j+1}$ and $o_j\in T^{\G_{\T,\A}}$. Further, by the construction of the canonical model item (Ch\ref{def:canonicalmodel:conceptProps}) we can infer that $o_j\in D^{\G_{\T,\A}}$.
\end{compactitem}

($\Rightarrow$) In order to show completeness we first introduce the notion of degree, that comes from the canonical model construction: We denote individuals or pair of individuals in the extension of concepts and roles as \emph{facts}, i.e., $v\in B^{\G_j}$ or $(v,w)\in r^{\G_j}$. A fact that occurs for the first time in $\G_j$ has degree $j$, and a set of facts gets the sum of the degrees of its elements. 
The degree of sequence $o_s\sigma_1o_1\dots\sigma_no_n\sigma_fo_f$ witnessing $(o_s,o_f)\in\alpha^{\G_{\T,\A}}(x,y)$
is the minimal degree of a set of facts that is sufficient to witness all a run in \Cref{def:runN2RPE}, including the nested automata. 
Note that there may be more than one way to witness not only the automata, but even for a fixed (top level) sequence, we may have choices, but we don't care. We set the degree to be the smallest, and then every added transition will allow us to witness the same sequence but with a strictly smaller degree.
Having the notion of degree for a pair in place we show by induction that in each rewriting step we are "moving up" the part of the canonical model (decreasing the degree) that is induced by the axioms in $\T$. 
% For that consider two sequences $(o_0,o_n)\in\mathbf{\nNFA{i}}^{\I_{\T,\A}}$ and $(o'_0,o'_n)\in\nNFA{i+1}^{\I_{\T,\A}}$, such that $\nNFA{i+1}$ was obtained by applying the skipping function from \Cref{def:skipping} to $\mathbf{A_{i}}$. 
In the following we show that for each of the rules in \Cref{def:canonicalmodel} there is a corresponding rule in \Cref{def:skipAnonymous} producing a \nnfa~$\nNFA{i+1}$ from $\nNFA{i}$ such that the degree of the accepted sequence is strictly smaller. 
\begin{compactitem}[rightmargin=0pt]
    \item[(Ch\ref{def:canonicalmodel:RI})] 
    Consider $r \ISA p\in\T$ and $(o_k,o_{k+1})\in r^{\G_j}$, $(o_k,o_{k+1})\notin p^{\G_{j}}$ and $(o_k,o_{k+1})\in p^{\G_{j+1}}$ by construction of the canonical model. Consider the sequence $(o_s\sigma_1 o_1\dots o_k r o_{k+1}\dots o_n\sigma_fo_f )$ such that $(o_s,o_f)\in\nNFA{i}^{\G_{\T,\A}}$. In the rewriting we apply (R\ref{def:skipping:ruleRI}) to obtain $\nNFA{i+1}$ from $\nNFA{i}$, where we add for each transition $(s,P,s')$ with $p\in P$ (or $p^-\in P)$ a new transition $(s, R,s')$ such that $R=(P\setminus \{p\})\cup \{r\}$ (or $R=(P\setminus \{p^-\})\cup \{r^-\}$), thus $(o_1\dots o_k r o_{k+1}\dots o_n)$ will be accepted by $\nNFA{i+1}^{\G_{\T,\A}}$, for which the degree is strictly smaller then for $(o_s\sigma_1o_1\dots o_k p o_k\dots o_n\sigma_fo_f)$. 
    \item[(Ch\ref{def:canonicalmodel:roleProps}),(Ch\ref{def:canonicalmodel:CI})] For $T\ISA r'\in\T$ the skipping rule (R\ref{def:skipping:ruleDataTestR}) and for $A \ISA B\in\T$  (R\ref{def:skipping:ruleCI}) applies. In both cases the reasoning is similar to (Ch\ref{def:canonicalmodel:RI}).
    \item[(Ch\ref{def:canonicalmodel:existsLHS})] Given $\exists r.\top \ISA B\in\T$ and by construction of the canonical model some individual $o_k\in(\exists r.\top)^{\G_j}$, $o_k\notin B^{\G_{j}}$ $o_k\in B^{\I_{j+1}}$. 
    Consider a sequence of the form $(o_s\sigma_1o_1\dots o_k B o_k \dots o_n\sigma_fo_f)$ such that $(o_s,o_f)\in \nNFA{i}^{\G_{\T,\A}}$. 
    In skipping we apply rule (R\ref{def:skipping:ruleExistLHS}) to obtain $\nNFA{i+1}$ and add for each $(s,B?,s')\in\delta$ the transition $(s,\nesting{\alpha_{r}},s')$ to $\delta$, where $\nNFA{r}$ is a fresh two-state NFA with a single transition $(s_0^1,\{r\},s_f^1)$. Since $o_k\in(\exists r.\top)^{\G_j}$ there has to be some sequence $(o_k r o'_k)$, such that $(o_k,o'_k)\in\nNFA{r}^{\G_{\T,\A}}$ and we can construct a run with strictly smaller degree such that $(o_s,o_f)\in\nNFA{i+1}^{\G_{\T,\A}}$.
    \item[(Ch\ref{def:canonicalmodel:conceptProps})] In this case that $T\ISA B\in\T$, $o_k\in T^\G_j$ and $o_k\notin B^\G_j$ triggered the rule. Consider again the sequence $(o_s\sigma_1o_1\dots o_k B o_k \dots o_n\sigma_fo_f)$ with $(o_s,o_f)\in \nNFA{i}^{\G_{\T,\A}}$. Based on $T\ISA B\in\T, o_k\in T^\G_j$ the skipping function applies rule (R\ref{def:skipping:ruleDataTestC}). It follows that for $(s,B?,s')\in\delta$ we add for each $t_\ell\in T$ a transition $(s_\ell,t_\ell?,s_{\ell+1})$, where $s=s_1$ and $s_{j+1}=s'$, and $s'_2,\dots,s'_k$ are fresh states. 
    Since $o_k\in T^{\G_j}$, it holds that $o_k\in t_\ell^{\G_j}$ for each of the $t_\ell\in T$ and the degree of the sequence from $o_s$ to $o_f$ is strictly smaller.
    \item[(Ch\ref{def:canonicalmodel:existsRHS}),(Ch\ref{def:canonicalmodel:existsRHS2})] Since we only consider $\tboxsc$ these rules do not apply (and are covered by \Cref{def:skipAnonymous}).
\end{compactitem}
\end{proof}

Together, \Cref{lemma:skipAnonCorrectness} and \Cref{lemma:skipCorrectness} give us the desired rewriting.    
\begin{proposition}\label{lemma:skipTotalCorrectness}
    Let $\T$ be a $\ontoLang$ TBox and $\alpha$ be an N2RPE. 
    Then, for every PG $\A$ and every pair of nodes $o_s,o_f$ from $\A$, 
    it holds that $(o_s,o_f) \in \alpha^{\I_{\T,\A}}$ iff 
    $(o_s,o_f) \in \skipping{\skippingAnon{\alpha}}^{\A}$. 
\end{proposition}

\section{Rewriting join-on-free CN2RPQs}
\label{sec:rewrite_cof_CN2RPQ}

\co{Utilising the two skip procedures from \Cref{sec:rewrite_cof_N2RPQ}, we can now present} an algorithm for \co{the rewriting of} join-on-free CN2RPQs \co{under  $\ontoLang$ ontologies. 
We dub this procedure ``\algoNameLong'', or \algoName for short. We proceed to describe the key parts of the procedure.}

We use a function that inverts nested automata, that is, it transforms the automata to accept the same words, but inverted.

\newcommand{\nfainv}[1]{\bar{#1}}

\begin{definition}\label[definition]{def:inverse}
We define an \emph{inverse} function over the alphabet $\mathbf{\Sigma}_{PG} \cup \bigcup_{k \in \mathbb{N}}\mathbf{A}^k$: 
if  $\sigma \not\in 2^{\allroles \cup \datatests}$, 
then  $\sigma^- = \sigma$, 
and if $\sigma \in 2^{\allroles \cup \datatests}$, then 
\[ \sigma^- = \{ t \mid t\in  \sigma \cap \datatests  \} \cup \{ r^- \mid r \in \sigma \cap  \rolenames \} \cup \{ r \mid r^- \in \sigma, 
            r \in \rolenames\}. \] 
        
% \[ \sigma^- = \begin{cases}

%             \begin{aligned}
%             &\{ t \mid t\in  \sigma \cap \datatests  \} \cup \{ r^- \mid r \in \sigma \cap  \rolenames \} \\ 
%             &\cup \{ r \mid r^- \in \sigma
%             r \in \rolenames\},
%             \end{aligned}   & \text{ if } \sigma \in 2^{\allroles \cup \datatests}, \\
%             \sigma & \text{otherwise.}
% 		 \end{cases} \]
% The inverse $w^-$ of a word $w = \sigma_1 \cdots \sigma_n$ is 
% $\sigma^-_n \cdots \sigma^-_1$, and for a language $L$, we let $L^- = \{ w^- \mid w \in L\}$. 

Let $\alpha = \langle Q, \mathbf{\Sigma}_{PG}, S, \delta, F \rangle$ be an N2RPE with $S \subseteq Q$ the initial states and $F \subseteq Q$ are the final states. 
Then $\alpha$ can be \emph{inverted} to obtain 
$\nfainv{\alpha} = \langle Q, \mathbf{\Sigma}_{PG}, F, \delta^-,S \rangle$, where  $\delta^- = \{ (s_j, \sigma^-, s_i) \mid 
(s_i,\sigma,s_j) \in \delta \}$. 
\end{definition}

\begin{lemma}
    \label{prop:inverseWorks}
    For every N2RPE $\alpha$, PG $\I$, and pairs of nodes $n_1,n_2$, 
    $ (n_1,n_2) \in \alpha^\I \text{~if and only if~} (n_1,n_2) \in (\nfainv{\alpha})^\I$
   %% For every N2RPE $\alpha$, %%\( L(\nfainv{\alpha}) = L(\alpha)^- \). 
%%    For every N2RPE $\alpha$, \( L(\nfainv{\alpha}) = L(\alpha)^- \). 
\end{lemma} 

% \todo{definition of inverse automata missing}
The first step of our rewriting algorithm removes all non-answer variables from the query by making use of nested expressions.

\begin{definition}
For a join-on-free CN2RPQ  $q(\vec{x})$, we denote by 
$q^{\nexist}(\vec{x})$ the result of replacing each atom 
$\alpha(x,y)$ with $y\not\in \vec{x}$ by $\nesting{\alpha}(x)$, and  each $\alpha(x,y)$ with $x \not\in \vec{x}$
by $\nesting{\nfainv{\alpha}}(y)$.
\end{definition}

\noindent 
A simple inspection of the CN2RPQ semantics shows that the transformation preserves query answers.  
\begin{lemma} \label{lemma:removeExistentialVariables}
%%Let $\T$ be a TBox and
Let $q(\vec{x})$ be a join-on-free CN2RPQ. 
For every interpretation $\I$ and tuple of nodes $\vec{a}$, we have that $\vec{a}$ is an answer to $q(\vec{x})$ in $\I$ iff 
$\vec{a}$ is an answer to $q^{\nexist}(\vec{x})$ in $\I$.
\end{lemma}

\co{We emphasise again that} the transformation to $q^{\nexist}$ \co{only} removes  non-answer variables and leaves all other variables untouched. 
% \todo{drop this line if it seems superfluous. }
% The rewriting of queries can now be achieved by applying $\skipping{\alpha}$ \co{ and  $\skippingAnon{\alpha}$ } to each atom independently. 
We denote by 
\co{$\algoMath{q^{\nexist}(\vec{x})}$} the result of executing 
Algorithm \ref{algo:rewriteCN2RPQ} on $q(\vec{x})$ and $\T$. As we can see, it \co{ begins with the removal of non-answer variables using $q^{\nexist}$,} and then applies $\skipping{\alpha}$ \co{ and $\skippingAnon{\alpha}$} to each atom in  $q^{\nexist}(\vec{x})$. 

\co{Next we show that \Cref{algo:rewriteCN2RPQ} is indeed %%correct, meaning 
sound and complete.}

\begin{theorem}
%%[Correctness Soundness and Completeness]
    \label{thm:sound_and_correcetness}
    Let $\T$ be a $\ontoLang$ TBox and 
    $q(\vec{x})$ be a join-on-free CN2RPQ.  
    Then, for every ABox $\A$ and every tuple $\vec{a}$ of individuals  from $\A$, it holds that $\vec{a}$ is a certain answer to $q(\vec{x})$ over $(\T,\A)$ iff $\vec{a}$ is an answer to $\co{\algoMath{q^{\nexist}(\vec{x})}}$ over $\A$.
\end{theorem}

\begin{proof} Given a $\ontoLang$ TBox $\T$ and an ABox $\A$ and a join-on-free CN2RPQ $q(\vec{x})$.  
For ($\Rightarrow$) assume that $\vec{a}$ of individuals from $\A$ is a certain answer to $q(\vec{x})$ over $(\T,\A)$. From \Cref{lemma:removeExistentialVariables} it holds that $\vec{a}$ is a certain answer to $q^{\nexist}(\vec{x})$ over $(\T,\A)$, i.e., $\vec{a}=\mu(\vec{x})$ and $(\mu(x),\mu(y))\in\alpha^{\I_{\T,\A}}$ (or $(\mu(x)\in\alpha^{\I_{\T,\A}}$) for each atom $\alpha(x,y)$ (or $\alpha(x)$) in $q^{\nexist}(\vec{x})$. By \Cref{lemma:skipCorrectness} we can imply that \co{$(\mu(x),\mu(y))\in\algoMath{\alpha}^\A$} and \co{$\mu(x)\in\algoMath{\alpha}^\A$} respectively. Thus, the claim holds that $\vec{a}$ is an answer to \co{$\algoMath{q^{\nexist}(\vec{x})}$} over $\A$.
The proof for ($\Leftarrow$) is analogous, since \Cref{lemma:skipCorrectness} and \Cref{lemma:removeExistentialVariables} are bidirectional.
\end{proof}

% Cem: commenting this out since we moved to introduce inverse automata formally 
% To invert the nNFA $\mathbf{A}$ we take the inverse of each transition, and make initial states to be final and vice versa. 

% Note that we assume that the alphabet is closed under inverses, i.e., $\textit{inv}(r)=r^-$, $\textit{inv}(r^-) = r$ and $\textit{inv}(A)=A$, for $r\in\rolenames$ and $A\in \conceptnames\cup \datatests$.

\begin{algorithm}[t]
% \small
    \caption{
    % $\mathtt{rewrite\_jof\_CNRPQ}$ -- 
    \algoNameLong (\algoName)} 
    % \mbox{Rewriting join-on-free CN2RPQs with data tests}}
    \label{algo:rewriteCN2RPQ}
    \SetKwInOut{Input}{Input}
    \SetKwInOut{Output}{Output}
    \SetKwFunction{FskipAnonymous}{$\skippingAnon{\alpha}$}
    \SetKwFunction{FskipCore}{$\skipping{\alpha}$}
    \SetKwFunction{FremoveExistVar}{$\mathtt{remove\exists Var}$}    
    \SetKwFunction{rewriteAtomic}{$\mathtt{\MakeLowercase{\algoName}}$}
    \SetKwProg{Fn}{function}{:}{}

    \Input{\ join-on-free CN2RPQ $q(\vec{x}), \text{ TBox } \T$}
    \Output{\ CN2RPQ $q'$}
    \Fn{\rewriteAtomic{$q$}}{
        $q$=\FremoveExistVar{$q$} \\
    \While{$q\neq q'$}{
        $q'=q$ \\
        \ForEach{$\alpha(x,y)\in q$}{
            $q'=q'[\alpha\setminus$\FskipAnonymous$]$ \\
            $q'=q'[\alpha\setminus$\FskipCore$]$
        }
    }
    \Return $q$
    }
    \Fn{\FremoveExistVar{$q$}}{
        \label{algo:rewriteCN2RPQ:line:removeExistVar}  
        \ForEach{$\alpha(x,y)\in q$}{
            \If{$y\not\in\vec{x}$}{
                $q=q[\alpha(x,y)\setminus \nesting{\alpha}(x)]$
            }\uElseIf{$x\not\in\vec{x}$}{
                $q=q[\alpha(x,y)\setminus \nesting{\alpha^-}(y)]$
            } 
        }
        \Return $q$
    }
\end{algorithm}

\begin{example}\label[example]{ex:fullalgo} \bl{For the algorithm \algoName the only new step in this section is the removal of existential variables. For demonstration purposes consider the query $q(x,y):=\alpha_1(x,y),\alpha_2(x_1,y),\alpha_3(x,y_1)$. Then, the returned query of $\FremoveExistVar{q}$ is $q(x,y):=\alpha_1(x,y),\nesting{\alpha_2^-}(y),\nesting{\alpha_3}(x)$. Since for $\alpha_2(x_1,y)$ variable $y$ on the right side is bound, we have to invert $\alpha_2$. If we consider $\alpha_2$ to be from \Cref{subfig:alpha2}, then the following is the inverse $\alpha_2^-$.}
\begin{center}
    \begin{tikzpicture}[scale=0.5,node distance = 1.3cm,shorten >=1pt,state/.style={circle,draw,inner sep=1.5pt}]
            \begin{scope}[local bounding box=automaton]
            \node[state, accepting] (s7) {$s_7$};
            \node[state, right of=s7] (s8) {$s_8$};
            \node[state, right of=s8] (s9) {$s_9$};
            \node[state, initial right, right of=s9] (s10) {$s_{10}$};
            \path[->]   (s8) edge[above] node{$\role{t}^-$} (s7)
                        (s9) edge[above] node{$\role{r}$} (s8)
                        (s10) edge[above] node{$\concept{C}?$} (s9);
            \end{scope}
            % 2. Place the name \alpha_1 to the left of the bounding box
            %\node[left=0.1cm of automaton.north west] {\LARGE $\alpha_2$:};
            %\node[right=0.1cm of automaton.north east] {\phantom{\LARGE $\alpha_1$:}};
        \end{tikzpicture}
\end{center}
\end{example}
\section{Implementation \& Experiments}\label{sec:implementation}
% \todo[inline]{R1.1. The high timeout rate during query execution on Neo4j is a critical issue. Please expand the discussion on why these timeouts occur (e.g., Cartesian products, large result sets) and outline potential optimizations to make the evaluation step more tractable in practice. What should the user do or what are your next steps to make it a feasable solution? \\
% R1.3. In Section 5, please clarify the methodology behind the 5 handcrafted queries for the MMM dataset. Why were these queries chosen, and what specific features do they test? \\
% R1.4. The visual representation of the timeouts in Figure 5 is somewhat unclear. Adjusting the charts to distinctly separate or highlight successful executions versus failed/timeout executions would improve readability. \\
% R3.2. In addition to the complexity analysis, I would be interested in a discussion on scalability with respect to TBox size/complexity. The DBPedia and MMM comparison is presented somewhat anecdotally rather than as a systematic complexity-driven analysis. \\
% R3.4. The experimental design mixes two different things: the rewriting time (which looks practically encouraging, under 600ms even for MMM) and evaluation time (which is poor, with frequent timeouts). The paper's own framing ("the generated queries with Neo4j still seem to be a challenge") undersells how central this weakness is to the paper's practicality claims. The synthetic query generation procedure (random concept/role selection with bounded path length) is also not validated against realistic query distributions.
% }

\na{In this section, our aim is to show  the practical utility of our rewriting algorithm \algoName(cf.~\Cref{sec:rewrite_cof_CN2RPQ}). To this end, we implemented a prototype of our rewriting algorithm in Java. This prototype takes as input a navigational query (join-on-free CN2PRQ), using a custom format, and rewrites it with respect to a $\ontoLang$ ontology (our DL-Lite extension supporting tests over property values) expressed in the OWL format. As part of the experiments, we then run the produced Cypher queries over a Neo4j database and report on the results.}
\CR{The source code of our prototype\footnote{\url{https://gitlab.com/austrian-neurocloud/software/owl2cypher/-/releases/CIKM2026}} is publicly available}.

\paragraph{Setup and Implementation}

\na{All experiments were executed on a virtual machine with 400\,GB RAM and 16 cores running Debian 12. The host machine is equipped with an AMD EPYC 7313 
%16-core 
processor clocked at 2.00\,GHz, running Neo4j 5.25.0.}
% Biancas text:
%All experiments were executed on a virtual machine running Debian GNU/Linux 12 (bookworm) hosted on a cluster node. The machine is equipped with an AMD EPYC 7313 16-Core Processor running at 2.00 GHz and 400 GB RAM, with Neo4j 5.25.0 also running on the same machine.
\na{We implemented \co{a} prototype \co{of} our rewriting algorithm \algoName in Java. We base our implementation on several external libraries: the OWL API {\cite{owlapi}} to parse and extend OWL ontologies, the Hermit reasoner~\cite{DBLP:journals/jar/GlimmHMSW14} to verify the JumpTrans candidates, ANTLR {\cite{antlr}} for parsing the input CN2RPQs, the Java library JgraphT {\cite{jgrapht}} for the internal representation of \nnfa{s}, and finally Cypher DSL \cite{cypherdsl} to convert the output into Cypher syntax.
}

%In this section, we present the results of an experimental evaluation of the algorithm from \Cref{sec:rewrite_cof_CN2RPQ}, that we implemented in the publicly available \emph{owl2cypher} prototype \cite{owl2cypher}.
%\emph{owl2cypher} is a Java implementation and utilises several external libraries: the OWL API {\cite{owlapi}} to parse and extend OWL ontology, the Hermit reasoner~\cite{DBLP:journals/jar/GlimmHMSW14} to verify the JumpTrans candidates, ANTLR {\cite{antlr}} for parsing the input CN2RPQ, the Java library JgraphT {\cite{jgrapht}} for the internal representation of \nnfa{s}, and finally Cypher DSL \cite{cypherdsl} to convert the output into Cypher syntax.

\paragraph{Ontology and Data.} 
To evaluate our rewriting algorithm,
we considered two real-world ontologies from different domains with divergent characteristics. In \Cref{tab:dataset_stats} we provide the number of concepts and roles for each ontology. DBPedia \cite{DBpediaDataset} is a dataset containing extracted data from Wikipedia articles, which is of interest to our rewriting algorithm due to its size. The second dataset is the open data of Montpellier Méditerranée Métropole~\cite{MMMDataset} (MMM), \co{ which covers the sewer system of the city}. MMM contains existentials on the right, triggering $\skippingAnon{\alpha_0}$ from \Cref{sec:rewrite_cof_N2RPQ}.

%\todo[inline]{NA: Say why we chose these ontologies. We should argue that our ontology/data covers relevant practical scenarios. Example: DBPedia has a large number of concepts and roles and is associated to a large dataset. MMM: existential roles on the right, which are challenging for rewrite algorithms?}
\begin{table}[t]%arxn
    \caption{Statistics of ontologies and Neo4j data}
    \label{tab:dataset_stats}
    \centering
    \begin{tabular}{ccccc}
    \toprule
        \multirow{2}{*}{\textbf{ontology}}   & \multirow{2}{*}{\textbf{\# concepts}} & \multirow{2}{*}{\textbf{\# roles}} & \multicolumn{2}{c}{\textbf{Neo4j data}}\\
        & & & \textbf{\#nodes} & \textbf{\#relations} \\
        \midrule
        DBPedia             & $\mathtt{1230}$ & $\mathtt{1212}$ & {$\mathtt{10\ 075\ 182}$} & $\mathtt{22\ 194\ 296}$  \\
        MMM\phantom{aa}                 & \0\0$\mathtt{53}$ & \0\0 $\mathtt{18}$ & \0\0\ $\mathtt{107\ 624}$ & \0\0\ $\mathtt{379\ 002}$ \\
        \bottomrule
    \end{tabular}
\end{table}

To the best of our knowledge, the current version of the OWL2 standard {\cite{owl2}} does not support data tests in role inclusions, i.e., axioms of the form $\{p\odot v \ ?\}\ISA r$; therefore, we do not consider this kind of axioms in the evaluation. 

% (OLD) As TBox we use the Cognitive Task Ontology (CogiTO) \cite{COGITO} and extended it by the axioms below (see \cite{owl2cypher} for this version of the ontology). There are multiple options to indicate that a participant is female; we have added axioms to cover all such cases, although we present only one representative example here.
% This ontology includes about \num{4686} concepts and \num{10002} axioms, whereas 720 axioms are of the form $\concept{Reading Task} \ISA \exists \role{has}.\concept{Read}\AND\exists \role{has}.\concept{Language-item}$.
% Note that CogiTO contains conjunction and qualified existentials on the left-hand side, which the prototype ignores. 
%     $\T=\big \{ \ \{\property{License}=\text{``CC0''} \ ? \} \ISA \concept{Reusable}$, $\{\property{BIDSVersion}=1.2.0 \ ? \}\ISA\concept{LatestBIDSVersion}$, $
%     \{\property{gender}= \text{``f''} \ ?\}\ISA\concept{Female}$, $\{\property{Manufacturer}=\text{``Siemens''} \}\ISA\concept{HighQuality}\big\}$

%\paragraph{Data.} 
\na{Each ontology in \Cref{tab:dataset_stats} is associated with a corresponding dataset stored in a separate Neo4j database.}
%For each of the ontologies in \Cref{tab:dataset_stats}, a corresponding dataset is available, \na{each of which} is stored in separate Neo4j databases. 
\na{We choose Neo4j because it provides a level of GQL compliance sufficient for our experimental evaluation.}
%We chose the Neo4j system as as it is currently, to the best of our knowledge, the closest available implementation to a GQL-compliant database.
\na{We relate concepts in the ontology to node labels,
% in the database
while roles are related to relationship types, i.e., edge labels.}
%We assume that concepts in the ontology correspond to node labels in the database, while roles correspond to relationship types (i.e., edge labels).

%(OLD) The dataset for the experiments (see \cite{owl2cypher}) are from the domain of cognitive neuroscience \cite{Ravenschlag2023a} and is stored in a Neo4j database, consisting of \num[round-precision=0]{396741} nodes and \num[round-precision=0]{2870405} relationships.

\paragraph{Queries.} 
Since no benchmark queries are available for rewriting join-on-free CN2RPQs, we adopt two approaches for constructing evaluation queries: (1) to evaluate the rewriting algorithm we synthetically generate 100 N2RPQs by randomly selecting concepts and roles from a given ontology, and by randomly applying operators such as concatenation, union, and the Kleene star, while enforcing structural constraints such as bounded path lengths and nesting levels; (2) for the evaluation over Neo4j
% of the rewritten queries
\CR{we wrote a set of 5 Cypher queries based on SPARQL queries present in the MMM dataset.}

\paragraph{Results.}

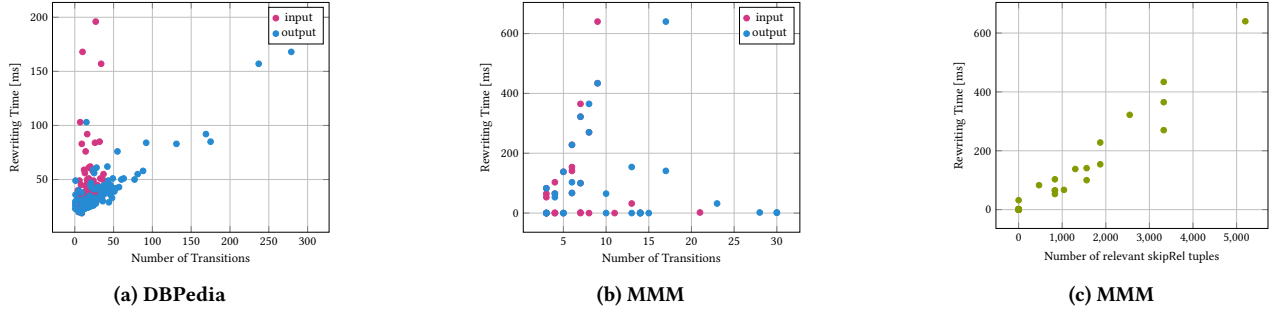
\begin{figure*}
    \centering

    \begin{subfigure}{0.3\textwidth}
        \centering
        \resizebox{!}{3.5cm}{   
\begin{tikzpicture}
    \begin{axis}[
        xlabel={Number of Transitions},
        ylabel={Rewriting Time [ms]},
        grid=major,
    ]

    \addplot[
        only marks,
        mark=*,
        solarized-magenta
    ] table[
        col sep=semicolon,
        header=true,
        x={Number of Transitions (Input)},
        y={Rewriting Time[ms]},
    ] {data/ResultsRewritingQueries_DBPedia.csv};
    \addlegendentry{input}
    \addplot[
        only marks,
        mark=*,
        solarized-blue
    ] table[
        col sep=semicolon,
        header=true,
        x={Number of Transitions (Output)},
        y={Rewriting Time[ms]},
    ] {data/ResultsRewritingQueries_DBPedia.csv};
    \addlegendentry{output}

    \end{axis}
\end{tikzpicture}
}
        \caption{DBPedia}
        \label{fig:results-dbpedia}
    \end{subfigure}
    \hfill
    \begin{subfigure}{0.3\textwidth}
        \centering
        \resizebox{!}{3.5cm}{   
\begin{tikzpicture}
    \begin{axis}[
        xlabel={Number of Transitions},
        ylabel={Rewriting Time [ms]},
        grid=major,
    ]
    \addplot[
        only marks,
        mark=*,
        solarized-magenta
    ] table[
        col sep=semicolon,
        header=true,
        x={Number of Transitions (Input)},
        y={Rewriting Time[ms]},
    ] {data/ResultsRewritingQueries_MMM.csv};
     \addlegendentry{input}
     \addplot[
        only marks,
        mark=*,
        solarized-blue
    ] table[
        col sep=semicolon,
        header=true,
        x={Number of Transitions (Output)},
        y={Rewriting Time[ms]},
    ] {data/ResultsRewritingQueries_MMM.csv};
    \addlegendentry{output}
    \end{axis}
\end{tikzpicture}
}
        \caption{MMM}
        \label{fig:results-mmm-transitions}
    \end{subfigure}
    \hfill
    \begin{subfigure}{0.3\textwidth}
        \centering
        \resizebox{!}{3.5cm}{
\begin{tikzpicture}
    \begin{axis}[
        xlabel={Number of relevant \textsf{skipRel} tuples},
        ylabel={Rewriting Time [ms]},
        grid=major,
    ]

    \addplot[
        only marks,
        mark=*,
        solarized-green
    ] table[
        col sep=semicolon,
        header=true,
        x={Number of JumpTransCandidates},
        y={Rewriting Time[ms]},
    ] {data/ResultsRewritingQueries_MMM.csv};
    %\addlegendentry{MMM}

    \end{axis}
\end{tikzpicture}
}
        \caption{MMM}
        \label{fig:results-mmm-jumptrans}
    \end{subfigure}
    \caption{Runtime of the rewriting algorithm.}
    \label{fig:resultsExperiments}
    
    \Description[A chart of our experiments.]
    {Two subfigures showing the rewriting time for DBPedia and MMM datasets.}
\end{figure*}

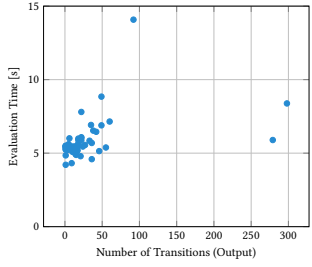
\begin{figure}
    \centering
    \resizebox{!}{3.5cm}{   
\begin{tikzpicture}
    \begin{axis}[
        xlabel={Number of Transitions (Output)},
        ylabel={Evaluation Time [s]},
        ymin=0,
        ymax=15000,
        scaled y ticks = false,
        yticklabel={\pgfmathparse{\tick/1000}\pgfmathprintnumber{\pgfmathresult}},
        grid=major,
    ]
    \addplot[
        only marks,
        mark=*,
        solarized-blue
    ] table[
        col sep=semicolon,
        header=true,
        x={Number of Transitions (Output)},
        y={Evaluation Time[ms]},
    ] {data/ResultsRewritingQueries_DBPedia.csv};
    %\addlegendentry{output}
    \end{axis}
\end{tikzpicture}
}
    \caption{Evaluation of DBPedia dataset. 56 timeouts at 20s.}
    \label{fig:DBPedia_eval_times}
    \vspace{-2mm}
\end{figure}

% \begin{figure*}
%      \centering
%     \begin{minipage}
%         {0.48\textwidth}
%         \centering
%         \input{data/Plot_TransitionsOutput-RewritingTime_DBPedia}
%     \end{minipage}
%      \begin{minipage}
%         {0.48\textwidth}
%         \centering
%         \input{data/Plot_TransitionsOutput-RewritingTime_MMM}
%     \end{minipage}

%     \caption{Results of Experimental Evaluation}
%     \label{fig:resultsExperiments}
%     \Description[A chart of our experiments.]{A scatter plot with colour-coded dots indicating how for a given dataset the experiments performed in our experiments.}
% \end{figure*}

\Cref{fig:resultsExperiments} shows the results of rewriting the generated queries. Notice in these figures we report the results as pairs: for each query, we show its transitions (seen as an N2RPE) before and after rewriting, on the x axis. Both data points share the same y-axis, indicating the rewriting time. Although the DBPedia ontology contains significantly more concepts and roles, rewriting is generally faster (under 250 ms) than with the MMM ontology (around 600 ms). This is because DBPedia does not contain any axioms with existentials on the right-hand side ($A\ISA\exists r.\top$). Therefore, there are no relevant \textsf{SkipRel} tuples. As can be seen from \Cref{fig:results-dbpedia}, the rewriting time depends on the number of transitions in the rewritten query for DBPedia. 
From the results with the MMM dataset, we can conclude that, if there are relevant \textsf{SkipRel} tuples, they seem to impact the rewriting time more than a high number of concepts (see \Cref{fig:results-mmm-transitions} and \Cref{fig:results-mmm-jumptrans}). It should be noted that for 62 of the 100 generated queries computing the power set was not feasible given the number of states of the query ($\ge 30)$. 

While the evaluation of rewriting times already shows practical feasibility,  the generated queries with Neo4j still seems to be a challenge, as suggested by the results in \Cref{tab:neo4jEvaluationMMM} and \Cref{fig:DBPedia_eval_times}.

In the DBPedia dataset, we see queries 
% falling into two groups: 
either running within 2 to 8 seconds or quickly timing out (with the timeout at 20s). We see only a weak correlation with number of transitions in the output.  Similarly, 3 of 5 queries over the MMM dataset time out. 

The reasons for the high evaluation times may be, on the one hand, query constructions that initiate the computation of the Cartesian product or, on the other hand, a high number of query answers. Note, that the generated queries do not return any answers, as could be anticipated. However, optimising the evaluation of the generated queries
% in a given database system 
was out of scope. \CR{From these results, it is clear that there is a need to understand exactly why rewritten queries time-out and how to achieve optimal rewritten queries. }

% The number of concepts, roughly indicating size and complexity of the rewriting, increases significantly from \num{3}–\num{5} concepts in the input to \num{67}–\num{124} in the rewritten queries. 
% The rewriting times for all queries are within a reasonable range of a few seconds (\num{3.3}s-\num{7.2}s) suggesting that the rewriting is practically feasible. The evaluation times for the rewritten queries, with the exception of $q_2$ which timed out at the 600s threshold, remain under \num{20}s. 
% Given the size of the ontology, these numbers demonstrate an acceptable performance for the majority of our queries, although it also reveals that certain queries can result in significantly longer evaluation times.
% A possible reason for the outlier $q_2$ is a high number of $\concept{Female}$ instances in the dataset. 
% We emphasise that the prototype is a very simple proof of concept and does not yet implement any optimisations.
% It has often been documented that queries obtained from rewriting algorithms tend to perform poorly: they introduce redundancy, nested disjunctions, and other complex subqueries that the engines are not optimised for. Thus, they require dedicated optimisations, which are often feasible given their somewhat predictable structure. We are confident that there is ample opportunity to optimise and improve the runtimes, and we plan to do so in the future.

\section{Conclusion} \label{sec:conclusion}
In this paper, we present the first practical algorithm for rewriting N2RPQs and a significant subset of CN2RPQs, \co{making use of the ability of modern graph databases to express complex navigational queries with nesting.  }
% thereby capturing a substantial portion of Cypher and GQL. 
\co{ We particularly note the ability of our approach to access key values within path expressions in the presence of an ontology with data tests.} 
% In our $\ontoLang$ ontology language, property value tests can be used to define concepts and roles in the ontology. 
The result is a highly flexible approach to ontology-mediated querying of property graphs that still allows for query rewriting into native graph database technologies.   We have demonstrated the formal correctness of our approach and provided a proof-of-concept implementation that can rewrite Cypher queries to incorporate ontological knowledge and evaluate the rewritten queries using Cypher.

To achieve a simple and practicable solution we focused on join-on-free CN2RPQs. This does not seem too limiting: even plain N2RPEs can already express many realistic examples that involve complex bidirectional navigation on the anonymous implicit facts; arbitrary conjunctions over the free variables make the query language even more powerful. We expect that real-world queries that require different regular paths to travel separately and join somewhere in the unnamed part of the canonical model will very rarely emerge in practice,  if at all.
We stress that algorithms for full extended CN2RPQs in expressive DLs have been available for over a decade \cite{Bienvenu2014}. It is not difficult to adapt these algorithms and extend our technique to be complete for all CN2RPQs and $\ontoLang$ ontologies; in fact, we do that in \co{the {extended version of this paper}}. 
However, this seems to dramatically increase the run time of the rewriting algorithm, and the added value seems limited, as we have not yet found a good use case for queries beyond our join-on-free fragment. We may explore this further.

\begin{table}[t]
    \caption{Evaluation of MMM dataset. Timeout(\VarClock) at 1000s.}
    \label{tab:neo4jEvaluationMMM}
    \centering
    \begin{tabular}{c c c c c c}
    \toprule
    \multirow{2}{*}{\textbf{query}} &  \multicolumn{2}{c}{\textbf{transitions}}&  \multirow{2}{*}{\textsf{\# skipRel}} & \multicolumn{2}{c}{\textbf{time [ms]}} \\
        & \textbf{input} & \textbf{output} & & \textbf{rew.} & \textbf{eval.} \\
    \midrule
    \textbf{Q1} & \num{6} & {6} & {52}& {153} & \num{52114} \\
    \textbf{Q2} & {6} & {9} & {468} & {137} & \VarClock \\
    \textbf{Q3} & {3} & {6} & {52} & {42} & \num{377845} \\
    \textbf{Q4} & {5} & {20} & {624} & {69} & \VarClock\\
    \textbf{Q5} & {4} & {10} & {312} & {41} & \VarClock \\
    \bottomrule
\end{tabular}
\end{table}

Other promising extensions that we plan to explore include improved support for property values and more expressive ontology languages. For the former, we will consider existing works on query rewriting with data values \cite{DBLP:conf/ecai/SavkovicC12,DBLP:conf/ijcai/BaaderBL17}. For the latter, we will build on recent efforts to identify richer ontology languages with NL data complexity \cite{arpasi2025horn,DBLP:conf/esws/LöhnertAOO25,DBLP:journals/jair/DimartinoWCP25}. 
\co{This work brings us significantly closer to
making ontology-based data-access over graph databases a reality, and as such, we aim 
% We plan 
to explore the potential of our technique for real-life OBDA systems with navigational capabilities, and to conduct experiments on systems supporting GQL as soon as they become available.} \CR{This includes the work of extending our data model to align more closely with PGs and also to explore practical methods to go beyond the join-on-free fragment.}
\co{Another interesting line of future work would be the creation of useful benchmarks that combine both complex navigational queries with rich ontologies expressed over them. }

% \todo[inline]{R3.5. Not a comment on this specific work, but an idea for future work the authors might be interested in pursuing: adding a discussion about ontology evolution, updates, or incremental rewriting, and how the approach would behave under concurrent updates to the underlying property graph.}

%%We also plan to keep moving towards full CN2RPQs. 

% }, and we look forward to exploring its advantages in \co{further} real-world use cases. 

%%
%% The acknowledgments section is defined using the "acks" environment
%% (and NOT an unnumbered section). This ensures the proper
%% identification of the section in the article metadata, and the
%% consistent spelling of the heading.
\begin{acks}
This research was funded in whole or in part by the Austrian Science Fund (FWF) 10.55776/COE12, 10.55776/PIN8884924 and 10.55776/P34962. For open access purposes, the author has applied a CC BY public copyright license to any author accepted manuscript version arising from this submission. This work was partially supported by the State of Salzburg under grant number 20102-F2101143-FPR (DNI). 
\CR{We thank the reviewers for their insightful comments and Jean Christoph Jung for pointing out an issue with the completeness of the algorithm in a preliminary version.}
\end{acks}

\section*{GenAI Usage Disclosure}
GitHub Copilot was used to help improve the prototype's README documentation and CLI. The authors remain fully responsible for all content, its implementation and interpretation, and the conclusions presented in this work.

%%
%% The next two lines define the bibliography style to be used, and
%% the bibliography file.
\bibliographystyle{ACM-Reference-Format}
\bibliography{main-cikm}

\newcommand{\etalchar}[1]{$^{#1}$}
\begin{thebibliography}{EOS{\etalchar{+}}12}

\bibitem[BCOS14]{Bienvenu2014}
Meghyn Bienvenu, Diego Calvanese, Magdalena Ortiz, and Mantas Simkus.
\newblock Nested regular path queries in description logics.
\newblock In Chitta Baral, Giuseppe~De Giacomo, and Thomas Eiter, editors, {\em Principles of Knowledge Representation and Reasoning: Proceedings of the Fourteenth International Conference, {KR} 2014, Vienna, Austria, July 20-24, 2014}. {AAAI} Press, 2014.

\bibitem[EOS{\etalchar{+}}12]{Eiter2012}
Thomas Eiter, Magdalena Ortiz, Mantas Simkus, Trung{-}Kien Tran, and Guohui Xiao.
\newblock Query rewriting for {Horn}-{SHIQ} plus rules.
\newblock In J{\"{o}}rg Hoffmann and Bart Selman, editors, {\em Proceedings of the Twenty-Sixth {AAAI} Conference on Artificial Intelligence, July 22-26, 2012, Toronto, Ontario, Canada}, pages 726--733. {AAAI} Press, 2012.

\end{thebibliography}


%%% -*-BibTeX-*-
%%% Do NOT edit. File created by BibTeX with style
%%% ACM-Reference-Format-Journals [18-Jan-2012].

\begin{thebibliography}{41}

%%% ====================================================================
%%% NOTE TO THE USER: you can override these defaults by providing
%%% customized versions of any of these macros before the \bibliography
%%% command.  Each of them MUST provide its own final punctuation,
%%% except for \shownote{} and \showURL{}.  The latter two
%%% do not use final punctuation, in order to avoid confusing it with
%%% the Web address.
%%%
%%% To suppress output of a particular field, define its macro to expand
%%% to an empty string, or better, \unskip, like this:
%%%
%%% \newcommand{\showURL}[1]{\unskip}   % LaTeX syntax
%%%
%%% \def \showURL #1{\unskip}           % plain TeX syntax
%%%
%%% ====================================================================

\ifx \showCODEN    \undefined \def \showCODEN     #1{\unskip}     \fi
\ifx \showISBNx    \undefined \def \showISBNx     #1{\unskip}     \fi
\ifx \showISBNxiii \undefined \def \showISBNxiii  #1{\unskip}     \fi
\ifx \showISSN     \undefined \def \showISSN      #1{\unskip}     \fi
\ifx \showLCCN     \undefined \def \showLCCN      #1{\unskip}     \fi
\ifx \shownote     \undefined \def \shownote      #1{#1}          \fi
\ifx \showarticletitle \undefined \def \showarticletitle #1{#1}   \fi
\ifx \showURL      \undefined \def \showURL       {\relax}        \fi
% The following commands are used for tagged output and should be
% invisible to TeX
\providecommand\bibfield[2]{#2}
\providecommand\bibinfo[2]{#2}
\providecommand\natexlab[1]{#1}
\providecommand\showeprint[2][]{arXiv:#2}

\bibitem[Angles(2018)]%
        {DBLP:conf/amw/Angles18}
\bibfield{author}{\bibinfo{person}{Renzo Angles}.} \bibinfo{year}{2018}\natexlab{}.
\newblock \showarticletitle{The Property Graph Database Model}. In \bibinfo{booktitle}{\emph{Proceedings of the 12th Alberto Mendelzon International Workshop on Foundations of Data Management, Cali, Colombia, May 21-25, 2018}} \emph{(\bibinfo{series}{{CEUR} Workshop Proceedings}, Vol.~\bibinfo{volume}{2100})}, \bibfield{editor}{\bibinfo{person}{Dan Olteanu} {and} \bibinfo{person}{Barbara Poblete}} (Eds.). \bibinfo{publisher}{CEUR-WS.org}.
\newblock
\urldef\tempurl%
\url{https://ceur-ws.org/Vol-2100/paper26.pdf}
\showURL{%
\tempurl}


\bibitem[Angles et~al\mbox{.}(2017)]%
        {Angles2017}
\bibfield{author}{\bibinfo{person}{Renzo Angles}, \bibinfo{person}{Marcelo Arenas}, \bibinfo{person}{Pablo Barcel{\'{o}}}, \bibinfo{person}{Aidan Hogan}, \bibinfo{person}{Juan~L. Reutter}, {and} \bibinfo{person}{Domagoj Vrgoc}.} \bibinfo{year}{2017}\natexlab{}.
\newblock \showarticletitle{Foundations of Modern Query Languages for Graph Databases}.
\newblock \bibinfo{journal}{\emph{{ACM} Comput. Surv.}} \bibinfo{volume}{50}, \bibinfo{number}{5} (\bibinfo{year}{2017}), \bibinfo{pages}{68:1--68:40}.
\newblock
\href{https://doi.org/10.1145/3104031}{doi:\nolinkurl{10.1145/3104031}}


\bibitem[Arp{\'a}si et~al\mbox{.}(2025)]%
        {arpasi2025horn}
\bibfield{author}{\bibinfo{person}{J{\'a}nos Arp{\'a}si}, \bibinfo{person}{Bartosz Bednarczyk}, {and} \bibinfo{person}{Magdalena Ortiz}.} \bibinfo{year}{2025}\natexlab{}.
\newblock \showarticletitle{A Horn Extension of DL-Lite with NL Data Complexity}. In \bibinfo{booktitle}{\emph{Proceedings of the International Workshop on Description Logics (DL 2025)}}.
\newblock
\newblock
\shownote{To appear}.


\bibitem[Artale et~al\mbox{.}(2012)]%
        {DBLP:conf/ecai/ArtaleRK12}
\bibfield{author}{\bibinfo{person}{Alessandro Artale}, \bibinfo{person}{Vladislav Ryzhikov}, {and} \bibinfo{person}{Roman Kontchakov}.} \bibinfo{year}{2012}\natexlab{}.
\newblock \showarticletitle{DL-Lite with Attributes and Datatypes}. In \bibinfo{booktitle}{\emph{{ECAI} 2012 - 20th European Conference on Artificial Intelligence. Including Prestigious Applications of Artificial Intelligence {(PAIS-2012)} System Demonstrations Track, Montpellier, France, August 27-31 , 2012}} \emph{(\bibinfo{series}{Frontiers in Artificial Intelligence and Applications}, Vol.~\bibinfo{volume}{242})}, \bibfield{editor}{\bibinfo{person}{Luc~De Raedt}, \bibinfo{person}{Christian Bessiere}, \bibinfo{person}{Didier Dubois}, \bibinfo{person}{Patrick Doherty}, \bibinfo{person}{Paolo Frasconi}, \bibinfo{person}{Fredrik Heintz}, {and} \bibinfo{person}{Peter J.~F. Lucas}} (Eds.). \bibinfo{publisher}{{IOS} Press}, \bibinfo{pages}{61--66}.
\newblock
\href{https://doi.org/10.3233/978-1-61499-098-7-61}{doi:\nolinkurl{10.3233/978-1-61499-098-7-61}}


\bibitem[Baader et~al\mbox{.}(2017)]%
        {DBLP:conf/ijcai/BaaderBL17}
\bibfield{author}{\bibinfo{person}{Franz Baader}, \bibinfo{person}{Stefan Borgwardt}, {and} \bibinfo{person}{Marcel Lippmann}.} \bibinfo{year}{2017}\natexlab{}.
\newblock \showarticletitle{Query Rewriting for DL-Lite with n-ary Concrete Domains}. In \bibinfo{booktitle}{\emph{Proceedings of the Twenty-Sixth International Joint Conference on Artificial Intelligence, {IJCAI} 2017, Melbourne, Australia, August 19-25, 2017}}, \bibfield{editor}{\bibinfo{person}{Carles Sierra}} (Ed.). \bibinfo{publisher}{ijcai.org}, \bibinfo{pages}{786--792}.
\newblock
\href{https://doi.org/10.24963/IJCAI.2017/109}{doi:\nolinkurl{10.24963/IJCAI.2017/109}}


\bibitem[Bate et~al\mbox{.}(2016)]%
        {DBLP:conf/kr/BateMGSH16}
\bibfield{author}{\bibinfo{person}{Andrew Bate}, \bibinfo{person}{Boris Motik}, \bibinfo{person}{Bernardo~Cuenca Grau}, \bibinfo{person}{Frantisek Simanc{\'{\i}}k}, {and} \bibinfo{person}{Ian Horrocks}.} \bibinfo{year}{2016}\natexlab{}.
\newblock \showarticletitle{Extending Consequence-Based Reasoning to {SRIQ}}. In \bibinfo{booktitle}{\emph{Principles of Knowledge Representation and Reasoning: Proceedings of the Fifteenth International Conference, {KR} 2016, Cape Town, South Africa, April 25-29, 2016}}, \bibfield{editor}{\bibinfo{person}{Chitta Baral}, \bibinfo{person}{James~P. Delgrande}, {and} \bibinfo{person}{Frank Wolter}} (Eds.). \bibinfo{publisher}{{AAAI} Press}, \bibinfo{pages}{187--196}.
\newblock
\urldef\tempurl%
\url{https://aaai.org/papers/20-12882-extending-consequence-based-reasoning-to-sriq/}
\showURL{%
\tempurl}


\bibitem[Batoul et~al\mbox{.}(2024)]%
        {MMMDataset}
\bibfield{author}{\bibinfo{person}{Haydar Batoul}, \bibinfo{person}{Nan{\'e}e Chahinian}, {and} \bibinfo{person}{Claude Pasquier}.} \bibinfo{year}{2024}\natexlab{}.
\newblock \showarticletitle{From Standards to an Ontology-Based Data Access system for Sewer Networks}. In \bibinfo{booktitle}{\emph{10{\`e}me Journ{\'e}es Doctorales en Hydrologie Urbaine}}.
\newblock


\bibitem[Bienvenu et~al\mbox{.}(2014)]%
        {Bienvenu2014}
\bibfield{author}{\bibinfo{person}{Meghyn Bienvenu}, \bibinfo{person}{Diego Calvanese}, \bibinfo{person}{Magdalena Ortiz}, {and} \bibinfo{person}{Mantas Simkus}.} \bibinfo{year}{2014}\natexlab{}.
\newblock \showarticletitle{Nested Regular Path Queries in Description Logics}. In \bibinfo{booktitle}{\emph{Principles of Knowledge Representation and Reasoning: Proceedings of the Fourteenth International Conference, {KR} 2014, Vienna, Austria, July 20-24, 2014}}, \bibfield{editor}{\bibinfo{person}{Chitta Baral}, \bibinfo{person}{Giuseppe~De Giacomo}, {and} \bibinfo{person}{Thomas Eiter}} (Eds.). \bibinfo{publisher}{{AAAI} Press}.
\newblock
\urldef\tempurl%
\url{http://www.aaai.org/ocs/index.php/KR/KR14/paper/view/8000}
\showURL{%
\tempurl}


\bibitem[Bienvenu and Ortiz(2015)]%
        {Bienvenu2015}
\bibfield{author}{\bibinfo{person}{Meghyn Bienvenu} {and} \bibinfo{person}{Magdalena Ortiz}.} \bibinfo{year}{2015}\natexlab{}.
\newblock \showarticletitle{Ontology-Mediated Query Answering with Data-Tractable Description Logics}. In \bibinfo{booktitle}{\emph{Reasoning Web. Web Logic Rules - 11th International Summer School 2015, Berlin, Germany, July 31 - August 4, 2015, Tutorial Lectures}} \emph{(\bibinfo{series}{Lecture Notes in Computer Science}, Vol.~\bibinfo{volume}{9203})}, \bibfield{editor}{\bibinfo{person}{Wolfgang Faber} {and} \bibinfo{person}{Adrian Paschke}} (Eds.). \bibinfo{publisher}{Springer}, \bibinfo{pages}{218--307}.
\newblock
\urldef\tempurl%
\url{https://doi.org/10.1007/978-3-319-21768-0_9}
\showURL{%
\tempurl}


\bibitem[Bienvenu et~al\mbox{.}(2015)]%
        {Bienvenu2015a}
\bibfield{author}{\bibinfo{person}{Meghyn Bienvenu}, \bibinfo{person}{Magdalena Ortiz}, {and} \bibinfo{person}{Mantas Simkus}.} \bibinfo{year}{2015}\natexlab{}.
\newblock \showarticletitle{Regular Path Queries in Lightweight Description Logics: Complexity and Algorithms}.
\newblock \bibinfo{journal}{\emph{J. Artif. Intell. Res.}}  \bibinfo{volume}{53} (\bibinfo{year}{2015}), \bibinfo{pages}{315--374}.
\newblock
\href{https://doi.org/10.1613/jair.4577}{doi:\nolinkurl{10.1613/jair.4577}}


\bibitem[Bonifati et~al\mbox{.}(2018)]%
        {DBLP:series/synthesis/2018Bonifati}
\bibfield{author}{\bibinfo{person}{Angela Bonifati}, \bibinfo{person}{George H.~L. Fletcher}, \bibinfo{person}{Hannes Voigt}, {and} \bibinfo{person}{Nikolay Yakovets}.} \bibinfo{year}{2018}\natexlab{}.
\newblock \bibinfo{booktitle}{\emph{Querying Graphs}}.
\newblock \bibinfo{publisher}{Morgan {\&} Claypool Publishers}.
\newblock
\showISBNx{978-3-031-00736-1}
\href{https://doi.org/10.2200/S00873ED1V01Y201808DTM051}{doi:\nolinkurl{10.2200/S00873ED1V01Y201808DTM051}}


\bibitem[Calvanese et~al\mbox{.}(2009)]%
        {DBLP:conf/ijcai/CalvaneseEO09}
\bibfield{author}{\bibinfo{person}{Diego Calvanese}, \bibinfo{person}{Thomas Eiter}, {and} \bibinfo{person}{Magdalena Ortiz}.} \bibinfo{year}{2009}\natexlab{}.
\newblock \showarticletitle{Regular Path Queries in Expressive Description Logics with Nominals}. In \bibinfo{booktitle}{\emph{{IJCAI} 2009, Proceedings of the 21st International Joint Conference on Artificial Intelligence, Pasadena, California, USA, July 11-17, 2009}}, \bibfield{editor}{\bibinfo{person}{Craig Boutilier}} (Ed.). \bibinfo{pages}{714--720}.
\newblock
\urldef\tempurl%
\url{http://ijcai.org/Proceedings/09/Papers/124.pdf}
\showURL{%
\tempurl}


\bibitem[Calvanese et~al\mbox{.}(2007)]%
        {Calvanese2007}
\bibfield{author}{\bibinfo{person}{Diego Calvanese}, \bibinfo{person}{Giuseppe~De Giacomo}, \bibinfo{person}{Domenico Lembo}, \bibinfo{person}{Maurizio Lenzerini}, {and} \bibinfo{person}{Riccardo Rosati}.} \bibinfo{year}{2007}\natexlab{}.
\newblock \showarticletitle{Tractable Reasoning and Efficient Query Answering in Description Logics: The {DL}-Lite Family}.
\newblock \bibinfo{journal}{\emph{Journal of Automated Reasoning}} \bibinfo{volume}{39}, \bibinfo{number}{3} (\bibinfo{date}{20 jul} \bibinfo{year}{2007}), \bibinfo{pages}{385--429}.
\newblock
\href{https://doi.org/10.1007/s10817-007-9078-x}{doi:\nolinkurl{10.1007/s10817-007-9078-x}}


\bibitem[Calvanese et~al\mbox{.}(1999)]%
        {Calvanese1999}
\bibfield{author}{\bibinfo{person}{Diego Calvanese}, \bibinfo{person}{Giuseppe~De Giacomo}, \bibinfo{person}{Maurizio Lenzerini}, {and} \bibinfo{person}{Moshe~Y. Vardi}.} \bibinfo{year}{1999}\natexlab{}.
\newblock \showarticletitle{Rewriting of Regular Expressions and Regular Path Queries}. In \bibinfo{booktitle}{\emph{Proceedings of the Eighteenth {ACM} {SIGACT-SIGMOD-SIGART} Symposium on Principles of Database Systems, May 31 - June 2, 1999, Philadelphia, Pennsylvania, {USA}}}, \bibfield{editor}{\bibinfo{person}{Victor Vianu} {and} \bibinfo{person}{Christos~H. Papadimitriou}} (Eds.). \bibinfo{publisher}{{ACM} Press}, \bibinfo{pages}{194--204}.
\newblock
\href{https://doi.org/10.1145/303976.303996}{doi:\nolinkurl{10.1145/303976.303996}}


\bibitem[Deutsch et~al\mbox{.}(2022)]%
        {Deutsch2022}
\bibfield{author}{\bibinfo{person}{Alin Deutsch}, \bibinfo{person}{Nadime Francis}, \bibinfo{person}{Alastair Green}, \bibinfo{person}{Keith Hare}, \bibinfo{person}{Bei Li}, \bibinfo{person}{Leonid Libkin}, \bibinfo{person}{Tobias Lindaaker}, \bibinfo{person}{Victor Marsault}, \bibinfo{person}{Wim Martens}, \bibinfo{person}{Jan Michels}, \bibinfo{person}{Filip Murlak}, \bibinfo{person}{Stefan Plantikow}, \bibinfo{person}{Petra Selmer}, \bibinfo{person}{Oskar van Rest}, \bibinfo{person}{Hannes Voigt}, \bibinfo{person}{Domagoj Vrgoc}, \bibinfo{person}{Mingxi Wu}, {and} \bibinfo{person}{Fred Zemke}.} \bibinfo{year}{2022}\natexlab{}.
\newblock \showarticletitle{Graph Pattern Matching in {GQL} and {SQL/PGQ}}. In \bibinfo{booktitle}{\emph{{SIGMOD} '22: International Conference on Management of Data, Philadelphia, PA, USA, June 12 - 17, 2022}}, \bibfield{editor}{\bibinfo{person}{Zachary~G. Ives}, \bibinfo{person}{Angela Bonifati}, {and} \bibinfo{person}{Amr~El Abbadi}} (Eds.). \bibinfo{publisher}{{ACM}}, \bibinfo{pages}{2246--2258}.
\newblock
\href{https://doi.org/10.1145/3514221.3526057}{doi:\nolinkurl{10.1145/3514221.3526057}}


\bibitem[Dimartino et~al\mbox{.}(2016)]%
        {Dimartino2016}
\bibfield{author}{\bibinfo{person}{Mirko~Michele Dimartino}, \bibinfo{person}{Andrea Cal{\`{\i}}}, \bibinfo{person}{Alexandra Poulovassilis}, {and} \bibinfo{person}{Peter~T. Wood}.} \bibinfo{year}{2016}\natexlab{}.
\newblock \showarticletitle{Query Rewriting under Linear \emph{EL} Knowledge Bases}. In \bibinfo{booktitle}{\emph{Web Reasoning and Rule Systems - 10th International Conference, {RR} 2016, Aberdeen, UK, September 9-11, 2016, Proceedings}} \emph{(\bibinfo{series}{Lecture Notes in Computer Science}, Vol.~\bibinfo{volume}{9898})}, \bibfield{editor}{\bibinfo{person}{Magdalena Ortiz} {and} \bibinfo{person}{Stefan Schlobach}} (Eds.). \bibinfo{publisher}{Springer}, \bibinfo{pages}{61--76}.
\newblock
\urldef\tempurl%
\url{https://doi.org/10.1007/978-3-319-45276-0_6}
\showURL{%
\tempurl}


\bibitem[Dimartino et~al\mbox{.}(2025)]%
        {DBLP:journals/jair/DimartinoWCP25}
\bibfield{author}{\bibinfo{person}{Mirko~Michele Dimartino}, \bibinfo{person}{Peter~T. Wood}, \bibinfo{person}{Andrea Cal{\`{\i}}}, {and} \bibinfo{person}{Alexandra Poulovassilis}.} \bibinfo{year}{2025}\natexlab{}.
\newblock \showarticletitle{Efficient Ontology-Mediated Query Answering: Extending DL-liteR and Linear {ELH}}.
\newblock \bibinfo{journal}{\emph{J. Artif. Intell. Res.}}  \bibinfo{volume}{82} (\bibinfo{year}{2025}), \bibinfo{pages}{851--899}.
\newblock
\href{https://doi.org/10.1613/JAIR.1.16401}{doi:\nolinkurl{10.1613/JAIR.1.16401}}


\bibitem[Dragovic et~al\mbox{.}(2023)]%
        {DBLP:conf/dlog/DragovicO023}
\bibfield{author}{\bibinfo{person}{Nikola Dragovic}, \bibinfo{person}{Cem Okulmus}, {and} \bibinfo{person}{Magdalena Ortiz}.} \bibinfo{year}{2023}\natexlab{}.
\newblock \showarticletitle{Rewriting Ontology-Mediated Navigational Queries into Cypher}. In \bibinfo{booktitle}{\emph{Proceedings of the 36th International Workshop on Description Logics {(DL} 2023) co-located with the 20th International Conference on Principles of Knowledge Representation and Reasoning and the 21st International Workshop on Non-Monotonic Reasoning {(KR} 2023 and {NMR} 2023)., Rhodes, Greece, September 2-4, 2023}} \emph{(\bibinfo{series}{{CEUR} Workshop Proceedings}, Vol.~\bibinfo{volume}{3515})}, \bibfield{editor}{\bibinfo{person}{Oliver Kutz}, \bibinfo{person}{Carsten Lutz}, {and} \bibinfo{person}{Ana Ozaki}} (Eds.). \bibinfo{publisher}{CEUR-WS.org}.
\newblock
\urldef\tempurl%
\url{https://ceur-ws.org/Vol-3515/paper-9.pdf}
\showURL{%
\tempurl}


\bibitem[Eiter et~al\mbox{.}(2012)]%
        {Eiter2012}
\bibfield{author}{\bibinfo{person}{Thomas Eiter}, \bibinfo{person}{Magdalena Ortiz}, \bibinfo{person}{Mantas Simkus}, \bibinfo{person}{Trung{-}Kien Tran}, {and} \bibinfo{person}{Guohui Xiao}.} \bibinfo{year}{2012}\natexlab{}.
\newblock \showarticletitle{Query Rewriting for {Horn}-{SHIQ} Plus Rules}. In \bibinfo{booktitle}{\emph{Proceedings of the Twenty-Sixth {AAAI} Conference on Artificial Intelligence, July 22-26, 2012, Toronto, Ontario, Canada}}, \bibfield{editor}{\bibinfo{person}{J{\"{o}}rg Hoffmann} {and} \bibinfo{person}{Bart Selman}} (Eds.). \bibinfo{publisher}{{AAAI} Press}, \bibinfo{pages}{726--733}.
\newblock
\urldef\tempurl%
\url{http://www.aaai.org/ocs/index.php/AAAI/AAAI12/paper/view/4931}
\showURL{%
\tempurl}


\bibitem[Francis et~al\mbox{.}(2022)]%
        {Francis2022}
\bibfield{author}{\bibinfo{person}{Nadime Francis}, \bibinfo{person}{Am{\'{e}}lie Gheerbrant}, \bibinfo{person}{Paolo Guagliardo}, \bibinfo{person}{Leonid Libkin}, \bibinfo{person}{Victor Marsault}, \bibinfo{person}{Wim Martens}, \bibinfo{person}{Filip Murlak}, \bibinfo{person}{Liat Peterfreund}, \bibinfo{person}{Alexandra Rogova}, {and} \bibinfo{person}{Domagoj Vrgoc}.} \bibinfo{year}{2022}\natexlab{}.
\newblock \showarticletitle{{GPC:} {A} Pattern Calculus for Property Graphs}.
\newblock \bibinfo{journal}{\emph{CoRR}}  \bibinfo{volume}{abs/2210.16580} (\bibinfo{year}{2022}).
\newblock
\showeprint[arxiv]{2210.16580}
\href{https://doi.org/10.48550/arXiv.2210.16580}{doi:\nolinkurl{10.48550/arXiv.2210.16580}}


\bibitem[Francis et~al\mbox{.}(2023)]%
        {Francis2023}
\bibfield{author}{\bibinfo{person}{Nadime Francis}, \bibinfo{person}{Am{\'{e}}lie Gheerbrant}, \bibinfo{person}{Paolo Guagliardo}, \bibinfo{person}{Leonid Libkin}, \bibinfo{person}{Victor Marsault}, \bibinfo{person}{Wim Martens}, \bibinfo{person}{Filip Murlak}, \bibinfo{person}{Liat Peterfreund}, \bibinfo{person}{Alexandra Rogova}, {and} \bibinfo{person}{Domagoj Vrgoc}.} \bibinfo{year}{2023}\natexlab{}.
\newblock \showarticletitle{A Researcher's Digest of {GQL} (Invited Talk)}. In \bibinfo{booktitle}{\emph{26th International Conference on Database Theory, {ICDT} 2023, March 28-31, 2023, Ioannina, Greece}} \emph{(\bibinfo{series}{LIPIcs}, Vol.~\bibinfo{volume}{255})}, \bibfield{editor}{\bibinfo{person}{Floris Geerts} {and} \bibinfo{person}{Brecht Vandevoort}} (Eds.). \bibinfo{publisher}{Schloss Dagstuhl - Leibniz-Zentrum f{\"{u}}r Informatik}, \bibinfo{pages}{1:1--1:22}.
\newblock
\href{https://doi.org/10.4230/LIPIcs.ICDT.2023.1}{doi:\nolinkurl{10.4230/LIPIcs.ICDT.2023.1}}


\bibitem[Francis et~al\mbox{.}(2018)]%
        {DBLP:conf/sigmod/FrancisGGLLMPRS18}
\bibfield{author}{\bibinfo{person}{Nadime Francis}, \bibinfo{person}{Alastair Green}, \bibinfo{person}{Paolo Guagliardo}, \bibinfo{person}{Leonid Libkin}, \bibinfo{person}{Tobias Lindaaker}, \bibinfo{person}{Victor Marsault}, \bibinfo{person}{Stefan Plantikow}, \bibinfo{person}{Mats Rydberg}, \bibinfo{person}{Petra Selmer}, {and} \bibinfo{person}{Andr{\'{e}}s Taylor}.} \bibinfo{year}{2018}\natexlab{}.
\newblock \showarticletitle{Cypher: An Evolving Query Language for Property Graphs}. In \bibinfo{booktitle}{\emph{Proceedings of the 2018 International Conference on Management of Data, {SIGMOD} Conference 2018, Houston, TX, USA, June 10-15, 2018}}, \bibfield{editor}{\bibinfo{person}{Gautam Das}, \bibinfo{person}{Christopher~M. Jermaine}, {and} \bibinfo{person}{Philip~A. Bernstein}} (Eds.). \bibinfo{publisher}{{ACM}}, \bibinfo{pages}{1433--1445}.
\newblock
\href{https://doi.org/10.1145/3183713.3190657}{doi:\nolinkurl{10.1145/3183713.3190657}}


\bibitem[Glimm et~al\mbox{.}(2014)]%
        {DBLP:journals/jar/GlimmHMSW14}
\bibfield{author}{\bibinfo{person}{Birte Glimm}, \bibinfo{person}{Ian Horrocks}, \bibinfo{person}{Boris Motik}, \bibinfo{person}{Giorgos Stoilos}, {and} \bibinfo{person}{Zhe Wang}.} \bibinfo{year}{2014}\natexlab{}.
\newblock \showarticletitle{HermiT: An {OWL} 2 Reasoner}.
\newblock \bibinfo{journal}{\emph{J. Autom. Reason.}} \bibinfo{volume}{53}, \bibinfo{number}{3} (\bibinfo{year}{2014}), \bibinfo{pages}{245--269}.
\newblock
\href{https://doi.org/10.1007/S10817-014-9305-1}{doi:\nolinkurl{10.1007/S10817-014-9305-1}}


\bibitem[Horridge et~al\mbox{.}(2007)]%
        {owlapi}
\bibfield{author}{\bibinfo{person}{Matthew Horridge}, \bibinfo{person}{Sean Bechhofer}, {and} \bibinfo{person}{Olaf Noppens}.} \bibinfo{year}{2007}\natexlab{}.
\newblock \bibinfo{title}{{The OWL API}}.
\newblock
\urldef\tempurl%
\url{https://owlcs.github.io/owlapi/}
\showURL{%
\tempurl}
\newblock
\shownote{Accessed: 2026-05-21}.


\bibitem[Kazakov(2009)]%
        {DBLP:conf/ijcai/Kazakov09}
\bibfield{author}{\bibinfo{person}{Yevgeny Kazakov}.} \bibinfo{year}{2009}\natexlab{}.
\newblock \showarticletitle{Consequence-Driven Reasoning for Horn {SHIQ} Ontologies}. In \bibinfo{booktitle}{\emph{{IJCAI} 2009, Proceedings of the 21st International Joint Conference on Artificial Intelligence, Pasadena, California, USA, July 11-17, 2009}}, \bibfield{editor}{\bibinfo{person}{Craig Boutilier}} (Ed.). \bibinfo{pages}{2040--2045}.
\newblock
\urldef\tempurl%
\url{http://ijcai.org/Proceedings/09/Papers/336.pdf}
\showURL{%
\tempurl}


\bibitem[Kr{\"{o}}tzsch et~al\mbox{.}(2013)]%
        {KrotzschRH13}
\bibfield{author}{\bibinfo{person}{Markus Kr{\"{o}}tzsch}, \bibinfo{person}{Sebastian Rudolph}, {and} \bibinfo{person}{Pascal Hitzler}.} \bibinfo{year}{2013}\natexlab{}.
\newblock \showarticletitle{Complexities of Horn Description Logics}.
\newblock \bibinfo{journal}{\emph{{ACM} Trans. Comput. Log.}} \bibinfo{volume}{14}, \bibinfo{number}{1} (\bibinfo{year}{2013}), \bibinfo{pages}{2:1--2:36}.
\newblock
\href{https://doi.org/10.1145/2422085.2422087}{doi:\nolinkurl{10.1145/2422085.2422087}}


\bibitem[L{\"{o}}hnert et~al\mbox{.}(2025)]%
        {DBLP:conf/dlog/LohnertAOO25}
\bibfield{author}{\bibinfo{person}{Bianca L{\"{o}}hnert}, \bibinfo{person}{Nikolaus Augsten}, \bibinfo{person}{Cem Okulmus}, {and} \bibinfo{person}{Magdalena Ortiz}.} \bibinfo{year}{2025}\natexlab{}.
\newblock \showarticletitle{Query Rewriting for Nested Navigational Queries over Property Graphs}. In \bibinfo{booktitle}{\emph{Proceedings of the 38th International Workshop on Description Logics - {DL} 2025, Opole, Poland, September 3-6, 2025}} \emph{(\bibinfo{series}{{CEUR} Workshop Proceedings})}, \bibfield{editor}{\bibinfo{person}{Lidia Tendera}, \bibinfo{person}{Yazm{\'{\i}}n Ib{\'{a}}{\~{n}}ez{-}Garc{\'{\i}}a}, {and} \bibinfo{person}{Patrick Koopmann}} (Eds.). \bibinfo{publisher}{CEUR-WS.org}.
\newblock
\urldef\tempurl%
\url{https://ceur-ws.org/Vol-4091/paper40.pdf}
\showURL{%
\tempurl}


\bibitem[Löhnert et~al\mbox{.}(2025)]%
        {DBLP:conf/esws/LöhnertAOO25}
\bibfield{author}{\bibinfo{person}{Bianca Löhnert}, \bibinfo{person}{Nikolaus Augsten}, \bibinfo{person}{Cem Okulmus}, {and} \bibinfo{person}{Magdalena Ortiz}.} \bibinfo{year}{2025}\natexlab{}.
\newblock \showarticletitle{Towards Practicable Algorithms for Rewriting Graph Queries beyond {DL-Lite}}. In \bibinfo{booktitle}{\emph{The Semantic Web - 22nt International Conference, {ESWC} 2025, Portorož, Slovenia, June 1-5, 2025 (accepted for publication)}} \emph{(\bibinfo{series}{Lecture Notes in Computer Science})}. \bibinfo{publisher}{Springer}.
\newblock


\bibitem[Löhnert et~al\mbox{.}(2026)]%
        {arxivVersion}
\bibfield{author}{\bibinfo{person}{Bianca Löhnert}, \bibinfo{person}{Nikolaus Augsten}, \bibinfo{person}{Cem Okulmus}, {and} \bibinfo{person}{Magdalena Ortiz}.} \bibinfo{year}{2026}\natexlab{}.
\newblock \showarticletitle{Rewriting Ontology-Mediated Property Graph Queries into GQL (extended version)}.
\newblock \bibinfo{journal}{\emph{CoRR}}  \bibinfo{volume}{abs/2608.20092} (\bibinfo{year}{2026}).
\newblock
\showeprint[arXiv]{2608.20092}
\href{https://doi.org/10.48550/arXiv.2608.20092}{doi:\nolinkurl{10.48550/arXiv.2608.20092}}


\bibitem[Mendes et~al\mbox{.}(2012)]%
        {DBpediaDataset}
\bibfield{author}{\bibinfo{person}{Pablo~N Mendes}, \bibinfo{person}{Max Jakob}, {and} \bibinfo{person}{Christian Bizer}.} \bibinfo{year}{2012}\natexlab{}.
\newblock \bibinfo{booktitle}{\emph{DBpedia: A multilingual cross-domain knowledge base}}.
\newblock \bibinfo{publisher}{European Language Resources Association (ELRA)}.
\newblock


\bibitem[Michail et~al\mbox{.}(2020)]%
        {jgrapht}
\bibfield{author}{\bibinfo{person}{Dimitrios Michail}, \bibinfo{person}{Joris Kinable}, \bibinfo{person}{Barak Naveh}, {and} \bibinfo{person}{John~V. Sichi}.} \bibinfo{year}{2020}\natexlab{}.
\newblock \bibinfo{title}{{JGraphT -- A Java Library for Graph Data Structures and Algorithms}}.
\newblock
\urldef\tempurl%
\url{https://jgrapht.org/}
\showURL{%
\tempurl}
\newblock
\shownote{Accessed: 2026-05-21}.


\bibitem[{Neo4j Contributors}({[n.\,d.]})]%
        {neo4j-gql-conformance}
\bibfield{author}{\bibinfo{person}{{Neo4j Contributors}}.} \bibinfo{year}{[n.\,d.]}\natexlab{}.
\newblock \bibinfo{title}{GQL Conformance - Cypher Manual}.
\newblock
\urldef\tempurl%
\url{https://neo4j.com/docs/cypher-manual/current/appendix/gql-conformance/}
\showURL{%
\tempurl}
\newblock
\shownote{Accessed: {2026-05-23}}.


\bibitem[Ortiz et~al\mbox{.}(2011)]%
        {Ortiz2011}
\bibfield{author}{\bibinfo{person}{Magdalena Ortiz}, \bibinfo{person}{Sebastian Rudolph}, {and} \bibinfo{person}{Mantas Simkus}.} \bibinfo{year}{2011}\natexlab{}.
\newblock \showarticletitle{Query Answering in the Horn Fragments of the Description Logics {SHOIQ} and {SROIQ}}. In \bibinfo{booktitle}{\emph{{IJCAI} 2011, Proceedings of the 22nd International Joint Conference on Artificial Intelligence, Barcelona, Catalonia, Spain, July 16-22, 2011}}, \bibfield{editor}{\bibinfo{person}{Toby Walsh}} (Ed.). \bibinfo{publisher}{{IJCAI/AAAI}}, \bibinfo{pages}{1039--1044}.
\newblock
\href{https://doi.org/10.5591/978-1-57735-516-8/IJCAI11-178}{doi:\nolinkurl{10.5591/978-1-57735-516-8/IJCAI11-178}}


\bibitem[Ostropolski{-}Nalewaja and Rudolph(2024)]%
        {DBLP:conf/kr/Ostropolski-Nalewaja24}
\bibfield{author}{\bibinfo{person}{Piotr Ostropolski{-}Nalewaja} {and} \bibinfo{person}{Sebastian Rudolph}.} \bibinfo{year}{2024}\natexlab{}.
\newblock \showarticletitle{The Sticky Path to Expressive Querying: Decidability of Navigational Queries under Existential Rules}. In \bibinfo{booktitle}{\emph{Proceedings of the 21st International Conference on Principles of Knowledge Representation and Reasoning, {KR} 2024, Hanoi, Vietnam. November 2-8, 2024}}, \bibfield{editor}{\bibinfo{person}{Pierre Marquis}, \bibinfo{person}{Magdalena Ortiz}, {and} \bibinfo{person}{Maurice Pagnucco}} (Eds.).
\newblock
\href{https://doi.org/10.24963/KR.2024/54}{doi:\nolinkurl{10.24963/KR.2024/54}}


\bibitem[Parr(2012)]%
        {antlr}
\bibfield{author}{\bibinfo{person}{Terence Parr}.} \bibinfo{year}{2012}\natexlab{}.
\newblock \bibinfo{title}{{ANTLR (ANother Tool for Language Recognition)}}.
\newblock
\urldef\tempurl%
\url{https://www.antlr.org/}
\showURL{%
\tempurl}
\newblock
\shownote{Accessed: 2026-05-21}.


\bibitem[Savkovic and Calvanese(2012)]%
        {DBLP:conf/ecai/SavkovicC12}
\bibfield{author}{\bibinfo{person}{Ognjen Savkovic} {and} \bibinfo{person}{Diego Calvanese}.} \bibinfo{year}{2012}\natexlab{}.
\newblock \showarticletitle{Introducing Datatypes in DL-Lite}. In \bibinfo{booktitle}{\emph{{ECAI} 2012 - 20th European Conference on Artificial Intelligence. Including Prestigious Applications of Artificial Intelligence {(PAIS-2012)} System Demonstrations Track, Montpellier, France, August 27-31 , 2012}} \emph{(\bibinfo{series}{Frontiers in Artificial Intelligence and Applications})}. \bibinfo{publisher}{{IOS} Press}, \bibinfo{pages}{720--725}.
\newblock
\href{https://doi.org/10.3233/978-1-61499-098-7-720}{doi:\nolinkurl{10.3233/978-1-61499-098-7-720}}


\bibitem[Simons and {Neo4j Contributors}(2020)]%
        {cypherdsl}
\bibfield{author}{\bibinfo{person}{Michael~J. Simons} {and} \bibinfo{person}{{Neo4j Contributors}}.} \bibinfo{year}{2020}\natexlab{}.
\newblock \bibinfo{title}{{The Neo4j Cypher-DSL}}.
\newblock
\urldef\tempurl%
\url{https://github.com/neo4j/cypher-dsl}
\showURL{%
\tempurl}
\newblock
\shownote{Accessed: 2026-05-21}.


\bibitem[Stefanoni et~al\mbox{.}(2014)]%
        {DBLP:journals/jair/StefanoniMKR14}
\bibfield{author}{\bibinfo{person}{Giorgio Stefanoni}, \bibinfo{person}{Boris Motik}, \bibinfo{person}{Markus Kr{\"{o}}tzsch}, {and} \bibinfo{person}{Sebastian Rudolph}.} \bibinfo{year}{2014}\natexlab{}.
\newblock \showarticletitle{The Complexity of Answering Conjunctive and Navigational Queries over {OWL} 2 {EL} Knowledge Bases}.
\newblock \bibinfo{journal}{\emph{J. Artif. Intell. Res.}}  \bibinfo{volume}{51} (\bibinfo{year}{2014}), \bibinfo{pages}{645--705}.
\newblock
\href{https://doi.org/10.1613/JAIR.4457}{doi:\nolinkurl{10.1613/JAIR.4457}}


\bibitem[Steigmiller et~al\mbox{.}(2014)]%
        {DBLP:journals/ws/SteigmillerLG14}
\bibfield{author}{\bibinfo{person}{Andreas Steigmiller}, \bibinfo{person}{Thorsten Liebig}, {and} \bibinfo{person}{Birte Glimm}.} \bibinfo{year}{2014}\natexlab{}.
\newblock \showarticletitle{Konclude: System description}.
\newblock \bibinfo{journal}{\emph{J. Web Semant.}}  \bibinfo{volume}{27-28} (\bibinfo{year}{2014}), \bibinfo{pages}{78--85}.
\newblock
\href{https://doi.org/10.1016/J.WEBSEM.2014.06.003}{doi:\nolinkurl{10.1016/J.WEBSEM.2014.06.003}}


\bibitem[{W3C OWL Working Group}(2012)]%
        {owl2}
\bibfield{author}{\bibinfo{person}{{W3C OWL Working Group}}.} \bibinfo{year}{2012}\natexlab{}.
\newblock \bibinfo{title}{{OWL 2 Web Ontology Language Document Overview (Second Edition)}}.
\newblock
\urldef\tempurl%
\url{https://www.w3.org/TR/owl2-overview/}
\showURL{%
\tempurl}
\newblock
\shownote{W3C Recommendation. Accessed: 2026-05-21}.


\bibitem[Xiao et~al\mbox{.}(2018)]%
        {DBLP:conf/ijcai/XiaoCKLPRZ18}
\bibfield{author}{\bibinfo{person}{Guohui Xiao}, \bibinfo{person}{Diego Calvanese}, \bibinfo{person}{Roman Kontchakov}, \bibinfo{person}{Domenico Lembo}, \bibinfo{person}{Antonella Poggi}, \bibinfo{person}{Riccardo Rosati}, {and} \bibinfo{person}{Michael Zakharyaschev}.} \bibinfo{year}{2018}\natexlab{}.
\newblock \showarticletitle{Ontology-Based Data Access: {A} Survey}. In \bibinfo{booktitle}{\emph{Proceedings of the Twenty-Seventh International Joint Conference on Artificial Intelligence, {IJCAI} 2018, July 13-19, 2018, Stockholm, Sweden}}, \bibfield{editor}{\bibinfo{person}{J{\'{e}}r{\^{o}}me Lang}} (Ed.). \bibinfo{publisher}{ijcai.org}, \bibinfo{pages}{5511--5519}.
\newblock
\href{https://doi.org/10.24963/IJCAI.2018/777}{doi:\nolinkurl{10.24963/IJCAI.2018/777}}


\end{thebibliography}

%%
%% If your work has an appendix, this is the place to put it.
% \appendix

\end{document}